\documentclass[american,twosided,12pt]{article}
\usepackage{etex}
\usepackage{setspace}
\usepackage{tikz}
\usepackage{epsfig}
\usepackage{amsmath}
\usepackage{amsfonts}
\usepackage{hyperref}
\usepackage{amssymb}
\usepackage{pgffor}
\usepackage[margin=1.5in]{geometry}
\usepackage{babel}
\usepackage{subfig}
\usepackage{booktabs}
\usetikzlibrary {positioning}
\usepackage{rotating}
\usetikzlibrary{matrix}
\usepackage{diagbox}
\usepackage{tabularx}
\usetikzlibrary{decorations.pathmorphing}
\usepackage[authoryear]{natbib}
\usepackage{mathrsfs}
\usepackage{comment}
\usepackage{enumitem}
\usepackage{pgfplots}
\usepackage{multirow}

\usepackage{graphicx,kantlipsum,setspace}
\usepackage{caption}
\usepackage[labelfont=bf]{caption}

\newenvironment{absolutelynopagebreak}
  {\par\nobreak\vfil\penalty0\vfilneg
   \vtop\bgroup}
  {\par\xdef\tpd{\the\prevdepth}\egroup
   \prevdepth=\tpd}

\newcommand*{\field}[1]{\mathbb{#1}}

\definecolor{mondrianRed}{rgb}{0.8666,0.1334,0}
\definecolor{mondrianBlue}{rgb}{0.133333, 0.313725, 0.584314}
\definecolor{mondrianYellow}{rgb}{0.980392, 0.788235, 0.00392157}
\definecolor{mondrianGrey}{rgb}{0.756863, 0.784314, 0.788235}
\definecolor{mondrianOrange}{rgb}{0.923529, 0.396078, 0.00196078}
\definecolor{mondrianGreen}{rgb}{0.556863, 0.54902, 0.294118}
\definecolor{mondrianCyan}{rgb}{0.345098, 0.431373, 0.439216}
\definecolor{mondrianWhite}{rgb}{0.976471, 0.976471, 0.976471}
\definecolor{mondrianPurple}{rgb}{0.107843, 0.158824, 0.292157}

\definecolor {processblue}{cmyk}{0.96,0,0,0}
\usetikzlibrary{arrows}
\input{glyphtounicode}
\DeclareMathAlphabet{\mathpzc}{OT1}{pzc}{m}{it}

\makeatletter
\renewcommand{\section}{\@startsection{section}{1}{0mm}{-1.5\baselineskip}{0.8\baselineskip}{\normalfont\large\centering}}
\renewcommand{\subsection}{\@startsection{subsection}{2}{0mm}{-0.1\baselineskip}{0.5\baselineskip}{\normalfont\bf\flushleft}}
\renewcommand{\@seccntformat}[1]{\csname the#1\endcsname \hspace{+0mm}\large{.}\hspace{+1.9mm}}
\renewcommand{\@seccntformat}[2]{\csname the#1\endcsname \hspace{+0mm}\large{.}\hspace{+1.9mm}}
\makeatother

\newtheorem{corollary}{Corollary}

\newtheorem{definition}{Definition}

\newtheorem{lemma}{Lemma}

\newtheorem{proposition}{Proposition}

\newtheorem{result}{Result}

\newenvironment{support}[1][\textit{Support}]{\textbf{#1:} }{\,\,\,\,\rule{0.5em}{0.5em}}

\renewcommand{\theequation}{\arabic{equation}}

\newlength{\extraspace}
\newlength{\extraspaces}
\newcounter{dummy}

\newcommand{\baa}{
\addtocounter{equation}{1} \setcounter{dummy}{\value{equation}}
\setcounter{equation}{0}
\renewcommand{\theequation}{\arabic{dummy}\alph{equation}}
\begin{eqnarray}
\addtolength{\abovedisplayskip}{\extraspaces}
\addtolength{\belowdisplayskip}{\extraspaces}
\addtolength{\abovedisplayshortskip}{\extraspace}
\addtolength{\belowdisplayshortskip}{\extraspace}}

\newcommand{\eaa}{
\end{eqnarray}
\setcounter{equation}{\value{dummy}}
\renewcommand{\theequation}{\arabic{equation}}}

\newcommand{\be}{\begin{equation}
\addtolength{\abovedisplayskip}{\extraspaces}
\addtolength{\belowdisplayskip}{\extraspaces}
\addtolength{\abovedisplayshortskip}{\extraspace}
\addtolength{\belowdisplayshortskip}{\extraspace}}
\newcommand{\ee}{\end{equation}}

\newcommand{\ba}{\begin{eqnarray}
\addtolength{\abovedisplayskip}{\extraspaces}
\addtolength{\belowdisplayskip}{\extraspaces}
\addtolength{\abovedisplayshortskip}{\extraspace}
\addtolength{\belowdisplayshortskip}{\extraspace}}
\newcommand{\ea}{\end{eqnarray}}

\newcommand{\bd}{\begin{displaymath}
\addtolength{\abovedisplayskip}{\extraspaces}
\addtolength{\belowdisplayskip}{\extraspaces}
\addtolength{\abovedisplayshortskip}{\extraspace}
\addtolength{\belowdisplayshortskip}{\extraspace}}
\newcommand{\ed}{\end{displaymath}}

\newcommand{\deel}[2]{{\textstyle{#1 \over #2}}}
\newcommand{\hf}{{\textstyle{1\over 2}}}

\def\inbar{\,\vrule height1.5ex width.4pt depth0pt}
\font\rms=cmr12 at 12pt
\def\ce{\relax\ifmmode\mathchoice
{\hbox{$\inbar\kern-.3em{\rm C}$}} {\hbox{$\inbar\kern-.3em{\rm
C}$}} {\lower.9pt\hbox{\rms $\inbar\kern-.3em{\rm C}$}}
{\lower1.2pt\hbox{\rms $\inbar\kern-.3em{\rm C}$}}
\else{$\inbar\kern-.3em{\rm C}$}\fi}
\font\cmss=cmss12 \font\cmsss=cmss12 at 12pt
\def\ze{\relax\ifmmode\mathchoice
{\hbox{\cmss Z\kern-.4em Z}}{\hbox{\cmss Z\kern-.4em Z}}
{\lower.9pt\hbox{\cmsss Z\kern-.4em Z}} {\lower1.2pt\hbox{\cmsss
Z\kern-.4em Z}}\else{\cmss Z\kern-.4em Z}\fi}

\def\Re{{\mathbb{R}}}

\newcommand{\br}{B\hspace*{-0.6pt}R}

\newcommand{\refsection}[1]{
\vspace{1mm} \pagebreak[3] \addtocounter{section}{1}
\begin{center}
{\large #1}
\end{center}
\nopagebreak
\medskip
\nopagebreak}

\def\thebibliography#1{\refsection{\bf References}
\vspace*{-4mm}\list
 {\relax}{\itemsep=1pt \parsep=0pt
 \usecounter{enumiv}\leftmargin=3em\itemindent=-\leftmargin}%
 \def\newblock{\hskip .11em plus .33em minus .07em}
 \sloppy\clubpenalty4000\widowpenalty4000
 \sfcode`\.=1000\relax}

\newcommand{\startappendix}{
\renewcommand{\thesection}{\Alph{section}}
\setcounter{section}{0}
\renewcommand{\theequation}{\thesection.\arabic{equation}}}

\newcommand{\newappendix}[1]{
\vspace{3mm}
\pagebreak[3]
\addtocounter{section}{1}
\setcounter{equation}{0}
\setcounter{subsection}{0}
\begin{center}
{\large Appendix \thesection. #1}
\vspace{0mm}
\end{center}
\nopagebreak
\vspace{-1mm}
\nopagebreak}

\newcommand{\q}[1]{``#1''}

\begin{document}
\newcolumntype{L}[1]{>{\raggedright\arraybackslash}p{#1}}
\newcolumntype{C}[1]{>{\centering\arraybackslash}p{#1}}
\newcolumntype{R}[1]{>{\raggedleft\arraybackslash}p{#1}}

\setcounter{page}{0}
\thispagestyle{empty}

\begin{center}
{\Huge\bf $\mbox{\bf\em M}$ Equilibrium}\\[3mm]
{\Large\sc A Theory of Beliefs and Choices in Games}\\[10mm]
{\large Jacob K. Goeree and Philippos Louis}\footnote{Goeree: AGORA Center for Market Design, UNSW. Louis: Department of Economics, University of Cyprus and AGORA Fellow. We gratefully acknowledge funding from the Australian Research Council (DP190103888). We thank the Co-Editor Jeffrey Ely and three anonymous reviewers for their extremely constructive and insightful suggestions and colleagues as well as seminar participants for their helpful feedback. Special thanks to Vai-Lam Mui for stressing the importance of alternatives to fixed-point models and to Lucas Pahl who taught us much of what we know about the geometry of optimization. The ``M'' terminology was inspired by \url{https://en.wikipedia.org/wiki/M-theory}.}\\[5mm]
\today\\[10mm]
{\bf Abstract}
\end{center}
\addtolength{\baselineskip}{-1.2mm}
\vspace*{-3mm}

\noindent We introduce a set-valued solution concept, $M$ equilibrium, to capture empirical regularities from over half a century of game-theory experiments. We show $M$ equilibrium serves as a meta theory for various models that hitherto were considered unrelated. $M$ equilibrium is empirically robust and, despite being set-valued, falsifiable. We report results from a series of experiments comparing $M$ equilibrium to leading behavioral-game-theory models and demonstrate its virtues in predicting observed choices and stated beliefs. Data from experimental games with a unique pure-strategy Nash equilibrium and multiple $M$ equilibria exhibit coordination problems that could not be anticipated through the lens of existing models.

\vfill

\addtocounter{footnote}{-1}

\newpage

\addtolength{\baselineskip}{0.8mm}

\section{Introduction}

The \cite{Nash1950} equilibrium has been the centerpiece of game theory since its introduction more than seventy years ago. Its simple formulation and broad applicability have helped game theory gain the central role originally envisioned by \cite{vonNeumannMoregenstern1944}. The Nash equilibrium is one of the most commonly used constructs in economics, with applications in other social sciences, evolutionary biology, computer science, and philosophy.

However, the Nash equilibrium fares poorly when confronted with data because of its strong assumptions and stark predictions. Players are subsumed to have \textit{perfect foresight}, i.e. they anticipate others' choices, and \textit{optimize perfectly}, i.e. they best respond to their beliefs. Randomization occurs only when players are indifferent. As a result, Nash equilibrium implies a high-degree of \textit{homogeneity}, e.g. in symmetric games with a unique equilibrium, players' choices are identical (and if the equilibrium is in pure strategies, any deviation causes a zero-likelihood problem). Even in games where players' choices are predicted to differ, their beliefs are identical. A survey of over half a century of experimental game theory results reveals that neither the assumptions underlying Nash equilibrium nor its predictions are validated empirically.

\subsection{Empirical Motivation for $M$ equilibrium}
\label{empMotivation}

There is overwhelming evidence from the laboratory that subjects do \textit{not} play Nash equilibrium. Their choices are often far away from Nash predictions.\footnote{See e.g. \citet{lieberman1960, lieberman1962, brayer1964, messick1967, oneill1987, brown1990, rapoport1992, stahl1994, stahl1995, nagel1995, mckelvey1992, McKelveyPalfrey1995, goeree2001,crawford2013,GoereeHoltPalfrey2016,dhami2016,eyster2019}. This list is far from exhaustive.} Moreover, observed choice behavior is \textit{heterogeneous}, both within and across individuals. By `within individuals' we mean that a subject may play the same game differently on two separate occasions, a phenomenon labeled \q{stochastic choice} in decision theory.\footnote{See e.g. \citet{tversky1969,hey1994,ballinger1997, hey2001,agranov2017,friedman2019}.} By `across individuals' we mean that subjects produce different choice distributions when playing the same role in the same game.\footnote{See \citet{brown1990, rapoport1992, stahl1994, stahl1995, nagel1995, mckelvey2000,costa2001cognition, arad2012, goeree2017,dellavigna2018} for examples and discussion.}

\cite{Nash1950} defined equilibrium solely in terms of choices, i.e. as a fixed-point of the best reply mappings, implicitly assuming that players' beliefs are correct. One possibility why Nash did not model beliefs explicitly is that he considered them ephemeral. However, by adapting belief-elicitation methods from decision theory, recent work in experimental game theory has produced a wealth of data on the relationship between beliefs and choices.\footnote{See \citet{schotter2014} and \citet{schlag2015} for excellent surveys. For tests of the state-of-the-art methods see \citet{holt2016} and \citet{danzVesterlundWilson2020}.}

If the empirical support for Nash-equilibrium choices is poor, it is worse for Nash-equilibrium beliefs.  A main finding from experiments that elicit beliefs and choices is that a significant proportion of choices are \textit{not} best responses to stated beliefs. In $3\times 3$ games, for instance, the percentage of best responses is typically between 45\% and 75\%. Furthermore, deviations from best responses are sensitive to expected payoffs, i.e. more costly mistakes occur less frequently.\footnote{See \citet{manski2004, croson2000, nyarko2002, costa2008, palfrey2009, rey2009, terracol2009, hyndman2009, hyndman2013, wang2011, hyndman2012, danz2012, sutter2013, hoffmann2014, trautmann2015, goeree2017, alempaki2019,esteban2020}.}  A second finding is that beliefs are generally \textit{not} correct. Subjects do not have perfect foresight -- not in one-shot games nor in repeated settings where learning is possible. Nonetheless, reported beliefs are not arbitrary -- they predict opponents' choices better than chance and reflect strategic thinking.\footnote{\citet{huck2002, weizsacker2003, ehrblatt2005, hyndman2010, hyndman2013, schlag2015, goeree2017, alempaki2019, polonio2019,friedman2019,esteban2020,foroughifar2020,wolff2021}.} A final finding is that beliefs differ across individuals in the same strategic role, and the same individual often reports different beliefs when facing the same strategic situation. If anything, belief heterogeneity appears \textit{more} pronounced than choice heterogeneity both within and across individuals.\footnote{This increased heterogeneity possibly reflects the different nature of belief elicitation, i.e. a  distribution over choices is elicited rather than a single discrete choice (pure strategy). But it may also reflect a different cognitive process when thinking about others' choices versus thinking about own choices \citep[e.g.][]{bhatt2005}.}

We look for minimal behavioral assumptions that produce a falsifiable theory consistent with these empirical observations.
\vspace*{-2.5mm}
\begin{itemize}\addtolength{\itemsep}{-2mm}
\item[(i)] \textbf{Monotonicity}: choice frequencies are positively but imperfectly related to expected payoffs based on beliefs, i.e. players ``better respond'' but do not necessarily best respond.
\item[(ii)] \textbf{Consequential Unbiasedness}: expected payoffs based on beliefs generate choice distributions that do not differ systematically from observed choices.
\item[(iii)] \textbf{Set Valued}: belief and choice predictions are set valued with predicted choices being a subset of predicted beliefs.
\end{itemize}
\vspace*{-2.5mm}
Monotonicity captures the idea that unbiased mistakes, or ``noise,'' dilute choice probability from better to worse options but not to the extent that they overturn their ordering. Consequential unbiasedness relaxes the rational-expectations assumption of correct beliefs -- the intuition is that there is no need to change beliefs if doing so leaves choices unaffected.\footnote{For instance, in a Prisoner's Dilemma game, all beliefs result in the same ranking of payoffs and, hence, choices. There is no need to correctly anticipate the other's behavior.} These behavioral postulates have empirical support \textit{per se}, and the requirement that $M$ equilibrium captures \textit{any} behavior that abides by them naturally results in a set-valued solution concept.

Being set valued, $M$ equilibrium imposes weak restrictions on the data. Hence, when its predictions are not violated, its contribution lies in setting a new baseline from which further restrictions can be added, resulting in sharper prediction that are also not violated. Importantly, $M$ equilibrium \textit{is} falsifiable. As we show in Section \ref{sec:robust}, the size of any $M$-choice set falls factorially fast with the number of players and strategies. For instance, in a three-player game where each player has three strategies, the size of any $M$-choice set is less than half a percent of the set of all strategy profiles, and even if there are multiple $M$ equilibria their total size is less than one-sixth. As such, we do not expect $M$ equilibrium to always be correct. When its predictions fail, the weak assumptions underlying $M$ equilibrium offer important insights about behavior: is monotonicity violated or are beliefs biased in terms of payoff consequences?

\subsection{Theoretical Motivation for $M$ equilibrium}
\label{theoMotivation}

The idea that players may make suboptimal choices, or \q{tremble,} has been part of game theory since its early days. The \textit{refinement} literature started with \cite{Selten1975} who introduced the notion of {\em perfectness} to define robust Nash equilibria. In \citeauthor{Selten1975}'s original definition, a perfect equilibrium is the limit of Nash equilibria of perturbed games in which strategy sets are restricted to an interior simplex to ensure no strategy receives weight less than $\varepsilon>0$. A more convenient definition is that a Nash profile $\sigma$ is perfect if there exists a sequence $\sigma(\varepsilon)$ of interior profiles to which $\sigma$ is a best response and that limits to $\sigma$ when $\varepsilon$ tends to zero.

From an empirical viewpoint, the notion of vanishing trembles is too restrictive as sizeable deviations are the rule rather than the exception. To be robust, a solution concept should include any deviation that satisfies monotonicity. Generically, this yields a full-dimensional choice \textit{set}, which, generically, is supported by a \textit{set} of consequentially unbiased beliefs. In other words, the reason for embracing a set-valued solution concept is to accommodate the more fundamental behavioral assumptions of monotonicity and consequential unbiasedness, which replace standard notions of rationality.

The insistence on solution \textit{sets} rather than fixed-\textit{points} departs from the approach taken in the (refinement) literature. We show, however, that $M$ equilibrium serves as a \textit{meta theory} for various existing models that hitherto may have been considered unrelated. Specifically, $M$-choice sets unify equilibria from models based on strategy perturbations (\`a la \citeauthor{Selten1975}, \citeyear{Selten1975}), payoff perturbations (\`a la \citeauthor{Harsanyi1973}, \citeyear{Harsanyi1973}), as well as $\varepsilon$-proper equilibria (\`a la \citeauthor{Myerson1978}, \citeyear{Myerson1978}).

In particular, the $M$-choice sets contain \textit{all} Quantal Response Equilibria (\citeauthor{McKelveyPalfrey1995}, \citeyear{McKelveyPalfrey1995}), i.e. all the fixed-points that can be obtained by varying players' quantal responses over an infinite-dimensional function space. As such, $M$ equilibrium provides the perfect lens to address the critique by \cite{HaileHortacsuKosenok2008} about the (lack of) empirical content of QRE. Falsifiability of $M$ equilibrium implies that QRE \textit{has} empirical content. The weak monotonicity assumption underlying $M$ equilibrium renders the critique by \cite{HaileHortacsuKosenok2008} vacuous. To generate \textit{any} type of behavior one would have to drop monotonicity and allow options with lower expected payoffs to be chosen more frequently.

Besides unifying choice predictions from seemingly unrelated models, $M$ equilibrium highlights the role of beliefs. The aforementioned models all impose a \textit{rational-expectations} assumption that is used to \q{close the model} to obtain a set of fixed-point conditions. This approach is not unique to game theory and underlies virtually all modeling in general equilibrium, finance, and macro-economics. The rational-expectations approach has met considerable criticism (e.g. \citeauthor{Simon1984}, \citeyear{Simon1984}; \citeauthor{Woodford1990}, \citeyear{Woodford1990}; \citeauthor{Sargent1991}, \citeyear{Sargent1991}), especially as it is applied to environments far more complex than the normal-form games studied here.\footnote{See, for instance, \cite{Kirman2014} for a survey of critiques of rational expectations.}

Criticism of rational expectations stems partly from the lack of empirical support, as foreseen by \cite{Muth1961} who introduced the concept but questioned its empirical validity (see Section 4.1). In addition, perfect foresight is not the only assumption that yields a consistent model of behavior. For instance, in \citeauthor{FudenbergLevine1993}'s (\citeyear{FudenbergLevine1993}) \q{self-confirming equilibrium} for extensive-form games, each player best responds to beliefs about the opponent's play and beliefs are correct on, but not off, the self-confirming equilibrium path. Hence, if a self-confirming equilibrium occurs repeatedly, no player needs to adjust their beliefs. \cite{FudenbergLevine1993} show that self-confirming equilibria need not be Nash equilibria. Related work in macroeconomics considers adaptive learning rules under which the economy may converge to \q{sunspot equilibria} \citep{Woodford1990}. And in related work in finance, agents are modeled as using heterogenous forecasts obtained by maximizing accuracy subject to individual resource constraints, which results in \q{self-justified equilibria} \citep{KublerScheidegger2018}.

Consequential unbiasedness allows for the possibility that beliefs settle even when they are biased, as in the aforementioned approaches. $M$ equilibrium applies consequential unbiasedness in the context of normal-form games and is agnostic as to how beliefs form, i.e. it does not assume a specific learning dynamic. $M$ equilibrium thus circumvents the \q{wilderness of behavioral assumptions} that \cite{Sargent1999} imagined necessary when abandoning rational expectations. Consequentially unbiased $M$-equilibrium beliefs are determined by the same simple ordinal payoff comparisons that govern the monotone $M$-equilibrium choices.

Since $M$ equilibrium does not impose rational expectations to close the model, it is not defined in terms of a set of fixed-point equations. Instead, the $M$-equilibrium sets are characterized by finitely many polynomial (in)equalities. Such sets are called \textit{semi-algebraic} and a main feature is that they can be computed using a (finite) algorithm. Indeed, as we show in Section \ref{sec:M}, it is typically \textit{easier} to determine $M$ equilibrium sets than to compute any of the fixed-points of the various models, such as Quantal Response Equilibrium, it contains.

\subsection{Related Prior Approaches}
\label{sec:prior}

\citeauthor{McKelveyPalfrey1995}'s (\citeyear{McKelveyPalfrey1995}) {\em Quantal Response Equilibrium} also provides a coherent model that incorporates sizeable choice deviations.
QRE is a parametric theory that requires the specification of players' \q{quantal} responses that map expected payoffs to choice probabilities. A QRE profile is then defined as a fixed-point. In other words, like Nash equilibrium, QRE assumes rational expectations even though its fixed-point equations are transcendental, e.g. for logit-QRE only numerical solutions are possible. \cite{crawford2018} notes that ``QRE is a fixed point in a high-dimensional space of distributions, making its thinking justification cognitively far more demanding than for Nash equilibrium.''

A different approach relaxes the rational-expectations assumption while maintaining best responses. Prominent examples include the \textit{Level-k model} \citep{stahl1994,stahl1995,nagel1995}, \textit{Cognitive Hierarchy} \citep{camerer2004}, and \emph{Random Belief Equilibrium} \citep{friedman2005}. Unlike QRE, which is essentially a homogeneous Bayes-Nash equilibrium, these models allow for heterogeneity in choices and beliefs. However, the best-response assumption they impose is not corroborated by the data (see Section \ref{empMotivation}). A final class of models, such as \emph{Noisy Introspection} \citep{goeree2004} and \emph{Subjective QRE} \citep{rogers2009}, relax both the perfect-foresight and perfect-decisionmaking assumptions.

These prior approaches are all parametric. Their success in explaining experimental data relies on estimation of specific parameters, e.g. error distributions for QRE or level distributions in Cognitive Hierarchy. At a more conceptual level, these models rest on specific assumptions about how players form beliefs and make choices based on their beliefs.  Assumptions, which, in our view, are unlikely to be met exactly or even approximately. We follow a different, if not opposite, approach by imposing only minimal conditions to obtain a falsifiable theory of strategic play compatible with the rich set of empirical regularities detailed in Section \ref{empMotivation}.

\subsection{Road Map}

This paper has a theory and an experimental part. For readers mainly interested in the relation between $M$ equilibrium and various prior approaches -- Nash, perfect equilibria, proper equilibria, QRE -- Sections 2--2.4 are most relevant. For those interested in how $M$ equilibrium relates to a parametric model, $\mu$ equilibrium, which can be used for calibration exercises, see Section 2.2. And for those interested in computing $M$ equilibria, see Section 2.1 and Appendix A. Robustness and falsifiability of $M$ equilibria are covered in Section 2.5.

Section 3 reports data from various $2\times 2$ and $3\times 3$ games used to test $M$ equilibrium and estimate $\mu$ equilibrium and compare them to leading behavioral-game-theory models.
The concluding Section \ref{sec:conc} discusses methodological implications of $M$ equilibrium. Appendices B--G contain proofs, details of the belief-elicitation procedure, statistical analyses, and instructions.


\section{$M$ Equilibrium}
\label{sec:M}

A finite normal-form game $G$ is a tuple $(N,\{S_{i},\Pi_{i}\}_{i\,\in\,N})$, with $N=\{1,\ldots,n\}$ the set of players, $S_{i}=\{s_{i1},\ldots,s_{iK_{i}}\}$ the set of pure strategies for player $i$, and $\Pi_{i}:S\rightarrow\Re$, where $S=\prod_{i=1}^nS_i$, player $i$'s payoff function. Let $\Sigma_{i}$ denote the set of probability distributions over $S_i$ and let $\Sigma=\prod_{i=1}^n\Sigma_i$. For $i\in N$,  $\Omega_i=\prod_{j\neq i}\Sigma_i$ is player $i$'s set of (independent) beliefs. Let $\Omega=\prod_{i=1}^n\Omega_i$.
We extend player $i$'s payoff function over $\Sigma_i\times\Omega_i$ as follows: for $(\sigma_i,\omega_i)\in\Sigma_i\times\Omega_i$, player $i$'s expected payoff is $\sum_{k=1}^{K_i}\sigma_{ik}\pi_{ik}(\omega_i)$ where $\sigma_{ik}$ is the probability that player $i$ chooses strategy $s_{ik}$ and $\pi_{ik}(\omega_i)=\sum_{s_{-i}\in S_{-i}}p_i(s_{-i})\Pi_{i}(s_k,s_{-i})$, with $p_i(s_{-i})=\prod_{j\neq i}\omega_{ij}(s_j)$, is the expected payoff associated with $s_{ik}$. Let $\pi_i(\sigma_{-i})$ denote the vector of expected payoffs when player $i$'s beliefs are correct, i.e. $\omega_i=\sigma_{-i}$. Finally, let $\Sigma_{int}$ denote the set of totally-mixed strategies and, for $S\subset\Sigma$ and $S'\subset\Omega$, let $\overline{S}$ and $\overline{S}'$ denote their (topological) closures relative to $\Sigma$ and $\Omega$ respectively.
\begin{definition} \label{Mdef}
The pair $(\overline{M}^c,\overline{M}^b)\subseteq\Sigma\times\Omega$ form an {\bf\em M Equilibrium} of $G$ if they are the closures of the maximal sets $M^c\subseteq\Sigma_{int}$ and $M^b\subseteq\Omega$ that satisfy
\begin{subequations}\label{M}
\begin{align}
  \pi_{ij}(\omega_{i})\,<\,\pi_{ik}(\omega_{i}) &\,\,\Longrightarrow\,\sigma_{ij}\,<\,\sigma_{ik}\label{M-mon} \\[1mm]
  \pi_{ij}(\omega_{i})\,<\,\pi_{ik}(\omega_{i}) &\,\Longleftrightarrow\,\pi_{ij}(\sigma_{-i})\,<\,\pi_{ik}(\sigma_{-i})\label{M-conu}
\end{align}
\end{subequations}
for all $i\in N$, $1\leq j,k\leq K_i$, $\sigma\in M^c$, $\omega\in M^b$. An $M$ equilibrium is \textbf{colorable} if \eqref{M-mon} is sharpened to $\pi_{ij}(\omega_{i})<\pi_{ik}(\omega_{i})\Leftrightarrow\sigma_{ij}<\sigma_{ik}$. The set of all $M\!$ equilibria of $G$ is denoted $\overline{\mathcal{M}}(G)$.
\end{definition}
Definition \ref{Mdef} implements the three desiderata of the Introduction: strategies with higher expected payoffs are chosen more likely (monotonicity), beliefs generate the same payoff ordering as observed choices (consequential unbiasedness), and the solution is defined in terms of maximal sets.\footnote{Maximality is defined with respect to set inclusion ($\subseteq$). To show existence of a maximal set consider the collection of sets $\mathcal{C} \equiv \{ (A,B) \in 2^{\Sigma_{int}} \times 2^{\Omega}\,|\,\eqref{Mdef}\,\,\text{holds}\,\,\forall \sigma \in A,\,\omega \in B\}$. We define a partial order $\preceq$ on $\mathcal{C}$ as follows: $(A,B) \preceq (A',B')$ if $A \subseteq A'$ and $B \subseteq B'$. Let $\mathcal{C}_o$ denote any totally ordered subcollection of $\mathcal{C}$. Define $\mathcal{A}=\bigcup\,\{ A \subseteq 2^{\Sigma_{int}}\,|\,\exists B \subset \Omega\,\,\text{such that}\,\, (A,B) \in \mathcal{C}_o \}$ and $\mathcal{B}=\{ \omega \in \Omega\,|\,\eqref{Mdef}\,\,\text{holds}\,\,\forall \sigma \in \mathcal{A}\}$. Then $\mathcal{A}\times\mathcal{B}$ is an upper bound of $\mathcal{C}_o$ and existence of a maximal element of $\mathcal{C}$ is implied by Zorn's lemma, see e.g. \cite{Halmos1998}. We should like to thank Lucas Pahl for helping us formulate this argument.} Note that \eqref{M-mon} does not restrict choice probabilities when the associated expected payoff are equal. Colorability requires choice probabilities to be equal in this case.

The \q{equilibrium} terminology reflects the consistency condition \eqref{M-conu}. As with other equilibrium concepts this implicitly assumes players' beliefs and choices have settled, e.g. toward the end of an experiment. However, unlike other equilibrium concepts, consequential unbiasedness does not require beliefs to be correct or the same in each period. Rather beliefs and choices should belong to the same $M$-belief and $M$-choice \textit{sets} in each period.

\medskip

\noindent\textbf{Example 1: Mondriaan on a Square}

\begin{figure}[t]
\begin{center}
\begin{tabular}{ccc}
\multicolumn{1}{r}{}& \multicolumn{1}{|c}{$A$} & \multicolumn{1}{c}{$B$}  \\ \cline{1-3}
\multicolumn{1}{r}{\rule{0pt}{5mm}$A$} & \multicolumn{1}{|c}{$0$, $0$} & \multicolumn{1}{c}{$6$, $1$} \\
\multicolumn{1}{r}{\rule{0pt}{5mm}$B$} & \multicolumn{1}{|c}{$1$, $10$} & \multicolumn{1}{c}{$2$, $2$} \\
\end{tabular}

\vspace*{2mm}
An asymmetric game of chicken.
\vspace*{2mm}

\begin{tikzpicture}[scale=4.7]
\draw[line width=.5pt,gray] (0,0) -- (1,0) -- (1,1) -- (0,1) -- (0,0);
\draw[line width=1pt] (1/2,0) -- (1/2,1);
\draw[line width=1pt] (0,1/2) -- (1,1/2);
\node[scale=0.6] at (-0.03,-0.06) {$0$};
\node[scale=0.6] at (-0.03,1) {$1$};
\node[scale=0.6] at (1,-0.06) {$1$};
\node[scale=0.8] at (0.65,-0.06) {$p$};
\node[scale=0.8] at (-0.08,0.65) {$q$};

\node[scale=0.6] at (0.25,0.32) {$p<\hf$};
\node[scale=0.6] at (0.25,0.18) {$q<\hf$};
\node[scale=0.6] at (0.75,0.32) {$p>\hf$};
\node[scale=0.6] at (0.75,0.18) {$q<\hf$};
\node[scale=0.6] at (0.25,0.82) {$p<\hf$};
\node[scale=0.6] at (0.25,0.68) {$q>\hf$};
\node[scale=0.6] at (0.75,0.82) {$p>\hf$};
\node[scale=0.6] at (0.75,0.68) {$q>\hf$};

\node at (1.9,.46) {
\begin{tikzpicture}[scale=4.7]
\draw[line width=.5pt,gray] (0,0) -- (1,0) -- (1,1) -- (0,1) -- (0,0);
\draw[line width=1pt] (4/5,0) -- (4/5,1);
\draw[line width=1pt] (0,8/9) -- (1,8/9);
\node[scale=0.6] at (-0.03,-0.06) {$0$};
\node[scale=0.6] at (-0.03,1) {$1$};
\node[scale=0.6] at (1,-0.06) {$1$};
\node[scale=0.8] at (0.65,-0.06) {$\omega$};
\node[scale=0.8] at (-0.08,0.65) {$\nu$};
\node[scale=0.6] at (0.4,0.52) {$\pi_{C}(A)>\pi_{C}(B)$};
\node[scale=0.6] at (0.4,0.38) {$\pi_{R}(A)>\pi_{R}(B)$};
\node[scale=0.6] at (0.2,0.94) {$\pi_{C}(A)<\pi_{C}(B),$};
\node[scale=0.6] at (0.6,0.94) {$\pi_{R}(A)>\pi_{R}(B)$};
\node[scale=0.6] at (.9,.65) [rotate=-90] {$\pi_{C}(A)>\pi_{C}(B),$};
\node[scale=0.6] at (.9,.25) [rotate=-90] {$\pi_{R}(A)<\pi_{R}(B)$};
\end{tikzpicture}};

\node at (0.43,.46-1.2) {
\begin{tikzpicture}[scale=4.7]
\draw[line width=.5pt,gray] (0,0) -- (1,0) -- (1,1) -- (0,1) -- (0,0);
\filldraw[fill=mondrianYellow,draw=mondrianYellow,opacity=1] (1/2,1/2) -- (4/5,1/2) -- (4/5,8/9) -- (1/2,8/9) -- (1/2,1/2);
\filldraw[fill=mondrianBlue,draw=mondrianBlue,opacity=1] (4/5,1/2) -- (1,1/2) -- (1,0) -- (4/5,0) -- (4/5,1/2);
\filldraw[fill=mondrianRed,draw=mondrianRed,opacity=1] (1/2,8/9) -- (1/2,1) -- (0,1) -- (0,8/9) -- (1/2,8/9);
\node[scale=0.6pt] at (-0.03,-0.06) {$0$};
\node[scale=0.6pt] at (-0.03,1) {$1$};
\node[scale=0.6pt] at (1,-0.06) {$1$};
\node[scale=0.8] at (0.65,-0.06) {$p$};
\node[scale=0.8] at (-0.08,0.65) {$q$};
\draw[line width=1.5pt] (4/5,0) -- (4/5,8/9) -- (0,8/9);
\end{tikzpicture}};

\node at (1.9,.46-1.2) {
\begin{tikzpicture}[scale=4.7]
\draw[line width=.5pt,gray] (0,0) -- (1,0) -- (1,1) -- (0,1) -- (0,0);
\filldraw[fill=mondrianYellow,draw=mondrianYellow,opacity=1] (0,0) -- (4/5,0) -- (4/5,8/9) -- (0,8/9) -- (0,0);
\filldraw[fill=mondrianBlue,draw=mondrianBlue,opacity=1] (4/5,8/9) -- (1,8/9) -- (1,0) -- (4/5,0) -- (4/5,8/9);
\filldraw[fill=mondrianRed,draw=mondrianRed,opacity=1] (4/5,8/9) -- (4/5,1) -- (0,1) -- (0,8/9) -- (4/5,8/9);
\node[scale=0.6pt] at (-0.03,-0.06) {$0$};
\node[scale=0.6pt] at (-0.03,1) {$1$};
\node[scale=0.6pt] at (1,-0.06) {$1$};
\node[scale=0.8] at (0.65,-0.06) {$\omega$};
\node[scale=0.8] at (-0.08,0.65) {$\nu$};
\draw[line width=1.5pt] (4/5,0) -- (4/5,1);
\draw[line width=1.5pt] (0,8/9) -- (1,8/9);
\end{tikzpicture}};
\end{tikzpicture}

\end{center}
\vspace*{-5mm}
\caption{Construction of $M$-choice and $M$-belief sets for the asymmetric game of chicken shown at the top. The middle panels show partitions of the unit square based on orderings of choice probabilities (left) and expected payoffs (right). There are three $M$-choice sets for which these orderings match, see the lower-left panel, which can be labeled or colored by the ordering they represent. The colored sets in the lower-right panel show the beliefs that generate the same ordering of expected payoffs as choices of the same color.}\label{newAGCfig}
\vspace*{-4mm}
\end{figure}

\noindent To illustrate the construction of $M$-choice and $M$-belief sets, consider the asymmetric game of chicken at the top of Figure~\ref{newAGCfig}. Let $p$ and $q$ denote the probabilities with which Column and Row choose $A$. And let $\nu$ and $\omega$ denote Column's and Row's beliefs that the other player chooses $A$. Since $B$ is chosen with complementary probability, the sets of choice and belief profiles can be summarized by unit squares consisting of the pairs $(p,q)$ and $(\nu,\omega)$ respectively, see the second row of Figure~\ref{newAGCfig}. In the left panel, the quadrants reflect the four possible orderings of choice probabilities. In the right panel, $\pi_C$  and $\pi_R$ denote the expected payoffs for Column and Row. The horizontal line at $\omega=\deel{4}{5}$ indicates the belief for which Row is indifferent and the vertical line at $\nu=\deel{8}{9}$ does the same for Column. These \q{indifference curves} divide the unit square in four rectangles where expected payoffs are strictly ordered.

To check for \textit{monotonicity}, we match each of these rectangles with a quadrant in the left panel. For instance, the largest rectangle on the right for which the expected payoff of $A$ exceeds that of $B$ for both players is matched with the north-east quadrant on the left. Likewise, the smallest rectangle on the right for which the expected payoff of $B$ exceeds that of $A$ for both players is matched with the south-west quadrant on the left. More generally, these matchings are such that the choice probabilities on the left are ranked the same way as the payoffs given beliefs on the right. \textit{Consequential unbiasedness} requires that the ranking of expected payoffs based on beliefs is the same as that based on choices. Graphically, this equilibrium condition is implemented by superimposing the rectangles of the right panel on the left unit square.

The intersection of any quadrant with the matched rectangle determines the $M$-choice and $M$-belief sets. For instance, superimposing the smallest rectangle of the right panel on the left panel yields an empty intersection, i.e. there is no $M$ equilibrium in which both players more likely chose $B$. In contrast, superimposing the largest rectangle of the right panel on the left panel yields a non-empty intersection indicated by the yellow $M$-choice set in the bottom-left panel of Figure \ref{newAGCfig}.  The corresponding $M$-belief set is simply the large rectangle itself, as indicated by the yellow set in the bottom-right panel. Repeating this procedure for the different payoff rankings yields three full-dimensional $M$ equilibria that are colorable, see the bottom panels of Figure~\ref{newAGCfig}. (The thick black lines in the lower panels are discussed in Example 3 below.)

\medskip

\noindent\textbf{Example 2: Mondriaan on a Simplex}

\noindent The construction of $M$ sets is similar for games with more players or strategies, except that we cannot visualize them as the sets of choice and belief profiles are high dimensional. An exception is symmetric games with three strategies if we restrict attention to \textit{symmetric} $M$-equilibria. Then the set of choice profiles, for instance, is a single two-dimensional simplex, i.e. we do not take a product over players. We use this simplification in various parts of the paper. An important caveat is that drawing symmetric $M$ equilibria in a single simplex overstates their sizes, e.g. for a two-player game a symmetric $M$ equilibrium of size one-sixth corresponds to only a $\deel{1}{36}$ fraction of the set of all profiles.

Consider, for instance, the symmetric two-player game shown at the top of Figure \ref{mondrianGame} with pure strategies labeled $R$ (red), $B$ (blue), and $Y$ (yellow). The three lines in the middle-right panel show the set of beliefs for which two pure strategies yield the same expected payoff. These indifference curves partition the simplex into different areas with a strict ranking of payoffs. The middle-left panel shows the \q{barycentric} division of the simplex. Monotonicity and consequential unbiasedness require the ranking of expected payoffs to match that of the underlying choice probabilities. This results in three full-dimensional $M$-equilibrium choice sets colored by the same (strict) ranking of choice probabilities and expected payoffs, see the lower-left panel of Figure \ref{mondrianGame}. The lower-right panel shows the $M$-equilibrium belief sets. Each $M$-equilibrium choice set is supported by beliefs of the same color, i.e. beliefs that produce the same ranking of expected payoffs as the choices.

\begin{figure}[t]
\begin{center}
\begin{tabular}{c|ccc}
& $R$ & $B$ & $Y$ \\ \cline{1-4}
\multicolumn{1}{r}{\rule{0pt}{5mm}$R$} & \multicolumn{1}{|c}{$9$, $9$} & \multicolumn{1}{c}{$6$, $8$} & \multicolumn{1}{c}{$4$, $4$} \\
\multicolumn{1}{r}{$B$} & \multicolumn{1}{|c}{$8$, $6$} & \multicolumn{1}{c}{$8$, $8$} & \multicolumn{1}{c}{$2$, $4$} \\
\multicolumn{1}{r}{$Y$} & \multicolumn{1}{|c}{$4$, $4$} & \multicolumn{1}{c}{$4$, $2$} & \multicolumn{1}{c}{$5$, $5$} \\
\end{tabular}

\vspace*{2mm}
Mondriaan's \q{composition of Red, Yellow, and Blue} game.
\vspace*{2mm}

\begin{tikzpicture}[scale=5]
\draw[line width=.5pt,gray] (0,0) -- (1,0) -- (1/2,1/2*3^.5) -- (0,0);
\draw[line width=1pt] (1/2,0) -- (1/2,1/2*3^.5);
\draw[line width=1pt] (0,0) -- (3/4,1/4*3^.5);
\draw[line width=1pt] (1/4,1/4*3^.5) -- (1,0);
\node[scale=0.6] at (-0.03,-0.03) {$R$};
\node[scale=0.6] at (1.03,-0.03) {$B$};
\node[scale=0.6] at (1/2,1/2*3^.5+0.03) {$Y$};

\node[scale=0.6] at (.28,1/20) {$\sigma_R\!>\!\sigma_B\!>\!\sigma_Y$};
\node[scale=0.6] at (.7,1/20) {$\sigma_B\!>\!\sigma_R\!>\!\sigma_Y$};
\node[scale=0.6] at (.25,.2) [rotate=30] {$\sigma_R\!>\!\sigma_Y\!>\!\sigma_B$};
\node[scale=0.6] at (.73,.21) [rotate=-30] {$\sigma_B\!>\!\sigma_Y\!>\!\sigma_R$};
\node[scale=0.6] at (.43,.51) [rotate=90] {$\sigma_Y\!>\!\sigma_R\!>\!\sigma_B$};
\node[scale=0.6] at (.58,.51) [rotate=90] {$\sigma_Y\!>\!\sigma_B\!>\!\sigma_R$};

\node at (1.75,.43) {
\begin{tikzpicture}[scale=5]
\draw[line width=.5pt,gray] (0,0) -- (1,0) -- (1/2,1/2*3^.5) -- (0,0);
\draw[line width=1pt] (1/3,0) -- (3/4,1/4*3^.5);
\draw[line width=1pt] (5/12,5/12*3^.5) -- (2/3,1/3*3^.5);
\draw[line width=1pt] (2/7,2/7*3^.5) -- (5/7,2/7*3^.5);
\node[scale=0.6] at (-0.03,-0.03) {$R$};
\node[scale=0.6] at (1.03,-0.03) {$B$};
\node[scale=0.6] at (1/2,1/2*3^.5+0.03) {$Y$};
\node[scale=0.6] at (1/3,1/4) {$\pi_R\!>\!\pi_B\!>\!\pi_Y$};
\node[scale=0.6] at (2/3,1/10) {$\pi_B\!>\!\pi_R\!>\!\pi_Y$};
\node[scale=0.6] at (1/2,.55) {$\pi_R\!>\!\pi_Y\!>\!\pi_B$};
\node[scale=0.6] at (.7,.8) {$\pi_Y\!>\!\pi_R\!>\!\pi_B$};
\end{tikzpicture}};

\node at (0.5,-0.6) {
\begin{tikzpicture}[scale=5]
\draw[line width=.5pt,gray] (0,0) -- (1,0) -- (1/2,1/2*3^.5) -- (0,0);
\node[scale=0.6] at (-0.03,-0.03) {$R$};
\node[scale=0.6] at (1.03,-0.03) {$B$};
\node[scale=0.6] at (1/2,1/2*3^.5+0.03) {$Y$};

\filldraw[fill=mondrianRed,draw=mondrianRed,opacity=1] (0,0) -- (1/3,0) -- (1/2,1/10*3^.5) -- (1/2,1/6*3^.5) -- (0,0);
\filldraw[fill=mondrianBlue,draw=mondrianBlue,opacity=1] (1,0) -- (1/2,0) -- (1/2,1/10*3^.5) -- (4/7,1/7*3^.5) -- (1,0);
\filldraw[fill=mondrianYellow,draw=mondrianYellow,opacity=1] (1/2,1/2*3^.5) -- (5/12,5/12*3^.5) -- (1/2,7/18*3^.5) -- (1/2,1/2*3^.5);

\draw[line width=1.5pt,black] (1/3,0) -- (4/7,1/7*3^.5);
\draw[line width=1.5pt,black] (5/12,5/12*3^.5) -- (1/2,7/18*3^.5);


\end{tikzpicture}};

\node at (1.75,-0.6) {
\begin{tikzpicture}[scale=5]
\draw[line width=.5pt,gray] (0,0) -- (1,0) -- (1/2,1/2*3^.5) -- (0,0);
\node[scale=0.6] at (-0.03,-0.03) {$R$};
\node[scale=0.6] at (1.03,-0.03) {$B$};
\node[scale=0.6] at (1/2,1/2*3^.5+0.03) {$Y$};

\filldraw[fill=mondrianRed,draw=mondrianRed,opacity=1] (0,0) -- (1/3,0) -- (3/4,1/4*3^.5) -- (5/7,2/7*3^.5) -- (2/7,2/7*3^.5) -- (0,0);
\filldraw[fill=mondrianBlue,draw=mondrianBlue,opacity=1] (1,0) -- (1/3,0) -- (3/4,1/4*3^.5) -- (1,0);
\filldraw[fill=mondrianYellow,draw=mondrianYellow,opacity=1] (1/2,1/2*3^.5) -- (5/12,5/12*3^.5) -- (2/3,1/3*3^.5) -- (1/2,1/2*3^.5);

\draw[line width=1.5pt,black] (1/3,0) -- (3/4,1/4*3^.5);
\draw[line width=1.5pt,black] (5/12,5/12*3^.5) -- (2/3,1/3*3^.5);

\end{tikzpicture}};

\end{tikzpicture}
\vspace*{-1mm}
\caption{Construction of $M$-choice and $M$-belief sets for Mondriaan's game shown at the top. The colored sets in the lower-left panel consist of choices for which the barycentric division of the simplex (middle-right panel) matches the division based on expected payoffs (middle-left panel). The colored sets in the lower-right panel consist of beliefs that generate the same strict ordering of expected payoffs as choices of the same color.}\label{mondrianGame}
\end{center}
\vspace*{-7mm}
\end{figure}

\medskip

\noindent\textbf{Example 3: Lower-Dimensional $M$ sets}

\noindent An $M$-equilibrium exists for any normal-form game (Section \ref{sec:FPmodels}) and, generically, there exists at least one full-dimensional $M$ set (Section \ref{sec:robust}). But non-generic games may have only lower-dimensional $M$ sets.
For a matching-pennies game,\footnote{A $2\times 2$ symmetric zero-sum game with strategies labeled ``Heads'' and ``Tails'' in which player 1 wins \$1 if players' choices match and loses \$1 otherwise.} for instance,
the unique $M$-equilibrium consists of a \textit{single} profile that corresponds to both players randomizing uniformly. This $M$ equilibrium is colorable as both choice probabilities and expected payoffs are equal.

In generic games, lower-dimensional $M$ sets and full-dimensional $M$ sets may co-exist. In particular, there may be lower-dimensional $M$ sets contained in the boundaries of full-dimensional $M$ sets. As such they do not generate new predictions. In the asymmetric game of chicken of Figure \ref{newAGCfig}, for instance, there is a one-dimensional $M$ equilibrium associated with Column's indifference curve at $\nu=\deel{8}{9}$ where $\pi_{C}(A)=\pi_{C}(B)$. The inequalities in \eqref{Mdef} impose no restrictions on the choice probabilities of strategies with equal expected payoffs. The only restriction on $p$ is that for Row we must have $q>\hf$ and thus $\pi_{R}(A)>\pi_{R}(B)$, which implies $p<\deel{4}{5}$. This $M$-equilibrium choice set is indicated by the thick horizontal line in the bottom-left panel of Figure \ref{newAGCfig}. The thick vertical line in this figure, as well as the thick lines in Figure \ref{mondrianGame}, are explained similarly. These lower-dimensional $M$ sets are not colorable as the ranking of choice probabilities changes over the set and does not everywhere match the ranking of payoffs. The set of supporting beliefs is simply the entire indifference curve, as indicated by the thick lines in the bottom-right panels of Figures \ref{newAGCfig} and \ref{mondrianGame}.

These examples illustrate that there can be multiple $M$ equilibria of different dimensions. Full-dimensional $M$ sets correspond to strict payoff rankings while lower-dimensional $M$ sets occur when there are payoff ties. A natural way to enumerate all $M$ equilibria is thus to consider all possible payoff rankings, strict or not.

\subsection{Rank Characterization}
\label{sec:rankChar}

We say $\mu_i\in\Sigma_i$ is \textit{regular} if it is totally mixed and all its entries are different. For $i\in N$, let $\mu_i$ be regular and let $V_i(\mu_i)$ denote the set of points in $\Sigma_i$ that results by permuting the entries of $\mu_i$. The convex hull $P_i(\mu_i)=\text{co}(V_i(\mu_i))$ is the \textit{permutahedron} generated by $\mu_i$. The introduction of this permutahedron serves four purposes.
\begin{itemize}[leftmargin=7mm]\addtolength{\itemsep}{-2mm}
\item[1.] In this section, we use the permutahedron to determine all $M$ equilibria based on the observation that any payoff ranking corresponds (one-to-one) to a face of the permutahedron with a dimension equal to the number of payoff ties.\footnote{We follow the convention that faces can be of different dimension and include the vertices, edges, etc.}
\item[2.] In Section \ref{sec:FPmodels}, we use it to define, for $i\in N$, player $i$'s \textit{rank correspondence}
        \begin{equation}\label{rankgen}
        rank^{\mu_i}_i(\pi_i)=\mbox{co}\bigl(\bigl\{\sigma_i\in V_i(\mu_i)\mid\pi_{ij}(\sigma_{-i})<\pi_{ik}(\sigma_{-i})\Rightarrow\sigma_{ij}<\sigma_{ik}\,\text{for}\,1\leq j,k\leq K_i\bigr\}\bigr)
        \end{equation}
        which is a bounded-rational version of player $i$'s best-reply correspondence as suboptimal choices are chosen with positive probability.  The rank correspondence underlies a parametric fixed-point model that is applied to data from experiments in Section \ref{subsec:empeval}.
\item[3.] In Section \ref{sec:Meta}, we use the rank correspondences to connect $M$ equilibrium with \citeauthor{Selten1975}'s (\citeyear{Selten1975}) notion of trembling-hand perfection and \citeauthor{Myerson1978}'s (\citeyear{Myerson1978}) notion of properness. The intuition is that, for any $\mu_i\in\Sigma_i$, the $rank_i^{\mu_i}$ correspondences are best-reply correspondences if players' strategy sets are restricted to $P_i(\mu_i)$, i.e. $rank_i^{\mu_i}(\pi_i)=\text{argmax}_{\sigma_i\in P_i(\mu_i)}\sum_{k}\sigma_{ik}\pi_{ik}$. Hence, the fixed-points defined by rank correspondences are Nash equilibria of the perturbed game $(N,\{P_i(\mu_i),\Pi_i\}_{i\,\in\,N})$.
\item[4.] Finally, in Section \ref{sec:Nash}, we use the permutahedron approach to geometrically determine the set of \textit{profect} Nash equilibria. We show that profectness is a refinement criterium that is less strict than properness and more strict than perfectness.
\end{itemize}
The top panels of Figure \ref{rank3} illustrate the rank correspondence $rank^{\mu_i}_i$ for $K_i=3$ and some regular $\mu_i$. The left panel shows the six possible strict rankings, corresponding to vertices of the permutahedron. The middle panel shows the six possible rankings with one tie, corresponding to the edges, and the right panel shows the one ranking with two payoff ties, corresponding to the permutahedron itself. The bottom panels show the rank correspondence for an irregular $\hat{\mu}_i$ of the form $\hat{\mu}_i(\varepsilon)=(1,\varepsilon,\varepsilon)/(1+2\varepsilon)$ with $\varepsilon\in(0,1)$. Now the set $V_i(\hat{\mu}_i)$ of permutations has only three elements and its convex hull is a simplex that has fewer faces than a permutahedron, which reflects the fact that only the highest payoff matters not the entire ranking of payoffs.\footnote{In general, the simplex has $\smash[b]{\sum_{k\,=\,0}^{K_i-1}\bigl(\!\!\begin{array}{cc} K_i \\ k \end{array}\!\!\bigr)=2^{K_i}-1}$ faces while the number of faces of the permutahedron is given by the Fubini number, $f_{K_i}$, which can be recursively defined as $\smash{f_{K_i}=\sum_{k\,=\,0}^{K_i-1}\bigl(\!\!\begin{array}{cc} K_i \\ k \end{array}\!\!\bigr)f_k}$ with $\rule{0pt}{5mm}f_0=1$.}

\begin{figure}[t]
\begin{center}
\begin{tikzpicture}[scale=4.3]
\draw[line width=.5pt,gray] (0,0) -- (1,0) -- (1/2,1/2*3^.5) -- (0,0);
\node[scale=0.8pt] at (-0.05,-0.05) {$s_1$};
\node[scale=0.8pt] at (1.05,-0.05) {$s_2$};
\node[scale=0.8pt] at (1/2,1/2*3^.5+0.05) {$s_3$};

\filldraw[mondrianRed] (1/4,3^.5/20) circle (.4pt);
\filldraw[mondrianOrange] (1/5,3^.5/10) circle (.4pt);
\filldraw[mondrianYellow] (9/20,3^.5*7/20) circle (.4pt);
\filldraw[mondrianGreen] (11/20,3^.5*7/20) circle (.4pt);
\filldraw[mondrianCyan] (4/5,3^.5/10) circle (.4pt);
\filldraw[mondrianBlue] (3/4,3^.5/20) circle (.4pt);

\node[scale=0.5pt] at (0.32,0.03) {$1\succ 2\succ 3$};
\node[scale=0.5pt] at (0.68,0.03) {$2\succ 1 \succ 3$};
\node[scale=0.5pt] at (0.27,0.25) {$1\succ 3\succ 2$};
\node[scale=0.5pt] at (0.73,0.25) {$2\succ 3\succ 1$};
\node[scale=0.5pt] at (0.25,0.66) {$3\succ 1\succ 2$};
\node[scale=0.5pt] at (0.75,0.66) {$3\succ 2\succ 1$};

\node at (1.85,0.43) {
\begin{tikzpicture}[scale=4.3]
\draw[line width=.5pt,gray] (0,0) -- (1,0) -- (1/2,1/2*3^.5) -- (0,0);
\node[scale=0.8pt] at (-0.05,-0.05) {$s_1$};
\node[scale=0.8pt] at (1.05,-0.05) {$s_2$};
\node[scale=0.8pt] at (1/2,1/2*3^.5+0.05) {$s_3$};

\draw[line width=1pt,mondrianRed] (1/4,3^.5/20) -- (1/5,3^.5/10);
\draw[line width=1pt,mondrianOrange] (1/5,3^.5/10) -- (9/20,3^.5*7/20);
\draw[line width=1pt,mondrianYellow] (9/20,3^.5*7/20) -- (11/20,3^.5*7/20);
\draw[line width=1pt,mondrianGreen] (11/20,3^.5*7/20) -- (4/5,3^.5/10);
\draw[line width=1pt,mondrianBlue] (4/5,3^.5/10) -- (3/4,3^.5/20);
\draw[line width=1pt,mondrianPurple] (3/4,3^.5/20) -- (1/4,3^.5/20);

\node[scale=0.5pt] at (0.5,0.03) {$1\sim 2\succ 3$};
\node[scale=0.5pt] at (0.28,0.4) [rotate=60] {$1\sim 3\succ 2$};
\node[scale=0.5pt] at (0.72,0.4) [rotate=-60] {$2\sim 3\succ 1$};
\node[scale=0.5pt] at (0.13,0.07) [rotate=30] {$1\!\succ\!2\!\sim\!3$};
\node[scale=0.5pt] at (0.87,0.07) [rotate=-30] {$2\!\succ\!1\!\sim\!3$};
\node[scale=0.5pt] at (0.5,0.64) {$3\succ 1\sim 2$};

\end{tikzpicture}};

\node at (3.2,0.43) {
\begin{tikzpicture}[scale=4.3]
\draw[line width=.5pt,gray] (0,0) -- (1,0) -- (1/2,1/2*3^.5) -- (0,0);
\node[scale=0.8pt] at (-0.05,-0.05) {$s_1$};
\node[scale=0.8pt] at (1.05,-0.05) {$s_2$};
\node[scale=0.8pt] at (1/2,1/2*3^.5+0.05) {$s_3$};

\filldraw[mondrianGrey] (1/4,3^.5/20) -- (1/5,3^.5/10) -- (9/20,3^.5*7/20) -- (11/20,3^.5*7/20) -- (4/5,3^.5/10) -- (3/4,3^.5/20);

\node[scale=0.5pt] at (0.5,0.28) {$1\sim 2\sim 3$};

\end{tikzpicture}};

\node at (0.5,-0.7) {
\begin{tikzpicture}[scale=4.3]
\draw[line width=.5pt,gray] (0,0) -- (1,0) -- (1/2,1/2*3^.5) -- (0,0);
\node[scale=0.8pt] at (-0.05,-0.05) {$s_1$};
\node[scale=0.8pt] at (1.05,-0.05) {$s_2$};
\node[scale=0.8pt] at (1/2,1/2*3^.5+0.05) {$s_3$};

\node[scale=0.5pt] at (0.2,0.06) {$1\succ 2,3$};
\node[scale=0.5pt] at (0.8,0.06) {$2\succ 1,3$};
\node[scale=0.5pt] at (0.5,1/2*3^.5-0.2) {$3\succ 1,2$};

\filldraw[mondrianRed] (3/13,1/13*3^.5) circle (.4pt);
\filldraw[mondrianBlue] (10/13,1/13*3^.5) circle (.4pt);
\filldraw[mondrianYellow] (1/2,9/26*3^.5) circle (.4pt);
\end{tikzpicture}};

\node at (1.85,-0.7) {
\begin{tikzpicture}[scale=4.3]
\draw[line width=.5pt,gray] (0,0) -- (1,0) -- (1/2,1/2*3^.5) -- (0,0);
\draw[line width=1pt,mondrianOrange] (3/13,1/13*3^.5) -- (1/2,9/26*3^.5);
\draw[line width=1pt,mondrianGreen] (1/2,9/26*3^.5) -- (10/13,1/13*3^.5);
\draw[line width=1pt,mondrianPurple] (3/13,1/13*3^.5) -- (10/13,1/13*3^.5);
\node[scale=0.8pt] at (-0.05,-0.05) {$s_1$};
\node[scale=0.8pt] at (1.05,-0.05) {$s_2$};
\node[scale=0.8pt] at (1/2,1/2*3^.5+0.05) {$s_3$};

\node[scale=0.5pt] at (0.5,0.08) {$1\sim 2\succ 3$};
\node[scale=0.5pt] at (0.32,0.38) [rotate=60] {$1\sim 3\succ 2$};
\node[scale=0.5pt] at (0.68,0.38) [rotate=-60] {$2\sim 3\succ 1$};

\end{tikzpicture}};

\node at (3.2,-0.7) {
\begin{tikzpicture}[scale=4.3]
\draw[line width=.5pt,gray] (0,0) -- (1,0) -- (1/2,1/2*3^.5) -- (0,0);
\filldraw[mondrianGrey] (3/13,1/13*3^.5) -- (10/13,1/13*3^.5) --(1/2,9/26*3^.5) -- (3/13,1/13*3^.5);
\node[scale=0.8pt] at (-0.05,-0.05) {$s_1$};
\node[scale=0.8pt] at (1.05,-0.05) {$s_2$};
\node[scale=0.8pt] at (1/2,1/2*3^.5+0.05) {$s_3$};

\node[scale=0.5pt] at (0.5,0.28) {$1\sim 2\sim 3$};

\end{tikzpicture}};

\end{tikzpicture}
\end{center}
\vspace*{-5mm}
\caption{The top (bottom) panels show the rank correspondence $rank^{\mu_i}_i(\pi_i)$ for some regular (irregular) $\mu_i$ when $K_i=3$. The left, middle, and right panels correspond to zero, one, and two payoff indifferences respectively. In the top panels, each of the thirteen payoff orderings corresponds to a face of the permutahedron. In the bottom panels, each of the seven maximum-payoff orderings corresponds to a face of the simplex.}\label{rank3}
\vspace*{-3mm}
\end{figure}

Finally, the above construction captures standard best-reply correspondences when $\varepsilon=0$ so that $\hat{\mu}_i=(1,0,0)$. Then the set $V_i(\hat{\mu}_i)$ of permutations is $S_i$ and its convex hull is $\Sigma_i$. For this irregular $\hat{\mu}_i$, player $i$'s rank correspondence is simply the best-reply correspondence $\br_i(\pi_i)$, which maps each payoff vector to a face of the \textit{full} simplex with a dimension equal to the number of payoff ties.

$M$ equilibrium is based on ordinal comparisons, e.g. the ranks of the expected payoffs based on choices equal the ranks of the expected payoffs based on beliefs. This can be implemented by requiring, for $i\in N$, that $rank^{\mu_i}_i(\pi_i(\sigma_{-i}))=rank^{\mu_i}_i(\pi_i(\omega_i))$ for some regular $\mu_i$. But note that if this equality holds for \textit{some} regular $\mu_i$ it holds for \textit{all} regular $\mu_i$.
To reflect this $\mu_i$ independence, we define $rank_i$ (without superscript) as follows. For $\varepsilon\in(0,1)$, let $\mu_i(\varepsilon)=c_\varepsilon(1,\varepsilon,\ldots,\varepsilon^{K_i-1})$ with $c_\varepsilon=(1-\varepsilon)/(1-\varepsilon^{K_i})$ to ensure $\mu_i\in\Sigma_i$. Then $rank_i$ is player $i$'s rank correspondence for a specific value of $\varepsilon$, say $\varepsilon=\hf$. Let $rank$ denote the concatenations of the $rank_i$ and let $\mathcal{P}=\prod_{i\in N}\mathcal{P}_i$ where $\mathcal{P}_i$ denotes the face set of the permutahedron $P_i(\mu_i)$.
\begin{proposition}\label{rankchar}
For $r\in\mathcal{P}$, the closure of
\begin{displaymath}\label{MbyRank}
(M^c_r,M^b_r)\,=\,\{(\sigma,\omega)\in\Sigma_{int}\times\,\Omega\,|\,rank(\sigma)\,\subseteq\,rank(\pi(\omega))\,=\,rank(\pi(\sigma))\,=\,r\},
\end{displaymath}
if non-empty, is an $M$ equilibrium of $G$.
\end{proposition}
This proposition offers a systematic approach to computing $M$ equilibria based on payoff comparisons being strict or not. This approach is finite since there are only a finite number of such payoff comparisons, or equivalently, since the face set $\mathcal{P}$ is finite.\footnote{See Appendix A for an alternative finite method to compute $M$-equilibrium sets.} For the Mondriaan game, for instance, we first consider the cases of no ties, which correspond to the six vertices of $\mathcal{P}$, to obtain three full-dimensional $M$ sets. Then we consider the cases of one payoff tie, which correspond to the six edges of $\mathcal{P}$, to obtain two lower-dimensional $M$ sets. Finally, we consider the case of two payoff ties, which corresponds to $\mathcal{P}$ itself, but for the Mondriaan game there is no $M$ set for this case. More generally, there typically will not be an $M$ equilibrium for all $r\in\mathcal{P}$, but for any normal-form game there is at least one, see the next section.

The set of all $M$ equilibria is the closure of the disjoint union $\mathcal{M}(G)=\sqcup_{r\,\in\,\mathcal{P}}(M^c_r,M^b_r)$, i.e.
\begin{displaymath}
  \mathcal{M}(G)\,=\,\{(\sigma,\omega)\in\Sigma_{int}\times\,\Omega\,|\,rank(\sigma)\,\subseteq\,rank(\pi(\omega))=rank(\pi(\sigma))\}
\end{displaymath}
The conditions that define $\mathcal{M}(G)$ are intuitive generalizations of those underlying Nash equilibria. The perfect-maximization assumption, $\sigma\in\br(\pi(\omega))$, is replaced with monotonicity, i.e. $rank(\sigma)\subseteq rank(\pi(\omega))$, and the perfect-beliefs assumption, $\omega=\sigma$, is replaced with consequential unbiasedness, i.e. $rank(\pi(\omega))=rank(\pi(\sigma))$.

\subsection{$\mu$ Equilibrium: A Parametric Fixed-Point Model}
\label{sec:FPmodels}

The non-parametric nature of $M$ equilibrium makes it less amenable to calibration exercises. Here we introduce a closely related parametric fixed-point model, $\mu$ \textit{equilibrium}, that can be applied to (laboratory) data. It is based on the rank correspondences in \eqref{rankgen}. One difference is that in the previous section it did not matter which regular $\mu_i$ we used but in the parametric model the $\mu_i$ will serve as \q{rationality parameters,} as explained below.
For $i\in N$, let $\mu_i\in\Sigma_i$ be arbitrary. Let $\mu$ and let $rank^{\mu}$ denote the concatenations of the $\mu_i$ and $rank_i^{\mu_i}$.
\begin{definition}\label{ORE}
For $\mu\in\Sigma$, the pair $(\sigma,M^b)$ with $\sigma\in\Sigma$ and $M^b\subseteq\Omega$ is a {\bf\em $\boldsymbol{\mu}$ Equilibrium} of $G$ if
\be\label{OREdef}
\sigma\,\in\,rank^{\mu}(\pi(\omega))\,=\,rank^{\mu}(\pi(\sigma))
\ee
for all $\omega\in M^b$. Let $E_{\mu}(G)$ denote the set of all $\mu$ equilibria of $G$.
\end{definition}

\vspace*{-4mm}
\begin{proposition}\label{EG}
$E_{\mu}(G)$ is non-empty for any $\mu\in\Sigma$ and any normal-form game $G$.
\end{proposition}

\vspace*{-6mm}
\begin{corollary}
There exists an $M$ equilibrium for any normal-form game $G$.
\end{corollary}
\vspace*{-2mm}

\noindent The typical approach is to formulate a model with few parameters that reflect the degree of \q{rationality.} Parsimony is imposed so that these rationality parameters can be estimated from the data without overfitting and used to out-of-sample predict behavior in similar games. We next show $\mu$ equilibrium is ideally suited for this approach and that it offers computational advantages over the commonly used logit-QRE. Recall that the latter is defined as
\begin{equation}\label{logitQRE}
  \sigma_{ij}\,=\,\frac{\exp(\lambda\pi_{ij}(\sigma_{-i}))}{\sum_{k\,=\,1}^{K_i}\exp(\lambda\pi_{ik}(\sigma_{-i}))}
\end{equation}
for $i\in N$, $1\leq j\leq K_i$, see \cite{McKelveyPalfrey1995}. QRE is a fixed-point model based on rational expectations, i.e. choices on the left match beliefs on the right. The logistic formulation in \eqref{logitQRE} is not the only possibility. Any set of \textit{regular} quantal responses, $R_i:\field{R}^{K_i}\rightarrow\Sigma_i$, that are interior, continuous, strictly increasing, and monotone in expected payoffs can be used to define a QRE as a fixed-point: $\sigma=R(\pi(\sigma))$, where $R$ and $\pi$ denote the concatenations of players' quantal responses and payoffs. The infinite-dimensional space of all regular quantal responses is denoted $\mathscr{R}$. For a typical element of $\mathscr{R}$, the QRE profile can only be computed numerically as the fixed-point conditions involve transcendental functions, see e.g. \eqref{logitQRE}.

For $i\in N$, let $rank_i^\varepsilon$ denote player $i$'s rank correspondence when $\mu_i=\mu_i(\varepsilon)$ as defined in the previous section. When $\varepsilon=0$, player $i$'s rank correspondence coincides with the best-reply correspondence, i.e. $rank_i^0(\pi_i)=\br_i(\pi_i)$. When $\varepsilon$ tends to 1, player $i$'s rank correspondence limits to $rank^1_i(\pi_i)=(1/K_i,\ldots,1/K_i)$ irrespective of payoffs. In other words, behavior ranges from fully rational to completely random as the rationality parameter $\varepsilon$ varies between 0 and 1. The $\varepsilon$ parameter thus plays a similar role as the rationality parameter $\lambda$ in logit-QRE.

\begin{figure}[t]
\begin{center}

\begin{tikzpicture}[scale=3.5]
\draw[line width=.5pt,gray] (0,0) -- (1,0) -- (1,1) -- (0,1) -- (0,0);
\draw[line width=1pt,dashed] (1/2,0) -- (1/2,1);
\draw[line width=1pt] (0,1/2) -- (1,1/2);
\node[scale=0.6pt] at (-0.03,-0.06) {$0$};
\node[scale=0.6pt] at (-0.03,0.5) {$\hf$};
\node[scale=0.6pt] at (-0.03,1) {$1$};
\node[scale=0.6pt] at (0.5,-0.06) {$\hf$};
\node[scale=0.6pt] at (1,-0.06) {$1$};
\node[scale=0.6pt] at (0.2,0.1) {$\varepsilon=1$};
\node[scale=0.6pt] at (0.75,-0.06) {$p$};
\node[scale=0.6pt] at (-0.03,0.75) {$q$};

\node at (1.8,0.46) {
\begin{tikzpicture}[scale=3.5]
\draw[line width=.5pt,gray] (0,0) -- (1,0) -- (1,1) -- (0,1) -- (0,0);
\draw[line width=1pt,dashed] (2/3,0) -- (2/3,8/9) -- (1/3,8/9) -- (1/3,1);
\draw[line width=1pt] (0,2/3) -- (4/5,2/3) -- (4/5,1/3) -- (1,1/3);
\node[scale=0.6pt] at (-0.03,-0.06) {$0$};
\node[scale=0.6pt] at (-0.03,0.5) {$\hf$};
\node[scale=0.6pt] at (-0.03,1) {$1$};
\node[scale=0.6pt] at (0.5,-0.06) {$\hf$};
\node[scale=0.6pt] at (1,-0.06) {$1$};
\node[scale=0.6pt] at (0.2,0.1) {$\varepsilon=\hf$};
\node[scale=0.6pt] at (0.75,-0.06) {$p$};
\node[scale=0.6pt] at (-0.03,0.75) {$q$};
\end{tikzpicture}};

\node at (3.1,0.46) {
\begin{tikzpicture}[scale=3.5]
\draw[line width=.5pt,gray] (0,0) -- (1,0) -- (1,1) -- (0,1) -- (0,0);
\draw[line width=1pt,dashed] (4/5,0) -- (4/5,8/9) -- (1/5,8/9) -- (1/5,1);
\draw[line width=1pt] (0,4/5) -- (4/5,4/5) -- (4/5,1/5) -- (1,1/5);
\node[scale=0.6pt] at (-0.03,-0.06) {$0$};
\node[scale=0.6pt] at (-0.03,0.5) {$\hf$};
\node[scale=0.6pt] at (-0.03,1) {$1$};
\node[scale=0.6pt] at (0.5,-0.06) {$\hf$};
\node[scale=0.6pt] at (1,-0.06) {$1$};
\node[scale=0.6pt] at (0.2,0.1) {$\varepsilon=\deel{1}{4}$};
\node[scale=0.6pt] at (0.75,-0.06) {$p$};
\node[scale=0.6pt] at (-0.03,0.75) {$q$};
\end{tikzpicture}};

\node at (0.48,-0.87) {
\begin{tikzpicture}[scale=3.5]
\draw[line width=.5pt,gray] (0,0) -- (1,0) -- (1,1) -- (0,1) -- (0,0);
\draw[line width=1pt,dashed] (6/7,0) -- (6/7,8/9) -- (1/7,8/9) -- (1/7,1);
\draw[line width=1pt] (0,6/7) -- (4/5,6/7) -- (4/5,1/7) -- (1,1/7);
\node[scale=0.6pt] at (-0.03,-0.06) {$0$};
\node[scale=0.6pt] at (-0.03,0.5) {$\hf$};
\node[scale=0.6pt] at (-0.03,1) {$1$};
\node[scale=0.6pt] at (0.5,-0.06) {$\hf$};
\node[scale=0.6pt] at (1,-0.06) {$1$};
\node[scale=0.6pt] at (0.2,0.1) {$\varepsilon=\deel{1}{6}$};
\node[scale=0.6pt] at (0.75,-0.06) {$p$};
\node[scale=0.6pt] at (-0.03,0.75) {$q$};
\end{tikzpicture}};

\node at (1.78,-0.87) {
\begin{tikzpicture}[scale=3.5]
\draw[line width=.5pt,gray] (0,0) -- (1,0) -- (1,1) -- (0,1) -- (0,0);
\draw[line width=1pt,dashed] (8/9,0) -- (8/9,8/9) -- (1/9,8/9) -- (1/9,1);
\draw[line width=1pt] (0,8/9) -- (4/5,8/9) -- (4/5,1/9) -- (1,1/9);
\node[scale=0.6pt] at (-0.03,-0.06) {$0$};
\node[scale=0.6pt] at (-0.03,0.5) {$\hf$};
\node[scale=0.6pt] at (-0.03,1) {$1$};
\node[scale=0.6pt] at (0.5,-0.06) {$\hf$};
\node[scale=0.6pt] at (1,-0.06) {$1$};
\node[scale=0.6pt] at (0.2,0.1) {$\varepsilon=\deel{1}{8}$};
\node[scale=0.6pt] at (0.75,-0.06) {$p$};
\node[scale=0.6pt] at (-0.03,0.75) {$q$};
\end{tikzpicture}};

\node at (3.08,-0.87) {
\begin{tikzpicture}[scale=3.5]
\draw[line width=.5pt,gray] (0,0) -- (1,0) -- (1,1) -- (0,1) -- (0,0);
\draw[line width=1pt,dashed] (1,0) -- (1,8/9) -- (0,8/9) -- (0,1);
\draw[line width=1pt] (0,1) -- (4/5,1) -- (4/5,0) -- (1,0);
\node[scale=0.6pt] at (-0.03,-0.06) {$0$};
\node[scale=0.6pt] at (-0.03,0.5) {$\hf$};
\node[scale=0.6pt] at (-0.03,1) {$1$};
\node[scale=0.6pt] at (0.5,-0.06) {$\hf$};
\node[scale=0.6pt] at (1,-0.06) {$1$};
\node[scale=0.6pt] at (0.2,0.1) {$\varepsilon=0$};
\node[scale=0.6pt] at (0.75,-0.06) {$p$};
\node[scale=0.6pt] at (-0.03,0.75) {$q$};
\end{tikzpicture}};

\node at (0.56,-2.2) {
\begin{tikzpicture}[scale=3.5]
\draw[line width=.5pt,gray] (0,0) -- (1,0) -- (1,1) -- (0,1) -- (0,0);

\begin{axis}[scale=0.175, axis line style={draw=none}, tick style={draw=none}, ticks=none, xmin=0, xmax=1.2, ymin=0, ymax=1]
\addplot[thin, smooth] plot coordinates
            {
                (0,1)
                (0.0760436, 0.999883)
                (0.121351, 0.998872)
                (0.182805, 0.992879)
                (0.239426, 0.980621)
                (0.261676, 0.974261)
                (0.292481, 0.96439)
                (0.309935, 0.95837)
                (0.321106, 0.954391)
                (0.335729, 0.949064)
                (0.345955, 0.945274)
                (0.358668, 0.940505)
                (0.364309, 0.938374)
                (0.371951, 0.935476)
                (0.382511, 0.931458)
                (0.410921, 0.920649)
                (0.436396, 0.911094)
                (0.453556, 0.904813)
                (0.473075, 0.897895)
                (0.493087, 0.891132)
                (0.516138, 0.883871)
                (0.536948, 0.877921)
                (0.559995, 0.872143)
                (0.57812, 0.868297)
                (0.605533, 0.863826)
                (0.634131, 0.861112)
                (0.660017, 0.86054)
                (0.688908, 0.862121)
                (0.713482, 0.865314)
                (0.736683, 0.869837)
                (0.760893, 0.876027)
                (0.8,0.88888)
            };
\end{axis}

\begin{axis}[scale=0.175, axis line style={draw=none}, tick style={draw=none}, ticks=none, xmin=0, xmax=1.2, ymin=0, ymax=1]
\addplot[thin, smooth] plot coordinates
            {
                (0.5, 0.5)
                (0.525533, 0.510291)
                (0.541837, 0.51613)
                (0.565403, 0.523442)
                (0.580626, 0.527394)
                (0.617481, 0.534169)
                (0.631916, 0.535657)
                (0.653433, 0.536576)
                (0.689411, 0.534504)
                (0.71134, 0.530991)
                (0.733677, 0.525677)
                (0.756468, 0.518493)
                (0.779658, 0.509407)
                (0.803055, 0.498472)
                (0.833955, 0.481334)
                (0.863433, 0.462013)
                (0.890249, 0.441606)
                (0.913495, 0.421216)
                (0.932826, 0.401678)
                (0.951756, 0.379138)
                (0.965738, 0.358914)
                (0.983081, 0.324685)
                (0.991812, 0.296681)
                (0.994342, 0.284321)
                (1, 0)
            };
\end{axis}

\node[scale=0.6pt] at (-0.03,-0.06) {$0$};
\node[scale=0.6pt] at (-0.03,0.5) {$\hf$};
\node[scale=0.6pt] at (-0.03,1) {$1$};
\node[scale=0.6pt] at (0.5,-0.06) {$\hf$};
\node[scale=0.6pt] at (1,-0.06) {$1$};
\node[scale=0.6pt] at (0.75,-0.06) {$p$};
\node[scale=0.6pt] at (-0.03,0.75) {$q$};
\end{tikzpicture}};

\node at (1.78,-2.2) {
\begin{tikzpicture}[scale=3.5]
\draw[line width=.5pt,gray] (0,0) -- (1,0) -- (1,1) -- (0,1) -- (0,0);
\draw[line width=1pt,mondrianYellow] (1/2,1/2) -- (4/5,4/5) -- (4/5,1/2);
\draw[line width=1pt,mondrianBlue] (4/5,1/2) -- (4/5,1/5) -- (1,0);
\draw[line width=1pt,mondrianYellow] (4/5,8/9) -- (1/2,8/9);
\draw[line width=1pt,mondrianRed] (1/2,8/9) -- (1/9,8/9) -- (0,1);
\node[scale=0.6pt] at (-0.03,-0.06) {$0$};
\node[scale=0.6pt] at (-0.03,0.5) {$\hf$};
\node[scale=0.6pt] at (-0.03,1) {$1$};
\node[scale=0.6pt] at (0.5,-0.06) {$\hf$};
\node[scale=0.6pt] at (1,-0.06) {$1$};
\node[scale=0.6pt] at (0.75,-0.06) {$p$};
\node[scale=0.6pt] at (-0.03,0.75) {$q$};
\end{tikzpicture}};

\node at (3.08,-2.2) {
\begin{tikzpicture}[scale=3.5]
\draw[line width=.5pt,gray] (0,0) -- (1,0) -- (1,1) -- (0,1) -- (0,0);
\filldraw[fill=mondrianYellow,draw=mondrianYellow,opacity=1] (0,0) -- (4/5,0) -- (4/5,8/9) -- (0,8/9) -- (0,0);
\filldraw[fill=mondrianBlue,draw=mondrianBlue,opacity=1] (4/5,8/9) -- (1,8/9) -- (1,0) -- (4/5,0) -- (4/5,8/9);
\filldraw[fill=mondrianRed,draw=mondrianRed,opacity=1] (0,1) -- (4/5,1) -- (4/5,8/9) -- (0,8/9) -- (0,1);
\node[scale=0.6pt] at (-0.03,-0.06) {$0$};
\node[scale=0.6pt] at (-0.03,0.5) {$\hf$};
\node[scale=0.6pt] at (-0.03,1) {$1$};
\node[scale=0.6pt] at (0.5,-0.06) {$\hf$};
\node[scale=0.6pt] at (1,-0.06) {$1$};
\node[scale=0.6pt] at (0.75,-0.06) {$\omega$};
\node[scale=0.6pt] at (-0.03,0.75) {$v$};
\end{tikzpicture}};

\end{tikzpicture}

\end{center}
\vspace*{-6mm}
\caption{The top two rows show Column's (dashed) and Row's (solid) $rank^\varepsilon$ correspondences for the asymmetric game of chicken of Figure \ref{newAGCfig} for various levels of $\varepsilon$ (indicated in the panels). The curve in the bottom-middle panel shows the $\mu$ equilibria for $\varepsilon\in[0,1]$ and its color reflects the set of beliefs that support it, see the bottom-right panel. The bottom-left panel shows the logit-QRE, for which the supporting belief is only the logit-QRE itself.}\label{AGCfigBR}
\vspace*{-2mm}
\end{figure}

An advantage of $\mu$ equilibrium is that the fixed-point condition \eqref{OREdef} does not involve transcendental functions. Consider the asymmetric game of chicken at the top of Figure~\ref{newAGCfig}. Let $p$ ($q$) denote the probability with which Column (Row) chooses $A$ and let $\mu_R=\mu_C=(1,\varepsilon)/(1+\varepsilon)$ where $\varepsilon\in[0,1]$ is the rationality parameter. The $rank^\varepsilon$ correspondences
\bd
rank^{\varepsilon}_R(p)\,=\,\left\{
\begin{array}{ccl}
\frac{1}{1+\varepsilon} &\mbox{if}& p<\deel{4}{5}\\[1mm]
[\frac{\varepsilon}{1+\varepsilon},\frac{1}{1+\varepsilon}] &\mbox{if}& p=\deel{4}{5}\nonumber\\[1mm]
\frac{\varepsilon}{1+\varepsilon} &\mbox{if}& p>\deel{4}{5}\nonumber
\end{array}\right.
\,\,\,\,\,\,\,
\mbox{and}
\,\,\,\,\,\,\,
rank^{\varepsilon}_C(q)\,=\,\left\{
\begin{array}{ccl}
\frac{1}{1+\varepsilon} &\mbox{if}& q<\deel{8}{9}\nonumber\\[1mm]
[\frac{\varepsilon}{1+\varepsilon},\frac{1}{1+\varepsilon}] &\mbox{if}& q=\deel{8}{9}\nonumber\\[1mm]
\frac{\varepsilon}{1+\varepsilon} &\mbox{if}& q>\deel{8}{9}\nonumber
\end{array}\right.
\ed
are ``flat,'' i.e. $rank^{\varepsilon}_R(p)=rank^{\varepsilon}_C(q)=\hf$, when $\varepsilon=1$, and they limit to standard best responses when $\varepsilon=0$. The top two rows of Figure~\ref{AGCfigBR} show Column's and Row's rank correspondences when $\varepsilon$ decreases from 1 (top left) to 0 (bottom right). Their intersection typically consists of an odd number of points (1 or 3) except at $\varepsilon=\deel{1}{4}$, when it contains a component, and at $\varepsilon=\deel{1}{8}$, when a bifurcation occurs and the intersection contains another component.

In the bottom-middle panel of Figure \ref{AGCfigBR}, the piecewise-linear curve shows the $\mu$ equilibria for $0\leq\varepsilon\leq 1$. For comparison, the logit-QRE for this game are shown in the bottom-left panel. There are similarities, i.e. for low rationality levels there is a unique equilibrium close to $p=q=\hf$ (random behavior), a bifurcation occurs for intermediate rationality levels, and the limit points for high rationality levels correspond to Nash equilibria. Importantly, both the logit-QRE and $\mu$ equilibria belong to an $M$-choice set shown in the bottom-left panel of Figure \ref{newAGCfig}, a result we explain in the next section.

There are also differences, most notably that $\mu$ equilibria can be analytically computed and are supported by sets of beliefs. In the bottom-middle panel, the curve's color indicates the set of beliefs that support the $\mu$ equilibrium, see the bottom-right panel. In contrast, the rational-expectations assumption underlying logit-QRE means it is supported only by itself.

While a single-parameter model provides a parsimonious model that is easy to estimate, it would be more realistic to assume that players' rationality parameters differ, i.e. $\varepsilon_R\neq\varepsilon_C$. This begs the question what choices occur under a heterogeneous $\mu$ equilibrium. It is straightforward to verify that any profile $(p,q)$ that belongs to an $M$-choice set is a $\mu$ equilibrium for some choice of $(\varepsilon_R,\varepsilon_C)$:
\begin{displaymath}
  (p,q)\,=\,\left\{\begin{array}{ccl}
  (\frac{1}{1+\varepsilon_C},\frac{\varepsilon_R}{1+\varepsilon_R}) & \text{if} & \varepsilon_C\leq\deel{1}{4}\\[1mm]
  (\frac{\varepsilon_C}{1+\varepsilon_C},\frac{1}{1+\varepsilon_R}) & \text{if} & \varepsilon_R\leq\deel{1}{8}\\[1mm]
  (\frac{1}{1+\varepsilon_C},\frac{1}{1+\varepsilon_R}) & \text{if} & \varepsilon_C\geq\deel{1}{4},\,\varepsilon_R\geq\deel{1}{8}\end{array}\right.
\end{displaymath}
For example, when $\varepsilon_C$ is varied between $\deel{1}{4}$ and 1 and $\varepsilon_R$ between $\deel{1}{8}$ and 1, the set of profiles in the bottom row produce the rectangle $[\hf,\deel{4}{5}]\times[\hf,\deel{8}{9}]$, which corresponds to the yellow $M$ set in Figure \ref{newAGCfig}. Likewise, the top and middle row correspond to the blue and red $M$-choice sets respectively. In other words, by varying $(\varepsilon_R,\varepsilon_C)$, the $\mu$ equilibria \q{fill out} the $M$-choice sets.

\subsection{$M$ Equilibrium as a Meta Theory}
\label{sec:Meta}

$M$-choice sets unify equilibria from various parametric models in the refinement and behavioral-game-theory literature.
\begin{proposition}\label{prop:meta}
For almost all $\sigma\in\Sigma_{int}$ the following statements are equivalent:
\begin{itemize}\addtolength{\itemsep}{-2mm}
\vspace*{-2mm}
\item[1.] $\sigma$ belongs to an $M$-choice set of $G$.
\item[2.] $\sigma$ is a $\mu$-equilibrium of $G$ for some regular $\mu\in\Sigma$.
\item[3.] $\sigma$ is an $\varepsilon$-proper equilibrium of $G$ for some $\varepsilon\in(0,1)$.
\item[4.] $\sigma$ is a Quantal Response Equilibrium of $G$ for some regular quantal response $R\in\mathscr{R}$.\footnote{Assuming generic games with non-thick indifference curves, see the proof in Appendix B.}
\end{itemize}
\end{proposition}
\noindent This equivalence result is surprising given the different origins of the models -- $\mu$ equilibria are Nash equilibria on restricted strategy sets, $\varepsilon$-proper equilibria are defined by inequalities, and QRE result from adding perturbations to expected payoffs. We briefly discuss these approaches.

\cite{Selten1975} was the first to consider robustness of Nash equilibria when players may tremble and all strategies have some small chance of being played. One way to formalize this is to restrict player $i$'s strategy space to an interior simplex, e.g. $\Sigma_i(\varepsilon)$ with vertices that are permutations of $\hat{\mu}_i=(1,\varepsilon,\ldots,\varepsilon)/(1+(K_i-1)\varepsilon)$ as in the bottom panels of Figure \ref{rank3}. The limit of Nash equilibria of such perturbed games then defines a (trembling-hand) perfect equilibrium. Regular $\mu$ equilibrium restricts strategy sets to an interior permutahedron, e.g. $P_i(\varepsilon)$ with vertices that are permutations of $\mu_i(\varepsilon)=c_\varepsilon(1,\varepsilon,\varepsilon^2,\ldots,\varepsilon^{K_i-1})$ as in the top panels of Figure \ref{rank3}. To illustrate the difference, consider the game in the left panel of Figure \ref{refinement}, which has three Nash equilibria: $(R,R)$, $(B,B)$, and $(Y,Y)$. Only the last two are perfect and only the last one is proper. It is readily verified that on the interior simplex both $(\varepsilon,1,\varepsilon)/(1+2\varepsilon)$ and $(\varepsilon,\varepsilon,1)/(1+2\varepsilon)$ are Nash equilibria. For $\varepsilon\in[0,1]$, these equilibria are shown by the blue lines in the middle panel of Figure \ref{refinement}. On the permutahedron, only $(\varepsilon^2,\varepsilon,1)/(1+\varepsilon+\varepsilon^2)$ is a Nash equilibrium as reflected by the green curve in the middle panel. Finally, the black curve in the middle panel shows the logit-QRE.

The middle panel shows the unique $M$-choice set in green (the right panel shows the supporting belief set). Note that $\mu$ equilibria and logit-QRE belong to this set, but not the (path leading to the) perfect equilibrium $(B,B)$: it is a best response to interior profiles that satisfy $\sigma_B\geq\sigma_R\geq\sigma_Y$, but the resulting expected payoffs are ordered $\pi_B>\pi_Y>\pi_R$. Monotonicity implies more costly mistakes are less likely, whence $(B,B)$ cannot be part of any $M$-choice set.

\begin{figure}[t]
\begin{center}

\begin{tikzpicture}[scale=3.6]
\draw[line width=.5pt,gray] (0,0) -- (1,0) -- (1/2,1/2*3^.5) -- (0,0);
\node[scale=0.8pt] at (-0.05,-0.05) {$R$};
\node[scale=0.8pt] at (1.05,-0.05) {$B$};
\node[scale=0.8pt] at (1/2,1/2*3^.5+0.05) {$Y$};
\filldraw[fill=mondrianGreen,draw=mondrianGreen,opacity=1] (1/2,1/2*3^.5) -- (1/2,1/6*3^.5) -- (3/4,1/4*3^.5) -- (1/2,1/2*3^.5);
\begin{axis}[scale=0.175, axis line style={draw=none}, tick style={draw=none}, ticks=none, xmin=0, xmax=1.2, ymin=0, ymax=1]
\addplot[thin, green, smooth] plot coordinates
            {
               (0.5, 0.866025)
               (0.540541, 0.780203)
               (0.564516, 0.698408)
               (0.57554, 0.62304)
               (0.576923, 0.555144)
               (0.571429, 0.494872)
               (0.561224, 0.44185)
               (0.547945, 0.395445)
               (0.532787, 0.354928)
               (0.516605, 0.319567)
               (0.5, 0.288675)
            };
\end{axis}
\draw[line width=1.25pt,mondrianBlue] (0.5, 0.288675) -- (0.5, 0.866025);
\begin{axis}[scale=0.175, axis line style={draw=none}, tick style={draw=none}, ticks=none, xmin=0, xmax=1.2, ymin=0, ymax=1]
\addplot[thin, smooth] plot coordinates
            {
              (0.5, 0.288675)
              (0.52836, 0.30822)
              (0.556731, 0.336564)
              (0.582273, 0.376725)
              (0.600668, 0.431172)
              (0.60759, 0.499086)
              (0.602123, 0.573303)
              (0.588429, 0.642609)
              (0.572271, 0.699427)
              (0.557272, 0.742616)
              (0.54471, 0.774424)
              (0.534663, 0.79767)
              (0.526797, 0.814704)
              (0.520696, 0.82727)
              (0.515985, 0.836608)
              (0.512353, 0.843597)
              (0.509552, 0.848862)
              (0.507392, 0.85285)
              (0.505725, 0.855886)
              (0.504438, 0.858205)
              (0.503442, 0.859983)
            };
\end{axis}
\draw[line width=1.25pt,mondrianBlue](0.5, 0.288675) -- (1,0);

\node at (2,0.43) {
\begin{tikzpicture}[scale=3.6]
\draw[line width=.5pt,gray] (0,0) -- (1,0) -- (1/2,1/2*3^.5) -- (0,0);
\node[scale=0.8pt] at (-0.05,-0.05) {$R$};
\node[scale=0.8pt] at (1.05,-0.05) {$B$};
\node[scale=0.8pt] at (1/2,1/2*3^.5+0.05) {$Y$};
\filldraw[fill=mondrianGreen,draw=mondrianGreen,opacity=1] (1/4,1/4*3^.5) -- (1,0) -- (1/2,1/2*3^.5) -- (1/4,1/4*3^.5);
\end{tikzpicture}};

\node at (-0.9,0.35) {
\begin{tabular}{c|ccc}
& $R$ & $B$ & $Y$ \\ \cline{1-4}
\multicolumn{1}{r}{\rule{0pt}{5mm}$R$} & \multicolumn{1}{|c}{$1$, $1$} & \multicolumn{1}{c}{$1$, $1$} & \multicolumn{1}{c}{$1$, $0$} \\
\multicolumn{1}{r}{$B$} & \multicolumn{1}{|c}{$1$, $1$} & \multicolumn{1}{c}{$2$, $2$} & \multicolumn{1}{c}{$2$, $2$} \\
\multicolumn{1}{r}{$Y$} & \multicolumn{1}{|c}{$0$, $1$} & \multicolumn{1}{c}{$2$, $2$} & \multicolumn{1}{c}{$3$, $3$} \\
\end{tabular}};

\end{tikzpicture}
\end{center}
\vspace*{-6mm}

\caption{The top-left panel shows a game from \cite{BBD1991} with three Nash equilibria: $(R,R)$, $(B,B)$, and $(Y,Y)$. Only the last two are perfect and only the last one is proper.}\label{refinement}
\vspace*{-4mm}
\end{figure}

\cite{Myerson1978} provides a \q{logical} definition for perfectness: $\sigma$ is perfect if it is the limit of $\varepsilon$-perfect equilibria, i.e. $\sigma(\varepsilon)\in\Sigma_{int}$ such that $\pi_{ij}(\sigma(\varepsilon))<\pi_{ik}(\sigma(\varepsilon))$ implies $\sigma_{ij}(\varepsilon)\leq\varepsilon$. \cite{Myerson1978} also proposes \textit{properness} as an alternative refinement:\footnote{See \cite{vanDamme1996} for an excellent survey of the refinement literature.} $\sigma$ is proper if it is the limit of $\varepsilon$-proper equilibria, i.e. $\sigma(\varepsilon)\in\Sigma_{int}$ such that $\pi_{ij}(\sigma(\varepsilon))<\pi_{ik}(\sigma(\varepsilon))$ implies $\sigma_{ij}(\varepsilon)\leq\varepsilon\sigma_{ik}(\varepsilon)$. Nash equilibria on the permutahedron $P_i(\varepsilon)$ with vertices that are permutations of $\mu_i(\varepsilon)$ defined above are examples of $\varepsilon$-proper equilibria. The $M$-choice sets unify all $\varepsilon$-proper equilibria and contain all proper Nash equilibria.\footnote{Proposition \ref{prop:meta} shows that $\varepsilon$-proper equilibria are contained in some $M$-choice set, which is closed, and, hence contains the limit $\sigma$.}  Interestingly, $M$ equilibrium also suggests a refinement criterium that falls between perfectness and properness, see Section \ref{sec:Nash}.

\citeauthor{Harsanyi1973}'s (\citeyear{Harsanyi1973}) work on \textit{purification,} i.e. interpreting pure and mixed Nash equilibria as limits of pure-strategy Bayes Nash equilibria of games with randomly perturbed payoffs, underlies \citeauthor{McKelveyPalfrey1995}'s (\citeyear{McKelveyPalfrey1995}) Quantal Response Equilibrium.\footnote{\cite{Harsanyi1973} considers separate perturbations for each of the $\prod_i |S_i|$ payoffs while QRE is based on perturbations of the $\sum_i|S_i|$ expected payoffs.} \q{Integrating out} the random perturbations yields \q{quantal responses,} $R_i:\field{R}^{K_i}\rightarrow\Sigma_i$ for $i\in N$, that map expected payoffs to choice probabilities. The quantal responses are regular if the payoff perturbations are interchangeable, which is a weaker condition than i.i.d., see \cite{goeree2005}. The space of all regular quantal responses, $\mathscr{R}$, is infinite dimensional. Yet the $M$-choice set unifies all the fixed-points that can be obtained using different elements of this space (e.g. logit, probit, etc.).

Importantly, $M$ equilibrium is \textit{not} an example of a QRE as it is set-valued while QRE is point-valued. And the union of different QREs is \textit{not} a QRE. For example, suppose one data set fits logit-QRE but not probit-QRE while for another data set the opposite is true. Then both data sets can be explained by $M$ equilibrium but not by QRE. (See also Section \ref{subsec:noLogit}.)

The number of parameters underlying the different models in Proposition \ref{prop:meta} differs sharply (e.g. one for $\varepsilon$-proper and infinitely many for QRE). In contrast, $M$ equilibrium is parameter free, i.e. in the first item of Proposition \ref{prop:meta} there is no \q{for-some-parameter-value(s)} qualifier. Interestingly, while the computation of Nash equilibria on permutahedra, or of QRE, or of $\varepsilon$-proper equilibria can be cumbersome, determining the set of all of them follows straightforwardly by matching rankings of payoffs and choice probabilities.

The main distinction, however, between $M$ equilibrium and the other models is that it relaxes the stringent rational-expectations condition of correct beliefs. Just like $M$-choice sets extend \citeauthor{Selten1975}'s (\citeyear{Selten1975}) idea of ``trembles'' by allowing for sizeable deviations, $M$-belief sets generalize \citeauthor{Selten1975}'s notion of robustness -- rather than requiring a Nash-equilibrium profile to be a best response to an infinitesimal path of beliefs converging to it, elements in the $M$-choice set are better responses to an entire set of beliefs that contains the $M$-choice set.

\subsection{Profect Nash Equilibria in $\overline{\mathcal{M}}(G)$}
\label{sec:Nash}

Monotonicity implies more costly mistakes are less likely, which holds for proper equilibria but not necessarily for perfect equilibria. As a result, any proper equilibrium is contained in some $M$-choice set while perfect equilibria are not necessarily part of $\overline{\mathcal{M}}(G)$.
However, monotonicity does \textit{not} require the more stringent condition that more costly mistakes are \textit{infinitely} less likely as implied by properness. Consider the game in the top-left panel of Figure \ref{tab:BBD}, which has three pure-strategy Nash equilibria: $(R,R)$, $(B,B)$, and $(Y,Y)$, and a mixed-strategy Nash equilibrium $\sigma=(0.05,0.9,0.05)$. None of the pure-strategy equilibria are proper, including $(R,R)$.\footnote{To see that $(R,R)$ is not proper, note that a $Y$ tremble is far worse than a $B$ tremble so the trembling choice profile would have to be $\sigma(\varepsilon)=(1,\varepsilon,\varepsilon^2)$ appropriately normalized. But then the expected payoff of $B$ exceeds that of $R$ for sufficiently small $\varepsilon$.} The only proper equilibrium is the mixed-strategy equilibrium, which is supported by a \textit{single} belief, i.e. the mixed-strategy profile itself. Intuitively, this makes the proper equilibrium empirically irrelevant given that $(R,R)$ is supported by a large full-dimensional belief set, see the top-right panel.
In this example, $(R,R)$ is a \textit{profect equilibrium}.\footnote{In Latin, the noun \q{profect} is the doublet of profit and \q{to profect} means to benefit, profit, or advance, see \url{https://en.wiktionary.org/wiki/profect}.}

\begin{figure}[t]
\begin{center}

\begin{tikzpicture}[scale=3.6]
\draw[line width=.5pt,gray] (0,0) -- (1,0) -- (1/2,1/2*3^.5) -- (0,0);
\node[scale=0.8pt] at (-0.05,-0.05) {$R$};
\node[scale=0.8pt] at (1.05,-0.05) {$B$};
\node[scale=0.8pt] at (1/2,1/2*3^.5+0.05) {$Y$};
\filldraw[fill=mondrianRed,draw=mondrianRed,opacity=1] (0,0) -- (1/2,1/6*3^.5) -- (1/2,3^.5/74) -- (0,0);
\filldraw[fill=mondrianBlue,draw=mondrianBlue,opacity=1] (1/2,0) -- (1/2,3^.5/74) -- (37/40,3^.5/40) -- (18/19,0) -- (1/2,0);
\draw[line width=1pt,black] (0,0) -- (37/40,3^.5/40);
\draw[line width=1pt,black] ((37/40,3^.5/40) -- (18/19,0);

\node at (2,0.43) {
\begin{tikzpicture}[scale=3.6]
\draw[line width=.5pt,gray] (0,0) -- (1,0) -- (1/2,1/2*3^.5) -- (0,0);
\node[scale=0.8pt] at (-0.05,-0.05) {$R$};
\node[scale=0.8pt] at (1.05,-0.05) {$B$};
\node[scale=0.8pt] at (1/2,1/2*3^.5+0.05) {$Y$};
\filldraw[fill=mondrianRed,draw=mondrianRed,opacity=1] (0,0) -- (37/40,3^.5/40) -- (1/4,1/4*3^.5) -- (0,0);
\filldraw[fill=mondrianBlue,draw=mondrianBlue,opacity=1] (0,0) -- (37/40,3^.5/40) -- (18/19,0) -- (0,0);
\draw[line width=1pt,black] (0,0) -- (37/40,3^.5/40);
\draw[line width=1pt,black] ((37/40,3^.5/40) -- (18/19,0);
\end{tikzpicture}};

\node at (-0.9,0.35) {
\begin{tabular}{c|ccc}
& $R$ & $B$ & $Y$ \\ \cline{1-4}
\multicolumn{1}{r}{\rule{0pt}{5mm}$R$} & \multicolumn{1}{|c}{$1$, $1$} & \multicolumn{1}{c}{$0$, $1$} & \multicolumn{1}{c}{$1$, $-17$} \\
\multicolumn{1}{r}{$B$} & \multicolumn{1}{|c}{$1$, $0$} & \multicolumn{1}{c}{$1$, $1$} & \multicolumn{1}{c}{$-17$, $1$} \\
\multicolumn{1}{r}{$Y$} & \multicolumn{1}{|c}{$-17$, $1$} & \multicolumn{1}{c}{$1$, $-17$} & \multicolumn{1}{c}{$1$, $1$} \\
\end{tabular}};

\node at (0.5,0.43-1.05) {
\begin{tikzpicture}[scale=3.6]
\draw[line width=.5pt,gray] (0,0) -- (1,0) -- (1/2,1/2*3^.5) -- (0,0);
\node[scale=0.8pt] at (-0.05,-0.05) {$R$};
\node[scale=0.8pt] at (1.05,-0.05) {$B$};
\node[scale=0.8pt] at (1/2,1/2*3^.5+0.05) {$Y$};
\filldraw[fill=mondrianRed,draw=mondrianRed,opacity=1] (0,0) -- (1/2,1/6*3^.5) -- (1/2,0) -- (0,0);
\filldraw[fill=mondrianCyan,draw=mondrianCyan,opacity=1] (1,0) -- (3/4,1/4*3^.5)  -- (3/5,1/5*3^.5) -- (1,0);
\end{tikzpicture}};

\node at (2,0.43-1.05) {
\begin{tikzpicture}[scale=3.6]
\draw[line width=.5pt,gray] (0,0) -- (1,0) -- (1/2,1/2*3^.5) -- (0,0);
\node[scale=0.8pt] at (-0.05,-0.05) {$R$};
\node[scale=0.8pt] at (1.05,-0.05) {$B$};
\node[scale=0.8pt] at (1/2,1/2*3^.5+0.05) {$Y$};
\filldraw[fill=mondrianRed,draw=mondrianRed,opacity=1] (0,0) -- (1,0) -- (1/4,1/4*3^.5) -- (0,0);
\filldraw[fill=mondrianCyan,draw=mondrianCyan,opacity=1] (1,0) -- (1/3,1/3*3^.5)  -- (1/2,1/2*3^.5) -- (1,0);
\end{tikzpicture}};

\node at (-0.9,0.35-1.05) {
\begin{tabular}{c|ccc}
& $R$ & $B$ & $Y$ \\ \cline{1-4}
\multicolumn{1}{r}{\rule{0pt}{5mm}$R$} & \multicolumn{1}{|c}{$2$, $2$} & \multicolumn{1}{c}{$1$, $1$} & \multicolumn{1}{c}{$0$, $0$} \\
\multicolumn{1}{r}{$B$} & \multicolumn{1}{|c}{$1$, $1$} & \multicolumn{1}{c}{$1$, $1$} & \multicolumn{1}{c}{$1$, $1$} \\
\multicolumn{1}{r}{$Y$} & \multicolumn{1}{|c}{$0$, $0$} & \multicolumn{1}{c}{$1$, $1$} & \multicolumn{1}{c}{$1$, $1$} \\
\end{tabular}};

\end{tikzpicture}
\end{center}
\vspace*{-5mm}

\caption{For the game in the top-left panel, $(R,R)$, $(B,B)$, and $(Y,Y)$ are Nash equilibria. None are proper, only the mixed equilibrium $\sigma=(0.05,0.9,0.05)$ is, while $(R,R)$ is a \textit{profect} Nash equilibrium. For the game in the bottom-left panel (from \citeauthor{vanDamme1996}, \citeyear{vanDamme1996}) the simplex edge with $\sigma_R=0$ forms a Nash-equilibrium component. Only $(R,R)$ and $(B,B)$ are proper. The middle and right panels show the $M$-choice and $M$-belief sets respectively.}\label{tab:BBD}
\vspace*{-2mm}
\end{figure}

\begin{definition}\label{def:profect}
$\sigma\in\Sigma_{int}$ is an $\boldsymbol{\varepsilon}$\textbf{-profect equilibrium} if $\pi_{ij}(\sigma)<\pi_{ik}(\sigma)\Rightarrow\sigma_{ij}<\min(\varepsilon,\sigma_{ik})$ for all $i\in N$, $1\leq j,k\leq|S_i|$.
$\sigma\in\Sigma$ is a \textbf{profect equilibrium} iff there exist $\varepsilon_k\in(0,1)$ and $\sigma_k\in\Sigma_{int}$ for $k\in\field{N}$ such that $\lim_{k\rightarrow\infty}\varepsilon_k=0$ and $\sigma_k$ is $\varepsilon_k$-profect and $\lim_{k\rightarrow\infty}\sigma_k=\sigma$.
\end{definition}
Profectness implies more costly mistakes are less likely but not necessarily infinitely so. Let $\mathcal{N}(G)$ be the set of Nash equilibria of $G$ and $\mathcal{N}_s(G)$ the subset satisfying selection criterion $s$.
\begin{proposition}\label{prop:incl}
For any $G$, $\,\mathcal{N}(G)\,\supseteq\,\mathcal{N}_{per\!f\hspace*{-0.5pt}ect}(G)\,\supseteq\,\mathcal{N}_{pro\hspace*{-0.7pt}f\hspace*{-0.7pt}ect}(G)\,\supseteq\,\mathcal{N}_{proper}(G)$. The inclusions may be strict.
\end{proposition}
Existence of a profect equilibrium thus follows from existence of a proper equilibrium.

Profect equilibria can be found geometrically using the methods of Section \ref{sec:rankChar}. For $r\in\mathcal{P}$, let $\overline{M}^c_r$ denote the $M$-equilibrium choice set. Profiles in this set are better responses in the sense that choice probabilities are ranked the same as expected payoffs but they are not necessarily best responses. To find the best responses in $\overline{M}^c_r$, recall from Section \ref{sec:rankChar} that they lie in a face of the simplex. This simplex face, denoted $\Sigma_r$, is determined by $r$ via the best-reply mapping, i.e. $\Sigma_r=\br(r)$:
\begin{displaymath}
\mathcal{N}_{pro\hspace*{-0.7pt}f\hspace*{-0.7pt}ect}(G)\,=\,\sqcup_{r\,\in\,\mathcal{P}}\,\overline{M}^c_r\cap\Sigma_r
\end{displaymath}
For example, for the game in the top-left panel of Figure \ref{tab:BBD}, the red $M$-choice set contains the profect (and perfect, but not proper) equilibrium $(R,R)$. For the blue $M$-choice set payoffs are ranked $\pi_B>\pi_R>\pi_Y$ but its intersection with the $B$ vertex is empty. For a totally-mixed strategy, the simplex' face selected by the best-reply correspondence is the simplex itself. Hence, $\sigma=(0.05,0.9,0.05)$ is also a profect (and proper) equilibrium.

Finally, when a payoff indifference curve lies in the simplex' boundary, non-perfect Nash equilibria may be added on closing $M$ sets. For example, for the game in the bottom-left panel of Figure \ref{tab:BBD}, any profile $(0,\sigma,1-\sigma)$ is a Nash equilibrium and those with $\sigma\geq\hf$ are included when closing the cyan-colored $M$ set. Of these, only $\sigma=(0,1,0)$ is profect (and proper).

\subsection{Robust $M$ Equilibria}
\label{sec:robust}

Denote by $\Gamma_i=\field{R}^{|S|}$ the space of payoffs of player $i$ and let $\Gamma=\prod_{i=1}^n\Gamma_i$. Let $G(\varepsilon)\subset\Gamma$ denote the set of games that result by perturbing any of the payoffs of $G$ by at most $\varepsilon>0$.
\begin{definition}\label{def:BS}
$(\sigma,\omega)\in\Sigma\times\Omega$ is a {\bf\em robust profile} of $G$ if there exists $\varepsilon>0$ such that $(\sigma,\omega)\in\overline{\mathcal{M}}(G')$ for all $G'\in G(\varepsilon)$. $M\in\overline{\mathcal{M}}(G)$ is {\bf\em robust} if its interior consists of robust profiles.
\end{definition}
Robustness is akin to \citeauthor{WuJiang1962}'s (\citeyear{WuJiang1962}) notion of \q{essentiality,} which is also based on arbitrary payoff perturbations. A Nash profile is essential if the perturbed game has a Nash equilibrium ``close'' to it. We relax the best-response assumption to monotonicity, but sharpen the requirement that the profile belongs to an $M$ equilibrium of the perturbed game (rather than being ``close'' to an $M$ equilibrium of the perturbed game).

The proof of the next proposition provides a simple genericity condition under which full-dimensional $M$ sets exist.
\begin{proposition}\label{gen1}
There exists an open and dense set $\mathcal{O}\subset\Gamma$ such that for $G\in\mathcal{O}$:
\begin{itemize}\addtolength{\itemsep}{-2mm}
\vspace*{-2mm}
\item[(i)] A full dimensional, colorable, and robust $M$ equilibrium exists;
\item[(ii)] The measure of each $M$-choice set is bounded by $1/\prod_{i\in N}K_i!$;
\item[(iii)] The total measure of all $M$-choice sets combined is bounded by $1/\max_{i\in N}K_i!$;
\item[(iv)] In contrast, an $M$-belief set may have full measure.
\end{itemize}
\end{proposition}
\noindent The fact that the measure of any full-dimensional $M$-choice set falls factorially fast with the number of players and the number of strategies shows $M$ equilibrium is falsifiable. Moreover, combining the results in Propositions \ref{prop:meta} and \ref{gen1}, implies that the set of all regular QRE, i.e. all fixed-points obtained by varying over an infinite-dimensional space of regular quantal responses, is bounded by $1/\max_{i\in N}K_i!$. In other words, regular QRE is falsifiable.

This finding contrasts with \citeauthor{HaileHortacsuKosenok2008}'s (\citeyear{HaileHortacsuKosenok2008}) critique about the (lack of) empirical content of QRE. Their construction, though, relies on creating anti-monotonic behavior by assuming payoff perturbations that are not interchangeable.\footnote{To illustrate, consider a simple decision task in which subjects are asked to choose between option $A$ that pays \$0 and option $B$ that pays \$1. Now suppose these payoffs are perturbed by adding the mean-zero random variables $\varepsilon_A=2$ with probability $p\in(0,1)$ and $\varepsilon_A=-2p/(1-p)$ with probability $1-p$, while $\varepsilon_B=0$ for sure. Then the inferior option $A$ is chosen with (arbitrary) probability $p\in(0,1)$. See also \cite{goeree2005}.} When perturbations are interchangeable, monotonicity holds and renders the critique by \cite{HaileHortacsuKosenok2008} vacuous.

\begin{figure}[t]
\begin{center}

\begin{tikzpicture}[scale=3.6]
\draw[line width=.5pt,gray] (0,0) -- (1,0) -- (1/2,1/2*3^.5) -- (0,0);
\node[scale=0.8pt] at (-0.05,-0.05) {$R$};
\node[scale=0.8pt] at (1.05,-0.05) {$B$};
\node[scale=0.8pt] at (1/2,1/2*3^.5+0.05) {$Y$};
\filldraw[fill=mondrianGrey,draw=mondrianGrey,opacity=1] (1/4,1/4*3^.5) -- (1/2,1/2*3^.5) -- (3/4,1/4*3^.5) -- (1/2,1/6*3^.5) -- (1/4,1/4*3^.5);
\filldraw[fill=mondrianGrey,draw=mondrianGrey,opacity=1] (0,0) -- (1/2,0)  -- (1/4,1/12*3^.5) -- (0,0);
\draw[line width=1.5pt,black] (1/2,0) -- (1/8,1/8*3^.5);

\node at (2,0.43) {
\begin{tikzpicture}[scale=3.6]
\draw[line width=.5pt,gray] (0,0) -- (1,0) -- (1/2,1/2*3^.5) -- (0,0);
\node[scale=0.8pt] at (-0.05,-0.05) {$R$};
\node[scale=0.8pt] at (1.05,-0.05) {$B$};
\node[scale=0.8pt] at (1/2,1/2*3^.5+0.05) {$Y$};
\filldraw[fill=mondrianGrey,draw=mondrianGrey,opacity=1] (1/8,1/8*3^.5) -- (1/2,0) -- (1,0) -- (1/2,1/2*3^.5) -- (1/8,1/8*3^.5);
\filldraw[fill=mondrianGrey,draw=mondrianGrey,opacity=1] (0,0) -- (1/2,0)  -- (1/8,1/8*3^.5) -- (0,0);
\draw[line width=1.5pt,black] (1/2,0) -- (1/8,1/8*3^.5);
\end{tikzpicture}};

\node at (-0.9,0.35) {
\begin{tabular}{c|ccc}
& $R$ & $B$ & $Y$ \\ \cline{1-4}
\multicolumn{1}{r}{\rule{0pt}{5mm}$R$} & \multicolumn{1}{|c}{$3$, $3$} & \multicolumn{1}{c}{$2$, $3$} & \multicolumn{1}{c}{$1$, $2$} \\
\multicolumn{1}{r}{$B$} & \multicolumn{1}{|c}{$3$, $2$} & \multicolumn{1}{c}{$2$, $2$} & \multicolumn{1}{c}{$1$, $3$} \\
\multicolumn{1}{r}{$Y$} & \multicolumn{1}{|c}{$2$, $1$} & \multicolumn{1}{c}{$3$, $1$} & \multicolumn{1}{c}{$4$, $4$} \\
\end{tabular}};

\end{tikzpicture}
\end{center}
\vspace*{-5mm}
\caption{A non-generic game with \q{cloned} strategies $R$ and $B$ that has full-dimensional but non-colorable $M$ equilibrium sets.}\label{fig:nonGeneric}
\vspace*{-3mm}
\end{figure}

Generically, full-dimensionality, colorability, and robustness are equivalent -- they all reflect the same \textit{strict} ordering of choices and expected payoffs.  In non-generic games, however, full-dimensionality does not imply colorability. Consider, for instance, the game in the left panel of Figure \ref{fig:nonGeneric}, for which $\pi_R=\pi_B>\pi_Y$ for $\sigma_R>\sigma_Y+\hf$ and  $\pi_R=\pi_B<\pi_Y$ for $\sigma_R<\sigma_Y+\hf$. The two $M$-choice sets shown in grey are full-dimensional but not colorable since the ordering of choice probabilities varies over the $M$-choice set and does not everywhere match the ordering of expected payoffs. (Another non-colorable $M$-choice set is shown by the black line and consists of a component of proper equilibria that satisfy $\pi_R=\pi_B=\pi_Y$.)

Nor does colorability imply full-dimensionality in non-generic games. For example, for a matching-pennies game, $\overline{\mathcal{M}}(G)$ consists of a single colorable profile: the unique Nash equilibrium in which both players randomize uniformly. Perturbing the matching-pennies game results in a game with a full-dimensional $M$ set that includes the barycenter, so the unique $M$ equilibrium is robust.
Likewise, one of the non-colorable $M$-choice sets for the game in the left panel of Figure \ref{fig:nonGeneric} contains a robust subset of diagonal profiles $(\alpha,\alpha,1-2\alpha)$, for $0\leq\alpha\leq\deel{1}{3}$, that are colorable as the equality in payoffs is matched by equality of choice probabilities.

\section{An Experimental Test of $M$ Equilibrium}
\label{sec:exp-test}

We report the results from a series of experiments to illustrate how $M$ equilibrium provides a lens through which to better understand strategic behavior in games. The experiments shared some common features. Sixteen participants joined each experimental session. Subjects were first given instructions in a power-point presentation (read aloud). Subjects played two-player matrix games with two or three possible choices. Each game was played for 8 rounds ($2\times 2$ games) or 15 rounds ($3\times 3$ games). In each round, participants were randomly rematched with a different participant in a perfect stranger protocol to ensure that two players never played the same game together more than once. This feature was made explicitly clear to subjects.

On the screen that displayed the payoff matrix, subjects could select a row of the matrix\footnote{All subjects played by choosing a row. For asymmetric games this simply means that Column played as Row with a transposed payoff matrix.} and submit their beliefs about their opponent's choice in terms of ``percentage chances.'' Belief elicitation was incentivized using a generalization of a method proposed by \citet{wilson2018}, which is an implementation of \citeauthor{hossain2013}'s (\citeyear{hossain2013}) \emph{binarized scoring rule} (BSR). The BSR is incentive compatible for general risk-preferences and, hence, avoids issues of risk-aversion that plague other scoring rules (e.g. quadratic scoring rule).\footnote{See \citet{schotter2014} and \citet{schlag2015} for discussions of belief elicitation methods.} The method operationalizes the BSR for binary-choice settings in a way that is simple to explain.\footnote{See Appendix C for details and a generalization of the method to games with more than two choices.}

After all subjects had submitted their choices and beliefs, they were shown their opponent's choice, and the results of the belief elicitation task. To avoid hedging, their payoff in each round was randomly selected to be either their payoff from the game or their payoff from the belief elicitation task. At the end of the experiment, subjects were informed about their total earnings and paid in cash.

\subsection{Comparative Statics in Asymmetric Matching Pennies Games}
\label{subsec:AMP}

Besides providing a better fit to data from laboratory experiments, behavioral-game-theory models often yield more sensible comparative statics than Nash equilibrium. In many games, Nash predictions are unaffected by (non-critical) changes in game parameters, e.g. the number of sellers in a Bertrand game or the effort-cost in a minimum-effort coordination game,\footnote{See e.g. Table 2 in \cite{AndersonGoereeHolt2002} for more examples.} while such changes can have a dramatic effect on observed behavior (e.g. \citeauthor{CapraGoereeGomezHolt1999}, \citeyear{CapraGoereeGomezHolt1999}). The experiments reported in this section test comparative statics when not only the Nash equilibrium but also the behavioral-game-theory models predict identical behavior across games.

Consider the $2\times2$ asymmetric matching pennies (AMP) games in Table~\ref{AMPexp}, which are all derived from the same parametric form shown in the top-left panel. The use of this common parametric form guarantees that the best response structure is identical across games. Hence, all games in this family have the same unique mixed-strategy Nash equilibrium. Let $p$ and $q$ denote the probability with which Column and Row choose $A$ respectively. The Nash equilibrium for any of the games in Table \ref{AMPexp} is $(p^*,q^*)=(\deel{1}{2},\deel{1}{6})$.\footnote{The study of such games has been instrumental in the development of alternative models of strategic behavior. See, for instance, \citet{ochs1995, erev1998, goeree2003, selten2008}, as well as the comment by \citet{brunner2011} and the reply by \citet{selten2011}.} The parametric form we use also implies that logit-QRE predictions are the same in all games. Since logit-QRE is based on expected payoff differences, i.e. $\sigma_{ij}/\sigma_{ik}=\exp(\lambda(\pi_{ij}-\pi_{ik}))$ for $1\leq j,k\leq K_i$, $i\in N$, adding the same amount to the payoff numbers of both options (the $W$, $X$, $Y$, and $Z$ in the top-left panel of Table~\ref{AMPexp}) ensures that logit-QRE predictions are unaffected. The $M$-choice and $M$-belief sets, which are also the same for all games, are shown in the left and right panels of Figure~\ref{expdataAMP} respectively, along with the observed choice and belief averages for each of the five games.\footnote{The graphs in Figure~\ref{expdataAMP} use the data from all rounds in the experiment. Restricting the data to only the second half of the experiment, or just the last round, results in a very similar picture and results.}

\begin{table}[t]
\begin{center}
\begin{tabular}{lcl}
\begin{tabular}{ccc}
\multicolumn{1}{r}{AMP}& \multicolumn{1}{|c}{$A$} & \multicolumn{1}{c}{$B$}  \\ \cline{1-3}
\multicolumn{1}{r}{\rule{0pt}{4mm}$A$} & \multicolumn{1}{|c}{$X+10$, $Z$} & \multicolumn{1}{c}{$W$, $Z+50$} \\
\multicolumn{1}{r}{\rule{0pt}{4mm}$B$} & \multicolumn{1}{|c}{$X$, $Y+10$} & \multicolumn{1}{c}{$W+10$, $Y$} \\
\end{tabular}

&\hspace*{1cm}&

\begin{tabular}{ccc}
\multicolumn{1}{r}{AMP1}& \multicolumn{1}{|c}{$A$} & \multicolumn{1}{c}{$B$}  \\ \cline{1-3}
\multicolumn{1}{r}{\rule{0pt}{4mm}$A$} & \multicolumn{1}{|c}{$20$, $10$} & \multicolumn{1}{c}{$10$, $60$} \\
\multicolumn{1}{r}{\rule{0pt}{4mm}$B$} & \multicolumn{1}{|c}{$10$, $20$} & \multicolumn{1}{c}{$20$, $10$} \\
\end{tabular}\\[8mm]

\begin{tabular}{ccc}
\multicolumn{1}{r}{AMP2}& \multicolumn{1}{|c}{$A$} & \multicolumn{1}{c}{$B$}  \\ \cline{1-3}
\multicolumn{1}{r}{\rule{0pt}{4mm}$A$} & \multicolumn{1}{|c}{$60$, $10$} & \multicolumn{1}{c}{$10$, $60$} \\
\multicolumn{1}{r}{\rule{0pt}{4mm}$B$} & \multicolumn{1}{|c}{$50$, $20$} & \multicolumn{1}{c}{$20$, $10$} \\
\end{tabular}
&\hspace*{1cm}&
\begin{tabular}{ccc}
\multicolumn{1}{r}{AMP3}& \multicolumn{1}{|c}{$A$} & \multicolumn{1}{c}{$B$}  \\ \cline{1-3}
\multicolumn{1}{r}{\rule{0pt}{4mm}$A$} & \multicolumn{1}{|c}{$60$, $50$} & \multicolumn{1}{c}{$10$, $100$} \\
\multicolumn{1}{r}{\rule{0pt}{4mm}$B$} & \multicolumn{1}{|c}{$50$, $20$} & \multicolumn{1}{c}{$20$, $10$} \\
\end{tabular}\\[8mm]

\begin{tabular}{ccc}
\multicolumn{1}{r}{AMP4}& \multicolumn{1}{|c}{$A$} & \multicolumn{1}{c}{$B$}  \\ \cline{1-3}
\multicolumn{1}{r}{\rule{0pt}{4mm}$A$} & \multicolumn{1}{|c}{$60$, $50$} & \multicolumn{1}{c}{$10$, $100$} \\
\multicolumn{1}{r}{\rule{0pt}{4mm}$B$} & \multicolumn{1}{|c}{$50$, $60$} & \multicolumn{1}{c}{$20$, $50$} \\
\end{tabular}
&\hspace*{1cm}&
\begin{tabular}{ccc}
\multicolumn{1}{r}{AMP5}& \multicolumn{1}{|c}{$A$} & \multicolumn{1}{c}{$B$}  \\ \cline{1-3}
\multicolumn{1}{r}{\rule{0pt}{4mm}$A$} & \multicolumn{1}{|c}{$60$, $50$} & \multicolumn{1}{c}{$50$, $100$} \\
\multicolumn{1}{r}{\rule{0pt}{4mm}$B$} & \multicolumn{1}{|c}{$50$, $60$} & \multicolumn{1}{c}{$60$, $50$} \\
\end{tabular}
\end{tabular}
\vspace*{-1mm}
\caption{Asymmetric Matching Pennies (AMP) games.}\label{AMPexp}
\end{center}
\vspace*{-7mm}
\end{table}

\begin{result}
In  the AMP games:
\begin{itemize}\addtolength{\itemsep}{-2mm}
\vspace*{-2mm}
\item[(i)] Subjects' choices differ across games.
\item[(ii)] Subjects' beliefs differ across games.
\item[(iii)] Subjects' beliefs are not correct in any of the games.
\item[(iv)] Subjects' choices are heterogeneous.
\item[(v)] Subjects' beliefs are heterogeneous.
\end{itemize}
\end{result}

\noindent Support for these findings can be found in Appendix D.  They are obviously at odds with Nash and logit-QRE, which both predict that choices and beliefs are identical and homogeneous across all five games and that beliefs match choices (in addition, the Nash-equilibrium prediction, $(p^*,q^*)=(\deel{1}{2},\deel{1}{6})$ is far from observed choice averages.). Heterogeneous Quantal Response Equilibrium (HQRE), a generalization of QRE proposed by \citet{rogers2009}, allows for heterogeneous choices by assuming that players' rationality parameters are draws from commonly-known distributions. Players' beliefs are thus assumed to be correct. The HQRE model is at odds with findings (i)-(iii) and (v). \citet{rogers2009} also propose a variant called Subjective Quantal Response Equilibrium (SQRE), that allows for heterogeneity in choices and beliefs. SQRE assumes that players' have subjective beliefs about the distributions that others' rationality parameters are drawn from. It is therefore not at odds with findings (iii)-(v), but since it is based on logit quantal responses, it predicts no change in choices and beliefs across games, which is refuted by findings (i) and (ii). Finally, the level-$k$ and Cognitive Hierarchy models, which are based on best responses, yield identical predictions for choices and beliefs across games, contradicting findings (i) and (ii).

\begin{figure}[t]
\begin{center}
\begin{tikzpicture}[scale=5]
\draw[line width=.5pt,gray] (0,0) -- (1,0) -- (1,1) -- (0,1) -- (0,0);
\node[scale=0.8pt] at (-0.03,-0.06) {$0$};
\node[scale=0.8pt] at (-0.03,0.5) {$\hf$};
\node[scale=0.8pt] at (-0.03,1) {$1$};
\node[scale=0.8pt] at (0.5,-0.06) {$\hf$};
\node[scale=0.8pt] at (1,-0.06) {$1$};
\node[scale=0.8pt] at (0.75,-0.06) {$p$};
\node[scale=0.8pt] at (-0.03,0.75) {$q$};
\filldraw[fill=mondrianYellow,draw=mondrianYellow,opacity=1] (0,1/6) -- (1/2,1/6) -- (1/2,1/2) -- (0,1/2) -- (0,1/6);
\draw[line width=1pt,black] (0.22,0.46) circle [x radius=0.0358,y radius=0.0432]; \node[black,scale=0.68pt] at (0.22,0.46) {1};
\draw[line width=1pt,black] (0.26,0.48) circle [x radius=0.0381,y radius=0.0433]; \node[black,scale=0.68pt] at (0.26,0.48) {2};
\draw[line width=1pt,black] (0.32,0.49) circle [x radius=0.0403,y radius=0.0433]; \node[black,scale=0.68pt] at (0.32,0.49) {3};
\draw[line width=1pt,black] (0.23,0.36) circle [x radius=0.0366,y radius=0.0416]; \node[black,scale=0.68pt] at (0.23,0.36) {4};
\draw[line width=1pt,black] (0.22,0.39) circle [x radius=0.0357,y radius=0.0423]; \node[black,scale=0.68pt] at (0.22,0.39) {5};

\node at (2,0.47) {
\begin{tikzpicture}[scale=5]
\draw[line width=.5pt,gray] (0,0) -- (1,0) -- (1,1) -- (0,1) -- (0,0);
\node[scale=0.8pt] at (-0.03,-0.06) {$0$};
\node[scale=0.8pt] at (-0.03,0.5) {$\hf$};
\node[scale=0.8pt] at (-0.03,1) {$1$};
\node[scale=0.8pt] at (0.5,-0.06) {$\hf$};
\node[scale=0.8pt] at (1,-0.06) {$1$};
\node[scale=0.8pt] at (0.75,-0.06) {$p$};
\node[scale=0.8pt] at (-0.03,0.75) {$q$};
\filldraw[fill=mondrianYellow,draw=mondrianYellow,opacity=1] (0,1/6) -- (1/2,1/6) -- (1/2,1) -- (0,1) -- (0,1/6);
\draw[line width=1pt,black] (0.3528,0.5123) circle [x radius=0.0241,y radius=0.0248]; \node[black,scale=0.68pt] at (0.3528,0.5123) {1};
\draw[line width=1pt,black] (0.4285,0.5598) circle [x radius=0.0251,y radius=0.0238]; \node[black,scale=0.68pt] at (0.4285,0.5598) {2};
\draw[line width=1pt,black] (0.4678,0.5810) circle [x radius=0.0261,y radius=0.0238]; \node[black,scale=0.68pt] at (0.4678,0.5810) {3};
\draw[line width=1pt,black] (0.4091,0.5263) circle [x radius=0.0245,y radius=0.0258]; \node[black,scale=0.68pt] at (0.4091,0.5263) {4};
\draw[line width=1pt,black] (0.3390,0.5032) circle [x radius=0.0241,y radius=0.0240]; \node[black,scale=0.68pt] at (0.3390,0.5032) {5};

\end{tikzpicture}};
\end{tikzpicture}
\vspace*{-2mm}
\caption{Average choices (left) and beliefs (right) in each of the five AMP games. The ellipses represent 95\% confidence regions for sample means. The yellow areas represent the $M$-choice (left) and $M$-belief (right) sets.}\label{expdataAMP}
\end{center}
\vspace*{-5mm}
\end{figure}

\begin{result}
The findings in Result 1 contradict the comparative statics predictions of Nash equilibrium, QRE, HQRE, SQRE, level-$k$, and Cognitive Hierarchy, but accord well with $M$-equilibrium predictions.
\end{result}
\begin{support}
Set-valued $M$ equilibrium easily accommodates the variations in choices and beliefs across games as well as the fact that beliefs differ from choics.  As can be seen from Figure~\ref{expdataAMP}, average choices and beliefs fall within the $M$-equilibrium sets.

$M$ equilibrium relies on the assumption of monotonicity. This posits that subjects will choose the alternative that is best, given their beliefs, more often. In Table \ref{best_responses} we report the fraction of best responses given stated beliefs for the five games. As can be seen, these range between .55 and .75, in accordance with monotonicity.
\end{support}

\begin{table}[t]
\begin{center}
\begin{tabular}{ccccccc}
\multicolumn{1}{r}{AMP game}& \multicolumn{1}{|c}{1} & \multicolumn{1}{c}{2} & \multicolumn{1}{c}{3}  & \multicolumn{1}{c}{4}  & \multicolumn{1}{c}{5} & \multicolumn{1}{c}{$average$}  \\ \cline{1-7}
\multicolumn{1}{r}{Row} & \multicolumn{1}{|c}{.61} & \multicolumn{1}{c}{.56}  & \multicolumn{1}{c}{.55} & \multicolumn{1}{c}{.60} & \multicolumn{1}{c}{.58} & \multicolumn{1}{c}{.58}\\
\multicolumn{1}{r}{Column} & \multicolumn{1}{|c}{.75} & \multicolumn{1}{c}{.73}  & \multicolumn{1}{c}{.66} & \multicolumn{1}{c}{.72} & \multicolumn{1}{c}{.75} & \multicolumn{1}{c}{.72}\\
\multicolumn{1}{r}{average} & \multicolumn{1}{|c}{.68} & \multicolumn{1}{c}{.65}  & \multicolumn{1}{c}{.61} & \multicolumn{1}{c}{.66} & \multicolumn{1}{c}{.67} & \multicolumn{1}{c}{.65}\\
\end{tabular}
\vspace*{0mm}
\caption{Fraction of best responses for each role in each of the five AMP games.}\label{best_responses}
\end{center}
\vspace*{-5mm}
\end{table}

\subsection{$M$-Equilibrium Multiplicity and Mis-Coordination}

Next, consider the two symmetric $3\times3$ games in the first row of Table~\ref{3by3exp}. As in the AMP experiments, the two games share a common structure. In particular, one game can be obtained from the other by adding a constant to each column (row) of Row's (Column's) payoffs. As a result, the two games have the same best-response structure and the same unique pure-strategy Nash equilibrium $(Y,Y)$. Note that the games are dominance solvable. One might therefore expect Nash equilibrium to be a good predictor of behavior in these games, and that there would be little difference in how individuals play the two games.

Subjects' behavior does not support these predictions. The top-left (top-right) panel of Figure~\ref{expdataDS} shows average choices (beliefs) in the two experiments. As in the $2\times2$ AMP experiments, observed choices and beliefs are far from Nash and far from each other. These results may seem puzzling when viewed from the perspective of Nash equilibrium, but $M$ equilibrium allows us to glean some intuition for these findings.

\begin{table}[t]
\vspace*{10mm}
\begin{center}
\begin{tabular}{c|ccc}
DS1 &  $R$ & $B$ & $Y$ \\ \cline{1-4}
\multicolumn{1}{r}{\rule{0pt}{5mm}$R$} & \multicolumn{1}{|c}{$80$, $80$} & \multicolumn{1}{c}{$30$, $160$} & \multicolumn{1}{c}{$20$, $10$} \\
\multicolumn{1}{r}{$B$} & \multicolumn{1}{|c}{$160$, $30$} & \multicolumn{1}{c}{$30$, $30$} & \multicolumn{1}{c}{$10$, $40$} \\
\multicolumn{1}{r}{$Y$} & \multicolumn{1}{|c}{$10$, $20$} & \multicolumn{1}{c}{$40$, $10$} & \multicolumn{1}{c}{$30$, $30$} \\
\end{tabular}
\hspace*{10mm}
\begin{tabular}{c|ccc}
DS2 & $R$ & $B$ & $Y$ \\ \cline{1-4}
\multicolumn{1}{r}{\rule{0pt}{5mm}$R$} & \multicolumn{1}{|c}{$75$, $75$} & \multicolumn{1}{c}{$5$, $155$} & \multicolumn{1}{c}{$190$, $5$} \\
\multicolumn{1}{r}{$B$} & \multicolumn{1}{|c}{$155$, $5$} & \multicolumn{1}{c}{$5$, $5$} & \multicolumn{1}{c}{$180$, $15$} \\
\multicolumn{1}{r}{$Y$} & \multicolumn{1}{|c}{$5$, $190$} & \multicolumn{1}{c}{$15$, $180$} & \multicolumn{1}{c}{$200$, $200$} \\
\end{tabular}
\vspace*{5mm}\\
\begin{tabular}{c|ccc}
 NL & $R$ & $B$ & $Y$ \\ \cline{1-4}
\multicolumn{1}{r}{\rule{0pt}{5mm}$R$} & \multicolumn{1}{|c}{$70$, $70$} & \multicolumn{1}{c}{$60$, $500$} & \multicolumn{1}{c}{$10$, $50$} \\
\multicolumn{1}{r}{$B$} & \multicolumn{1}{|c}{$500$, $60$} & \multicolumn{1}{c}{$40$, $40$} & \multicolumn{1}{c}{$0$, $61$} \\
\multicolumn{1}{r}{$Y$} & \multicolumn{1}{|c}{$50$, $10$} & \multicolumn{1}{c}{$61$, $0$} & \multicolumn{1}{c}{$30$, $30$} \\
\end{tabular}
\hspace*{10mm}
\begin{tabular}{c|ccc}
KM & $R$ & $B$ & $Y$ \\ \cline{1-4}
\multicolumn{1}{r}{\rule{0pt}{5mm}$R$} & \multicolumn{1}{|c}{$120$, $120$} & \multicolumn{1}{c}{$90$, $60$} & \multicolumn{1}{c}{$60$, $120$} \\
\multicolumn{1}{r}{$B$} & \multicolumn{1}{|c}{$60$, $90$} & \multicolumn{1}{c}{$90$, $90$} & \multicolumn{1}{c}{$60$, $90$} \\
\multicolumn{1}{r}{$Y$} & \multicolumn{1}{|c}{$120$, $60$} & \multicolumn{1}{c}{$90$, $60$} & \multicolumn{1}{c}{$30$, $30$} \\
\end{tabular}
\vspace*{-1mm}
\caption{Symmetric $3\times 3$ games.}\label{3by3exp}
\end{center}
\vspace*{-7mm}
\end{table}

In contrast to the unique Nash equilibrium prediction, these games have multiple $M$ equilibria. In particular, there are four $M$-choice sets and supporting $M$-belief sets, see the left and and right panels of Figure~\ref{expdataDS} respectively. The Nash equilibrium, $(Y,Y)$, is part of the yellow set, which is supported by a small $M$-belief set. The blue $M$-choice set is furthest away from Nash but is supported by the largest $M$-belief set. The blue $M$ equilibrium captures most of the observed behavior in DS1. However, in DS2, subjects do not appear to play any of the four $M$-equilibria, at least based on aggregate behavior.

Interestingly, multiplicity of $M$ equilibria introduces issues of strategic coordination, akin to those in standard game theory when games have multiple Nash equilibria. It is not enough for subjects to figure out what strategies constitute $M$-equilibria, they need to play the same one.  Whether a specific $M$ equilibrium is more salient and therefore facilitates strategic coordination depends on features of the game.

\begin{figure}[p]
\begin{center}

\begin{tikzpicture}[scale=5]
\draw[line width=.5pt,gray] (0,0) -- (1,0) -- (1/2,1/2*3^.5) -- (0,0);
\node[scale=0.8pt] at (-0.03,-0.03) {$R$};
\node[scale=0.8pt] at (1.03,-0.03) {$B$};
\node[scale=0.8pt] at (1/2,1/2*3^.5+0.03) {$Y$};

\filldraw[fill=mondrianBlue,draw=mondrianBlue,opacity=1] (1/2,1/6*3^.5) -- (1/2,0) -- (7/8,0) -- (13/16,1/16*3^.5) -- (1/2,1/6*3^.5);
\filldraw[fill=mondrianCyan,draw=mondrianCyan,opacity=1] (13/16,1/16*3^.5) -- (9/10,1/30*3^.5) -- (15/22,5/22*3^.5) -- (21/32,7/31*3^.5) -- (13/16,1/16*3^.5);
\filldraw[fill=mondrianGreen,draw=mondrianGreen,opacity=1] (15/22,5/22*3^.5) -- (12/17,4/17*3^.5) -- (1/2,2/5*3^.5) -- (1/2,7/18*3^.5) -- (15/22,5/22*3^.5);
\filldraw[fill=mondrianYellow,draw=mondrianYellow,opacity=1] (1/2,2/5*3^.5) -- (4/9,4/9*3^.5) -- (1/2,1/2*3^.5) -- (1/2,2/5*3^.5);

\filldraw [fill=red,draw=red,opacity=1] (.6260, .1317) circle (0.03);
\filldraw [fill=green,draw=green,opacity=1] (.5667, .4150) circle (0.03);

\node at (2,1/4*3^.5) {
\begin{tikzpicture}[scale=5]
\draw[line width=.5pt,gray] (0,0) -- (1,0) -- (1/2,1/2*3^.5) -- (0,0);
\node[scale=0.8pt] at (-0.03,-0.03) {$R$};
\node[scale=0.8pt] at (1.03,-0.03) {$B$};
\node[scale=0.8pt] at (1/2,1/2*3^.5+0.03) {$Y$};

\filldraw[fill=mondrianBlue,draw=mondrianBlue,opacity=1] (0,0) -- (7/8,0) -- (7/16,7/16*3^.5) -- (0,0);
\filldraw[fill=mondrianCyan,draw=mondrianCyan,opacity=1] (7/8,0) --(15/16,0) -- (15/34,15/34*3^.5) -- (7/16,7/16*3^.5) -- (7/8,0);
\filldraw[fill=mondrianGreen,draw=mondrianGreen,opacity=1] (15/16,0) -- (1,0) -- (4/9,4/9*3^.5) -- (15/34,15/34*3^.5) -- (15/16,0);
\filldraw[fill=mondrianYellow,draw=mondrianYellow,opacity=1] (1,0) -- (1/2,1/2*3^.5) -- (4/9,4/9*3^.5) -- (1,0);

\filldraw [fill=red,draw=red,opacity=1] (.5324, .1641) circle (.03);
\filldraw [fill=green,draw=green,opacity=1] (.5129, .3717) circle (.03);

\end{tikzpicture}};

\node at (1/2,-.6) {
\begin{tikzpicture}[scale=5]
\draw[line width=.5pt,gray] (0,0) -- (1,0) -- (1/2,1/2*3^.5) -- (0,0);
\node[scale=0.8pt] at (-0.03,-0.03) {$R$};
\node[scale=0.8pt] at (1.03,-0.03) {$B$};
\node[scale=0.8pt] at (1/2,1/2*3^.5+0.03) {$Y$};

\filldraw[fill=mondrianBlue,draw=mondrianBlue,opacity=1] (1/2,1/6*3^.5) -- (1/2,0) -- (7/8,0) -- (13/16,1/16*3^.5) -- (1/2,1/6*3^.5);
\filldraw[fill=mondrianCyan,draw=mondrianCyan,opacity=1] (13/16,1/16*3^.5) -- (9/10,1/30*3^.5) -- (15/22,5/22*3^.5) -- (21/32,7/31*3^.5) -- (13/16,1/16*3^.5);
\filldraw[fill=mondrianGreen,draw=mondrianGreen,opacity=1] (15/22,5/22*3^.5) -- (12/17,4/17*3^.5) -- (1/2,2/5*3^.5) -- (1/2,7/18*3^.5) -- (15/22,5/22*3^.5);
\filldraw[fill=mondrianYellow,draw=mondrianYellow,opacity=1] (1/2,2/5*3^.5) -- (4/9,4/9*3^.5) -- (1/2,1/2*3^.5) -- (1/2,2/5*3^.5);

\filldraw [fill=cyan,draw=cyan,opacity=1] (.6250, .1580) circle (0.31^.5*.075);
\filldraw [fill=green,draw=green,opacity=1] (.7059, .2547) circle (0.07^.5*.075);
\filldraw [fill=red,draw=red,opacity=1] (.7254, .1098) circle (0.15^.5*.075);
\filldraw [fill=gray,draw=gray,opacity=1] (.6860, .0604) circle (0.09^.5*.075);
\filldraw [fill=magenta,draw=magenta,opacity=1] (.6316, .5470) circle (0.04^.5*.075);
\filldraw [fill=orange,draw=orange,opacity=1] (.5914, .0931) circle (0.19^.5*.075);
\filldraw [fill=pink,draw=pink,opacity=1] (.5000, .0241) circle (0.15^.5*.075);
\end{tikzpicture}};

\node at (2,-.6) {
\begin{tikzpicture}[scale=5]
\draw[line width=.5pt,gray] (0,0) -- (1,0) -- (1/2,1/2*3^.5) -- (0,0);
\node[scale=0.8pt] at (-0.03,-0.03) {$R$};
\node[scale=0.8pt] at (1.03,-0.03) {$B$};
\node[scale=0.8pt] at (1/2,1/2*3^.5+0.03) {$Y$};

\filldraw[fill=mondrianBlue,draw=mondrianBlue,opacity=1] (0,0) -- (7/8,0) -- (7/16,7/16*3^.5) -- (0,0);
\filldraw[fill=mondrianCyan,draw=mondrianCyan,opacity=1] (7/8,0) --(15/16,0) -- (15/34,15/34*3^.5) -- (7/16,7/16*3^.5) -- (7/8,0);
\filldraw[fill=mondrianGreen,draw=mondrianGreen,opacity=1] (15/16,0) -- (1,0) -- (4/9,4/9*3^.5) -- (15/34,15/34*3^.5) -- (15/16,0);
\filldraw[fill=mondrianYellow,draw=mondrianYellow,opacity=1] (1,0) -- (1/2,1/2*3^.5) -- (4/9,4/9*3^.5) -- (1,0);

\filldraw [fill=cyan,draw=cyan,opacity=1] (.5875, .1485) circle (0.31^.5*.075);
\filldraw [fill=green,draw=green,opacity=1] (.7338, .2741) circle (0.07^.5*.075);
\filldraw [fill=red,draw=red,opacity=1] (.7740, .0453) circle (0.15^.5*.075);
\filldraw [fill=gray,draw=gray,opacity=1] (.1931, .0900) circle (0.09^.5*.075);
\filldraw [fill=magenta,draw=magenta,opacity=1] (.5268, .5853) circle (0.04^.5*.075);
\filldraw [fill=orange,draw=orange,opacity=1] (.4449, .2669) circle (0.19^.5*.075);
\filldraw [fill=pink,draw=pink,opacity=1] (.4031, .0618) circle (0.15^.5*.075);

\end{tikzpicture}};

\node at (1/2,-1.7) {
\begin{tikzpicture}[scale=5]
\draw[line width=.5pt,gray] (0,0) -- (1,0) -- (1/2,1/2*3^.5) -- (0,0);
\node[scale=0.8pt] at (-0.03,-0.03) {$R$};
\node[scale=0.8pt] at (1.03,-0.03) {$B$};
\node[scale=0.8pt] at (1/2,1/2*3^.5+0.03) {$Y$};

\filldraw[fill=mondrianBlue,draw=mondrianBlue,opacity=1] (1/2,1/6*3^.5) -- (1/2,0) -- (7/8,0) -- (13/16,1/16*3^.5) -- (1/2,1/6*3^.5);
\filldraw[fill=mondrianCyan,draw=mondrianCyan,opacity=1] (13/16,1/16*3^.5) -- (9/10,1/30*3^.5) -- (15/22,5/22*3^.5) -- (21/32,7/31*3^.5) -- (13/16,1/16*3^.5);
\filldraw[fill=mondrianGreen,draw=mondrianGreen,opacity=1] (15/22,5/22*3^.5) -- (12/17,4/17*3^.5) -- (1/2,2/5*3^.5) -- (1/2,7/18*3^.5) -- (15/22,5/22*3^.5);
\filldraw[fill=mondrianYellow,draw=mondrianYellow,opacity=1] (1/2,2/5*3^.5) -- (4/9,4/9*3^.5) -- (1/2,1/2*3^.5) -- (1/2,2/5*3^.5);

\filldraw [fill=orange,draw=orange,opacity=1] (.4886, .8463) circle (0.09^.5*.075);
\filldraw [fill=red,draw=red,opacity=1] (.6067, .2194) circle (0.16^.5*.075);
\filldraw [fill=cyan,draw=cyan,opacity=1] (.5992, .3636) circle (0.27^.5*.075);
\filldraw [fill=green,draw=green,opacity=1] (.6636, .2992) circle (0.11^.5*.075);
\filldraw [fill=pink,draw=pink,opacity=1] (.5071, .6804) circle (0.15^.5*.075);
\filldraw [fill=gray,draw=gray,opacity=1] (.4896, .1624) circle (0.10^.5*.075);
\filldraw [fill=magenta,draw=magenta,opacity=1] (.5439, .4558) circle (0.12^.5*.075);

\end{tikzpicture}};

\node at (2,-1.7) {
\begin{tikzpicture}[scale=5]
\draw[line width=.5pt,gray] (0,0) -- (1,0) -- (1/2,1/2*3^.5) -- (0,0);
\node[scale=0.8pt] at (-0.03,-0.03) {$R$};
\node[scale=0.8pt] at (1.03,-0.03) {$B$};
\node[scale=0.8pt] at (1/2,1/2*3^.5+0.03) {$Y$};

\filldraw[fill=mondrianBlue,draw=mondrianBlue,opacity=1] (0,0) -- (7/8,0) -- (7/16,7/16*3^.5) -- (0,0);
\filldraw[fill=mondrianCyan,draw=mondrianCyan,opacity=1] (7/8,0) --(15/16,0) -- (15/34,15/34*3^.5) -- (7/16,7/16*3^.5) -- (7/8,0);
\filldraw[fill=mondrianGreen,draw=mondrianGreen,opacity=1] (15/16,0) -- (1,0) -- (4/9,4/9*3^.5) -- (15/34,15/34*3^.5) -- (15/16,0);
\filldraw[fill=mondrianYellow,draw=mondrianYellow,opacity=1] (1,0) -- (1/2,1/2*3^.5) -- (4/9,4/9*3^.5) -- (1,0);

\filldraw [fill=orange,draw=orange,opacity=1] (.5011, .7991) circle (0.09^.5*.075);
\filldraw [fill=red,draw=red,opacity=1] (.5596, .1303) circle (0.16^.5*.075);
\filldraw [fill=cyan,draw=cyan,opacity=1] (.5048, .3446) circle (0.27^.5*.075);
\filldraw [fill=green,draw=green,opacity=1] (.7200, .2198) circle (0.11^.5*.075);
\filldraw [fill=pink,draw=pink,opacity=1] (.5745, .5455) circle (0.15^.5*.075);
\filldraw [fill=gray,draw=gray,opacity=1] (.3009, .1802) circle (0.10^.5*.075);
\filldraw [fill=magenta,draw=magenta,opacity=1] (.3825, .5163) circle (0.12^.5*.075);

\end{tikzpicture}};

\end{tikzpicture}
\caption{The four $M$ equilibria for the DS1 and DS2 games. In the top panels, the colored circles indicate the average choices (left) and beliefs (right) for DS1 (red) and DS2 (green). The middle (DS1) and bottom (DS2) panels show a $k$-means analysis for the DS games. The $k$-means algorithm was performed on the elicited belief data. Each colored circle corresponds to the average choices (left) and elicited beliefs (right) within each cluster. The size of the circles is proportional to the number of observations belonging to the particular cluster.}\label{expdataDS}
\end{center}
\vspace*{-25mm}
\end{figure}

For the two games in this experiment we conjecture that even though they both share the same best response structure, the differences in the payoff matrices can affect the salience of specific strategies and the associated $M$ equilibria. For instance, in DS1 the maximum payoff from $Y$ is 40, compared to 80 and 160 from $R$ and $B$ respectively. This makes $Y$ and any $M$ equilibrium involving a high choice probability for $Y$ unattractive, leaving the blue $M$-equilibrium as the most salient. In DS2, all actions have minimum and maximum payoffs of similar magnitudes. Thus, no $M$ equilibrium is particularly salient. We therefore expect a higher degree of strategic mis-coordination in DS2.

To examine whether strategic (mis-)coordination is driving observed behavior, we provide a more detailed analysis of the beliefs and their nexus to choices. We separate elicited beliefs in each of the two games in different clusters using the $k$-means clustering algorithm \citep{macqueen1967}, which organizes observations from a multi-dimensional space into $k$ separate clusters. It takes $k$ points as the provisional centers of the respective clusters and assigns each observation to the nearest one. Next, it computes the mean of the observations in each cluster, which becomes the new center. The process is repeated until convergence.\footnote{See Appendix D for a more detailed description of the $k$-means algorithm and how it is implemented here.}  We then compute the average of the choices corresponding to the beliefs in each cluster.

The middle panels of Figure~\ref{expdataDS} show the results for DS1. In this game, the elicited beliefs are mainly concentrated in the lower side of the blue belief set (right-middle panel) and the corresponding choices are also in the blue choice set (left-middle panel), indicating that subjects in this game are mostly playing the blue $M$ equilibrium. Interestingly, for the two belief clusters outside of the blue set, the corresponding average choices are very close to the set of the same color, indicating that some subjects played a different $M$ equilibrium.

The bottom panels of Figure~\ref{expdataDS} correspond to DS2. Elicited beliefs here are more spread out, with a substantial fraction outside of the blue set. Nevertheless, except for two clusters (depicted by the cyan and magenta colored disk), all other choice clusters lie within, or are very close to, the set with the same color of the corresponding beliefs. In contrast to DS1, the blue set contains the average choices of only two of these clusters, while one of the clusters is essentially playing the $M$ equilibrium that includes Nash. In view of these results, our conjecture that observed heterogeneity in DS2 is driven by strategic mis-coordination due to $M$-equilibrium multiplicity is plausible (see also the discussion in Section 4.4).

\subsection{No Logit}
\label{subsec:noLogit}

\begin{figure}[t]
\begin{center}

\begin{tikzpicture}[scale=5]
\draw[line width=.5pt,gray] (0,0) -- (1,0) -- (1/2,1/2*3^.5) -- (0,0);
\node[scale=0.8pt] at (-0.03,-0.03) {$R$};
\node[scale=0.8pt] at (1.03,-0.03) {$B$};
\node[scale=0.8pt] at (1/2,1/2*3^.5+0.03) {$Y$};

\filldraw[fill=mondrianBlue,draw=mondrianBlue,opacity=1] (1/2,19/118*3^.5) -- (1/2,0) -- (20/21,0) -- (1/2,19/118*3^.5);
\filldraw[fill=mondrianCyan,draw=mondrianCyan,opacity=1] (1/2,1/6*3^.5) -- (41/44,1/44*3^.5) -- (225/317,75/317*3^.5) -- (1/2,1/6*3^.5);
\filldraw[fill=mondrianGreen,draw=mondrianGreen,opacity=1] (225/317,75/317*3^.5) -- (129/178,43/178*3^.5) -- (1/2,41/86*3^.5) -- (1/2,143/326*3^.5) -- (225/317,75/317*3^.5);
\filldraw[fill=mondrianYellow,draw=mondrianYellow,opacity=1] (1/2,41/86*3^.5) -- (1/2,1/2*3^.5) -- (43/88,43/88*3^.5) -- (1/2,41/86*3^.5);

\filldraw [fill=green,draw=green,opacity=1] (.5100, .2771) circle (0.21^.5*.075);
\filldraw [fill=pink,draw=pink,opacity=1] (.6485, .1286) circle (0.21^.5*.075);
\filldraw [fill=purple,draw=purple,opacity=1] (.6111, .0) circle (0.04^.5*.075);
\filldraw [fill=brown,draw=brown,opacity=1] (.6842, .3191) circle (0.08^.5*.075);
\filldraw [fill=gray,draw=gray,opacity=1] (.5949, .0658) circle (0.16^.5*.075);
\filldraw [fill=cyan,draw=cyan,opacity=1] (.7736, .1634) circle (0.11^.5*.075);
\filldraw [fill=red,draw=red,opacity=1] (.7857, .2062) circle (0.04^.5*.075);
\filldraw [fill=orange,draw=orange,opacity=1] (.6286, .4206) circle (0.15^.5*.075);

\node at (2,1/4*3^.5) {
\begin{tikzpicture}[scale=5]
\draw[line width=.5pt,gray] (0,0) -- (1,0) -- (1/2,1/2*3^.5) -- (0,0);
\node[scale=0.8pt] at (-0.03,-0.03) {$R$};
\node[scale=0.8pt] at (1.03,-0.03) {$B$};
\node[scale=0.8pt] at (1/2,1/2*3^.5+0.03) {$Y$};

\filldraw[fill=mondrianBlue,draw=mondrianBlue,opacity=1] (0,0) -- (20/21,0) -- (1/4,1/4*3^.5) -- (0,0);
\filldraw[fill=mondrianCyan,draw=mondrianCyan,opacity=1] (20/21,0) -- (1/4,1/4*3^.5) -- (15/32,15/32*3^.5) -- (20/21,0);
\filldraw[fill=mondrianGreen,draw=mondrianGreen,opacity=1] (150/157,0) -- (43/45,0) -- (15/32,15/32*3^.5) -- (43/88,43/88*3^.5) -- (150/157,0);
\filldraw[fill=mondrianYellow,draw=mondrianYellow,opacity=1] (43/45,0) -- (1,0) -- (1/2,1/2*3^.5) -- (43/88,43/88*3^.5) -- (43/45,0);

\filldraw [fill=green,draw=green,opacity=1] (.5044, .2430) circle (0.21^.5*.075);
\filldraw [fill=pink,draw=pink,opacity=1] (.6201, .0859) circle (0.21^.5*.075);
\filldraw [fill=purple,draw=purple,opacity=1] (.1819, .0832) circle (0.04^.5*.075);
\filldraw [fill=brown,draw=brown,opacity=1] (.5125, .4834) circle (0.08^.5*.075);
\filldraw [fill=gray,draw=gray,opacity=1] (.4459, .0570) circle (0.16^.5*.075);
\filldraw [fill=cyan,draw=cyan,opacity=1] (.3435, .3181) circle (0.11^.5*.075);
\filldraw [fill=red,draw=red,opacity=1] (.8650, .0573) circle (0.04^.5*.075);
\filldraw [fill=orange,draw=orange,opacity=1] (.6737, .2751) circle (0.15^.5*.075);

\end{tikzpicture}};

\end{tikzpicture}
\vspace*{-3mm}
\caption{The four $M$ equilibrium sets for the ``no logit'' game. The graphs show choice (left) and belief (right) data organized in clusters obtained using the $k$-means algorithm on the elicited beliefs data. Each colored circle corresponds to the average choices (left) and elicited beliefs (right) within each cluster. The size of the circles is proportional to the number of observations belonging to the particular cluster. }\label{Nologitexp}
\end{center}
\vspace*{-9mm}
\end{figure}

For the symmetric $3\times3$ game in the lower-left panel of Table~\ref{3by3exp}, $(Y,Y)$ is the unique Nash equilibrium and there are four $M$ equilibria, see Figure~\ref{Nologitexp}.  The two $M$ equilibria with the largest belief sets also have the largest choice sets and are therefore expected to be empirically most relevant. An interesting feature in this game is that the profiles in the blue equilibrium cannot be part of any logit-QRE,\footnote{Since the Nash equilibrium is unique, there is a single path of logit equilibria that starts at the barycenter and ends at $Y$, see \cite{McKelveyPalfrey1995}. If this path had non-empty intersection with the blue $M$ set, it would need to cross the diagonal $\sigma_R=\sigma_Y$. But on this diagonal $\pi_Y>\pi_R$.}
which is why the game is labeled ``NL'' for ``no logit.''

Nevertheless, as can be seen from the experimental results depicted in Figure~\ref{Nologitexp}, both the blue and the cyan $M$ equilibria are empirically relevant. A $k$-means clustering analysis again shows strong evidence of strategic mis-coordination and reveals that the majority of choices (62\%) are in the ``no logit'' blue region.  As pointed out above, the issue of equilibrium-multiplicity only arises in the context of $M$ equilibrium -- the ``no logit'' game has a unique Nash equilibrium and a unique logit-QRE. The strategic coordination problems uncovered by the choice and belief data cannot be explained by existing behavioral game-theory models.

\subsection{Stability}

\begin{figure}[t]
\begin{center}

\begin{tikzpicture}[scale=5]
\draw[line width=.5pt,gray] (0,0) -- (1,0) -- (1/2,1/2*3^.5) -- (0,0);
\node[scale=0.8pt] at (-0.03,-0.03) {$R$};
\node[scale=0.8pt] at (1.03,-0.03) {$B$};
\node[scale=0.8pt] at (1/2,1/2*3^.5+0.03) {$Y$};

\filldraw[fill=mondrianOrange,draw=mondrianOrange,opacity=1] (0,0) -- (1/4,1/4*3^.5) -- (1/2,1/6*3^.5) -- (0,0);

\filldraw [fill=cyan,draw=cyan,opacity=1] (.0455, .0787) circle (0.23^.5*.075);
\filldraw [fill=pink,draw=pink,opacity=1] (.2553, .2211) circle (0.20^.5*.075);
\filldraw [fill=brown,draw=brown,opacity=1] (.1339, .2320) circle (0.23^.5*.075);
\filldraw [fill=gray,draw=gray,opacity=1] (.4000, .1732) circle (0.04^.5*.075);
\filldraw [fill=blue,draw=blue,opacity=1] (.3400, .3811) circle (0.10^.5*.075);
\filldraw [fill=green,draw=green,opacity=1] (.1170, .0921) circle (0.20^.5*.075);

\node at (2,1/4*3^.5) {
\begin{tikzpicture}[scale=5]
\draw[line width=.5pt,gray] (0,0) -- (1,0) -- (1/2,1/2*3^.5) -- (0,0);
\node[scale=0.8pt] at (-0.03,-0.03) {$R$};
\node[scale=0.8pt] at (1.03,-0.03) {$B$};
\node[scale=0.8pt] at (1/2,1/2*3^.5+0.03) {$Y$};

\filldraw[fill=mondrianOrange,draw=mondrianOrange,opacity=1] (0,0) -- (1,0) -- (1/3,1/3*3^.5) -- (0,0);

\filldraw [fill=cyan,draw=cyan,opacity=1] (.2807, .3196) circle (0.23^.5*.075);
\filldraw [fill=pink,draw=pink,opacity=1] (.4413, .3357) circle (0.20^.5*.075);
\filldraw [fill=brown,draw=brown,opacity=1] (.2210, .0985) circle (0.23^.5*.075);
\filldraw [fill=gray,draw=gray,opacity=1] (.6480, .5075) circle (0.04^.5*.075);
\filldraw [fill=blue,draw=blue,opacity=1] (.6550, .1749) circle (0.10^.5*.075);
\filldraw [fill=green,draw=green,opacity=1] (.4217, .1441) circle (0.20^.5*.075);

\end{tikzpicture}};

\end{tikzpicture}
\vspace*{-3mm}
\caption{The colored sets correspond to the unique symmetric $M$ equilibrium sets for the ``Kohlberg-Mertens game.'' The colored circles indicate the average choices (left) and beliefs (right) for each of six clusters obtained through the $k$-means algorithm. Each circle's size is proportional to the number of observations in the respective cluster.}\label{KMkmeans}
\end{center}
\vspace*{-9mm}
\end{figure}

A final example is a game introduced by \citet{KohlbergMertens1986}, which has six pure strategy Nash equilibria that can be arranged in a circle (with mixtures of adjacent pure equilibria corresponding to mixed equilibria), see the bottom-right game in Table~\ref{3by3exp}. The \q{KM} game was identified by \citet{mclennan2016} as one where ``observed behavior will not be characterized by repetition of any one of the equilibria.'' This prediction follows from the application of the \emph{index +1 principle}, introduced in that paper.\footnote{As \citet{mclennan2016} points out, a more precise (albeit less catchy) version of the ``index +1 principle'' is the ``Euler characteristic equals index principle.'' \citet{demichelis2003} show that the latter condition is necessary for any ``natural'' dynamic process of adjustment to converge. For the KM game, the Euler characteristic of the set of Nash equilibria is zero, while its index is $+1$.} In stark contrast to the predicted instability of Nash equilibria there is a unique symmetric and robust $M$ equilibrium.

Figure~\ref{KMkmeans} shows results for the KM game. The data from the experiment is broken down in clusters, using the same procedure as before. Far from the instability and unpredictability of behavior one might expect based on the classical game-theoretic reasoning, we find average choices to be concentrated near the $R$ vertex of the simplex. Interestingly, beliefs are far more spread out, but with $Y$ being the least expected choice by the opponent.
The orange areas in Figure~\ref{KMkmeans} show that the choice and belief sets for the unique $M$-equilibrium capture virtually all observed choices and beliefs. This example illustrates the value of $M$ equilibrium as an empirically relevant theory. In this case, as a theory that provides sharper and more accurate predictions than classical game-theory-solution concepts.

\subsection{An Empirical Evaluation}
\label{subsec:empeval}

The $3\times 3$ games in Table \ref{3by3exp} are of similar complexity and were played using the same experimental protocol by students from the same subject pool. As such they are well suited to compare $\mu$ equilibrium and logit-QRE in terms of fit and out-of-sample predictive power.

One difference between the models is in terms of beliefs: a $\mu$-equilibrium choice is supported by a set of beliefs while the belief has to match the choice under logit-QRE. The main focus of this section, however, is on observed choices. One reason is that many prior studies have applied behavioral-game-theory models to explain choices in experiments without belief elicitation. As such it is useful to evaluate $\mu$ equilibrium based solely on its choice predictions. We demonstrate below that $\mu$ equilibrium offers an interesting alternative to existing behavioral-game-theory models as it naturally predicts heterogeneous choices in games with multiple $M$ equilibria.

Another difference is that $\mu$ equilibria can be computed explicitly. For $\mu(\varepsilon)=(1,\varepsilon,\varepsilon^2)/(1+\varepsilon+\varepsilon^2)$ and $\varepsilon\in[0,1]$, let $E_{\varepsilon}(G)$ denote the set of $\mu$-equilibrium choices of $G$, which are examples of $\varepsilon$-proper equilibria (see Section \ref{sec:Meta}). We choose this parametrization for $\mu$ not because we think it is the best fitting model (our discussion of profect Nash equilibria in Section 2.4 suggests it is not) but because it allows us to measure the sizes of deviations using a well-known model from the refinement literature. Like the profect Nash equilibria of Section \ref{sec:Nash}, $E_{\varepsilon}(G)$ can be found geometrically, i.e.
\begin{displaymath}
  E_{\varepsilon}(G)\,=\,\sqcup_{r\,\in\,\mathcal{P}}\,\overline{M}^c_r\cap P_r(\varepsilon)
\end{displaymath}
where $P_r(\varepsilon)$ is the face of the permutahedron $P(\mu(\varepsilon))$ labeled by $r\in\mathcal{P}$. The top panels of Figure \ref{epsEqgraphs} show three permutahedra in the interior of the simplex corresponding to $\varepsilon=\deel{2}{3}$ (small), $\varepsilon=\deel{2}{5}$ (medium), and $\varepsilon=\deel{1}{5}$ (large). The $\mu$ equilibrium profiles are those for which the label of the $M$-choice set matches that of the permutahedron (see the top panels of Figure \ref{rank3}). The colored curves in the top panels of Figure \ref{epsEqgraphs} show the $\mu$-equilibrium choice predictions for the various games.\footnote{For instance, the lower blue curve in the top-middle panel of Figure \ref{epsEqgraphs} corresponds to the ranking $B\succ R\succ Y$ while the upper blue curve corresponds to $B\succ R\sim Y$. See Appendix F for an explicit description of the $\mu$ equilibria for the games in Table \ref{3by3exp}.} The bottom panels show logit-QRE choice predictions for $\lambda\in[0,\infty)$.

At first blush, the two models yield similar predictions. For the KM game, for instance, both models produce a curve that is confined to the orange $M$-choice set and that connects the simplex' centroid (random behavior) with the $R$ vertex (rational behavior), see the right panels. Likewise, the curves for the DS games in the left panels look similar, traveling through several $M$-choice sets and connecting the centroid to the unique Nash equilibrium.

\begin{figure}[t]
\begin{center}

\begin{tikzpicture}[scale=4.3]
\draw[line width=.5pt,gray] (0,0) -- (1,0) -- (1/2,1/2*3^.5) -- (0,0);
\node[scale=0.8pt] at (-0.03,-0.03) {$R$};
\node[scale=0.8pt] at (1.03,-0.03) {$B$};
\node[scale=0.8pt] at (1/2,1/2*3^.5+0.03) {$Y$};

\filldraw[fill=mondrianGrey,draw=mondrianGrey,opacity=1] (1/2,1/6*3^.5) -- (1/2,0) -- (7/8,0) -- (13/16,1/16*3^.5) -- (1/2,1/6*3^.5);
\filldraw[fill=mondrianGrey,draw=mondrianGrey,opacity=1] (13/16,1/16*3^.5) -- (9/10,1/30*3^.5) -- (15/22,5/22*3^.5) -- (21/32,7/31*3^.5) -- (13/16,1/16*3^.5);
\filldraw[fill=mondrianGrey,draw=mondrianGrey,opacity=1] (15/22,5/22*3^.5) -- (12/17,4/17*3^.5) -- (1/2,2/5*3^.5) -- (1/2,7/18*3^.5) -- (15/22,5/22*3^.5);
\filldraw[fill=mondrianGrey,draw=mondrianGrey,opacity=1] (1/2,2/5*3^.5) -- (4/9,4/9*3^.5) -- (1/2,1/2*3^.5) -- (1/2,2/5*3^.5);

\draw[line width=.5pt,dashed,black] (0.421053, 0.182321) -- (0.368421, 0.273482) -- (0.447368,0.410223) -- (0.552632, 0.410223) -- (0.631579, 0.273482) -- (0.578947,0.182321) -- (0.421053, 0.182321);
\draw[line width=.5pt,dashed,black] (0.307692,0.0888231) -- (0.230769, 0.222058) -- (0.423077,0.555144) -- (0.576923, 0.555144) -- (0.769231, 0.222058) -- (0.692308, 0.0888231)-- (0.307692, 0.0888231);
\draw[line width=.5pt,dashed,black] (0.177419, 0.0279363) -- (0.112903, 0.139682) -- (0.435484,0.698408) -- (0.564516, 0.698408) -- (0.887097, 0.139682) -- (0.822581,0.0279363) -- (0.177419, 0.0279363);

\begin{axis}[scale=0.175, axis line style={draw=none}, tick style={draw=none}, ticks=none, xmin=0, xmax=1.2, ymin=0, ymax=1]
\addplot[thin, mondrianBlue, smooth] plot coordinates
            {
                (0.865881,0.0157939)
                (0.825739,0.0269265)
                (0.788626,0.0400845)
                (0.754473,0.0548088)
                (0.72316,0.0707014)
                (0.694541,0.0874235)
                (0.668448,0.104691)
                (0.644706,0.122268)
                (0.623136,0.139964)
                (0.603563,0.157623)
                (0.585819,0.175123)
                (0.569743,0.192367)
                (0.555185,0.209283)
                (0.542007,0.225816)
                (0.530081,0.241926)
                (0.519289,0.257585)
                (0.509523,0.272776)
                (0.500686,0.287489)
            };
\end{axis}
\begin{axis}[scale=0.175, axis line style={draw=none}, tick style={draw=none}, ticks=none, xmin=0, xmax=1.2, ymin=0, ymax=1]
\addplot[thin, mondrianBlue, smooth] plot coordinates
            {
                (0.865881,0.0157939)
                (13/16,1/16*3^.5)
            };
\end{axis}
\begin{axis}[scale=0.175, axis line style={draw=none}, tick style={draw=none}, ticks=none, xmin=0, xmax=1.2, ymin=0, ymax=1]
\addplot[thin, mondrianCyan, smooth] plot coordinates
            {
                (13/16,1/16*3^.5)
                (0.739471,0.23667)
            };
\end{axis}
\begin{axis}[scale=0.175, axis line style={draw=none}, tick style={draw=none}, ticks=none, xmin=0, xmax=1.2, ymin=0, ymax=1]
\addplot[thin, mondrianCyan, smooth] plot coordinates
            {
                (0.807692,0.199852)
                (0.801846,0.203486)
                (0.796025,0.207009)
                (0.79023,0.210424)
                (0.784464,0.213733)
                (0.778726,0.216939)
                (0.773018,0.220043)
                (0.767342,0.223049)
                (0.761699,0.225958)
                (0.756089,0.228772)
                (0.750514,0.231494)
                (0.744974,0.234126)
                (0.739471,0.23667)
            };
\end{axis}
\begin{axis}[scale=0.175, axis line style={draw=none}, tick style={draw=none}, ticks=none, xmin=0, xmax=1.2, ymin=0, ymax=1]
\addplot[thin, mondrianCyan, smooth] plot coordinates
            {
                (0.807692,0.199852)
                (15/22,5/22*3^.5)
            };
\end{axis}
\begin{axis}[scale=0.175, axis line style={draw=none}, tick style={draw=none}, ticks=none, xmin=0, xmax=1.2, ymin=0, ymax=1]
\addplot[thin, mondrianGreen, smooth] plot coordinates
            {
                (15/22,5/22*3^.5)
                (0.576934,0.555431)
            };
\end{axis}
\begin{axis}[scale=0.175, axis line style={draw=none}, tick style={draw=none}, ticks=none, xmin=0, xmax=1.2, ymin=0, ymax=1]
\addplot[thin, mondrianGreen, smooth] plot coordinates
           {
                (0.57729,0.585725)
                (0.577308,0.584374)
                (0.577322,0.583026)
                (0.577334,0.581681)
                (0.577343,0.580339)
                (0.577348,0.579001)
                (0.57735,0.577665)
                (0.577349,0.576333)
                (0.577345,0.575003)
                (0.577339,0.573677)
                (0.577329,0.572353)
                (0.577316,0.571033)
                (0.5773,0.569716)
                (0.577281,0.568402)
                (0.577259,0.567091)
                (0.577234,0.565783)
                (0.577207,0.564478)
                (0.577176,0.563177)
                (0.577143,0.561878)
                (0.577107,0.560582)
                (0.577068,0.55929)
                (0.577026,0.558)
                (0.576981,0.556714)
                (0.576934,0.555431)
            };
\end{axis}
\begin{axis}[scale=0.175, axis line style={draw=none}, tick style={draw=none}, ticks=none, xmin=0, xmax=1.2, ymin=0, ymax=1]
\addplot[thin, mondrianGreen, smooth] plot coordinates
            {
                (0.57729,0.585725)
                (1/2,2/5*3^.5)
            };
\end{axis}
\begin{axis}[scale=0.175, axis line style={draw=none}, tick style={draw=none}, ticks=none, xmin=0, xmax=1.2, ymin=0, ymax=1]
\addplot[thin, mondrianYellow, smooth] plot coordinates
           {
                (1/2,2/5*3^.5)
                    (0.453456,0.763422)
            };
\end{axis}
\begin{axis}[scale=0.175, axis line style={draw=none}, tick style={draw=none}, ticks=none, xmin=0, xmax=1.2, ymin=0, ymax=1]
\addplot[thin, mondrianYellow, smooth] plot coordinates
           {
                    (0.5,0.866025)
                    (0.495099,0.857366)
                    (0.490396,0.848712)
                    (0.485886,0.840067)
                    (0.481567,0.831438)
                    (0.477435,0.822827)
                    (0.473486,0.81424)
                    (0.469718,0.80568)
                    (0.466127,0.797152)
                    (0.462708,0.788658)
                    (0.459459,0.780203)
                    (0.456376,0.77179)
                    (0.453456,0.763422)
            };
\end{axis}

\node at (1.8,0.455) {
\begin{tikzpicture}[scale=4.3]
\draw[line width=.5pt,gray] (0,0) -- (1,0) -- (1/2,1/2*3^.5) -- (0,0);
\node[scale=0.8pt] at (-0.03,-0.03) {$R$};
\node[scale=0.8pt] at (1.03,-0.03) {$B$};
\node[scale=0.8pt] at (1/2,1/2*3^.5+0.03) {$Y$};

\filldraw[fill=mondrianGrey,draw=mondrianGrey,opacity=1] (1/2,19/118*3^.5) -- (1/2,0) -- (20/21,0) -- (1/2,19/118*3^.5);
\filldraw[fill=mondrianGrey,draw=mondrianGrey,opacity=1] (1/2,1/6*3^.5) -- (41/44,1/44*3^.5) -- (225/317,75/317*3^.5) -- (1/2,1/6*3^.5);
\filldraw[fill=mondrianGrey,draw=mondrianGrey,opacity=1] (225/317,75/317*3^.5) -- (129/178,43/178*3^.5) -- (1/2,41/86*3^.5) -- (1/2,143/326*3^.5) -- (225/317,75/317*3^.5);
\filldraw[fill=mondrianGrey,draw=mondrianGrey,opacity=1] (1/2,41/86*3^.5) -- (1/2,1/2*3^.5) -- (43/88,43/88*3^.5) -- (1/2,41/86*3^.5);

\draw[line width=.5pt,dashed,black] (0.421053, 0.182321) -- (0.368421, 0.273482) -- (0.447368,0.410223) -- (0.552632, 0.410223) -- (0.631579, 0.273482) -- (0.578947,0.182321) -- (0.421053, 0.182321);
\draw[line width=.5pt,dashed,black] (0.307692,0.0888231) -- (0.230769, 0.222058) -- (0.423077,0.555144) -- (0.576923, 0.555144) -- (0.769231, 0.222058) -- (0.692308, 0.0888231)-- (0.307692, 0.0888231);
\draw[line width=.5pt,dashed,black] (0.177419, 0.0279363) -- (0.112903, 0.139682) -- (0.435484,0.698408) -- (0.564516, 0.698408) -- (0.887097, 0.139682) -- (0.822581,0.0279363) -- (0.177419, 0.0279363);

\begin{axis}[scale=0.175, axis line style={draw=none}, tick style={draw=none}, ticks=none, xmin=0, xmax=1.2, ymin=0, ymax=1]
\addplot[thin, mondrianBlue, smooth] plot coordinates
            {
            (0.948673,0.00228605)
            (0.902934,0.00821807)
            (0.86017,0.0171888)
            (0.820447,0.0286299)
            (0.783748,0.0420296)
            (0.749993,0.0569371)
            (0.719061,0.0729629)
            (0.6908,0.0897758)
            (0.665042,0.107099)
            (0.641609,0.124703)
            (0.620325,0.142402)
            (0.601014,0.160046)
            (0.583508,0.177515)
            (0.56765,0.194719)
            (0.553291,0.211585)
            (0.540293,0.228061)
            (0.528529,0.24411)
            (0.517885,0.259706)
            (0.508252,0.274831)
            (0.5,0.288675)
            };
\end{axis}
\begin{axis}[scale=0.175, axis line style={draw=none}, tick style={draw=none}, ticks=none, xmin=0, xmax=1.2, ymin=0, ymax=1]
\addplot[thin, mondrianBlue, smooth] plot coordinates
            {
            (0.948673,0.00228605)
            (0.5,0.288675)
            };
\end{axis}
\begin{axis}[scale=0.175, axis line style={draw=none}, tick style={draw=none}, ticks=none, xmin=0, xmax=1.2, ymin=0, ymax=1]
\addplot[thin, mondrianCyan, smooth] plot coordinates
            {
            (0.855997,0.165857)
            (0.826204,0.187691)
            (0.796849,0.206515)
            (0.768146,0.222628)
            (0.740249,0.236314)
            (0.713271,0.24784)
            (0.687283,0.257455)
            (0.662329,0.265385)
            (0.638429,0.271834)
            (0.615583,0.276987)
            (0.593779,0.281008)
            (0.572993,0.284042)
            (0.553196,0.286219)
            (0.534351,0.287652)
            (0.51642,0.288442)
            };
\end{axis}
\begin{axis}[scale=0.175, axis line style={draw=none}, tick style={draw=none}, ticks=none, xmin=0, xmax=1.2, ymin=0, ymax=1]
\addplot[thin, mondrianCyan, smooth] plot coordinates
            {
            (0.855997,0.165857)
            (225/317,75/317*3^.5)
            };
\end{axis}
\begin{axis}[scale=0.175, axis line style={draw=none}, tick style={draw=none}, ticks=none, xmin=0, xmax=1.2, ymin=0, ymax=1]
\addplot[thin, mondrianGreen, smooth] plot coordinates
            {
            (225/317,75/317*3^.5)
            (0.574168,0.637887)
            };
\end{axis}
\begin{axis}[scale=0.175, axis line style={draw=none}, tick style={draw=none}, ticks=none, xmin=0, xmax=1.2, ymin=0, ymax=1]
\addplot[thin, mondrianGreen, smooth] plot coordinates
            {
            (0.560381,0.716298)
            (0.561551,0.711492)
            (0.562672,0.706709)
            (0.563745,0.701949)
            (0.564771,0.697212)
            (0.565751,0.692498)
            (0.566685,0.687809)
            (0.567575,0.683144)
            (0.56842,0.678504)
            (0.569222,0.673888)
            (0.569982,0.669297)
            (0.570699,0.664732)
            (0.571376,0.660193)
            (0.572012,0.655679)
            (0.572609,0.651192)
            (0.573166,0.64673)
            (0.573686,0.642295)
            (0.574168,0.637887)
            };
\end{axis}
\begin{axis}[scale=0.175, axis line style={draw=none}, tick style={draw=none}, ticks=none, xmin=0, xmax=1.2, ymin=0, ymax=1]
\addplot[thin, mondrianGreen, smooth] plot coordinates
            {
            (0.560381,0.716298)
            (1/2,41/86*3^.5)
            };
\end{axis}
\begin{axis}[scale=0.175, axis line style={draw=none}, tick style={draw=none}, ticks=none, xmin=0, xmax=1.2, ymin=0, ymax=1]
\addplot[thin, mondrianYellow, smooth] plot coordinates
            {
            (1/2,41/86*3^.5)
            (0.489479,0.846982)
            };
\end{axis}
\begin{axis}[scale=0.175, axis line style={draw=none}, tick style={draw=none}, ticks=none, xmin=0, xmax=1.2, ymin=0, ymax=1]
\addplot[thin, mondrianYellow, smooth] plot coordinates
            {
            (0.5,0.866025)
            (0.499004,0.864293)
            (0.498016,0.862561)
            (0.497036,0.860829)
            (0.496064,0.859098)
            (0.495099,0.857366)
            (0.494143,0.855635)
            (0.493195,0.853903)
            (0.492254,0.852172)
            (0.491321,0.850442)
            (0.490396,0.848712)
            (0.489479,0.846982)
            };
\end{axis}

\end{tikzpicture}};

\node at (3.1,0.45) {
\begin{tikzpicture}[scale=4.3]
\draw[line width=.5pt,gray] (0,0) -- (1,0) -- (1/2,1/2*3^.5) -- (0,0);
\node[scale=0.8pt] at (-0.03,-0.03) {$R$};
\node[scale=0.8pt] at (1.03,-0.03) {$B$};
\node[scale=0.8pt] at (1/2,1/2*3^.5+0.03) {$Y$};
\filldraw[fill=mondrianGrey,draw=mondrianGrey,opacity=1] (0,0) -- (1/4,1/4*3^.5) -- (1/2,1/6*3^.5) -- (0,0);

\draw[line width=.5pt,dashed,black] (0.421053, 0.182321) -- (0.368421, 0.273482) -- (0.447368,0.410223) -- (0.552632, 0.410223) -- (0.631579, 0.273482) -- (0.578947,0.182321) -- (0.421053, 0.182321);
\draw[line width=.5pt,dashed,black] (0.307692,0.0888231) -- (0.230769, 0.222058) -- (0.423077,0.555144) -- (0.576923, 0.555144) -- (0.769231, 0.222058) -- (0.692308, 0.0888231)-- (0.307692, 0.0888231);
\draw[line width=.5pt,dashed,black] (0.177419, 0.0279363) -- (0.112903, 0.139682) -- (0.435484,0.698408) -- (0.564516, 0.698408) -- (0.887097, 0.139682) -- (0.822581,0.0279363) -- (0.177419, 0.0279363);

\begin{axis}[scale=0.175, axis line style={draw=none}, tick style={draw=none}, ticks=none, xmin=0, xmax=1.2, ymin=0, ymax=1]
\addplot[thin, mondrianOrange, smooth] plot coordinates
            {
            (0.,0.)
            (0.0540541,0.0780203)
            (0.112903,0.139682)
            (0.172662,0.186912)
            (0.230769,0.222058)
            (0.285714,0.247436)
            (0.336735,0.26511)
            (0.383562,0.276812)
            (0.42623,0.283943)
            (0.464945,0.28761)
            (0.5,0.288675)
            };
\end{axis}
\end{tikzpicture}};

\node at (0.545,-0.58) {
\begin{tikzpicture}[scale=4.3]
\draw[line width=.5pt,gray] (0,0) -- (1,0) -- (1/2,1/2*3^.5) -- (0,0);
\node[scale=0.8pt] at (-0.03,-0.03) {$R$};
\node[scale=0.8pt] at (1.03,-0.03) {$B$};
\node[scale=0.8pt] at (1/2,1/2*3^.5+0.03) {$Y$};

\filldraw[fill=mondrianGrey,draw=mondrianGrey,opacity=1] (1/2,1/6*3^.5) -- (1/2,0) -- (7/8,0) -- (13/16,1/16*3^.5) -- (1/2,1/6*3^.5);
\filldraw[fill=mondrianGrey,draw=mondrianGrey,opacity=1] (13/16,1/16*3^.5) -- (9/10,1/30*3^.5) -- (15/22,5/22*3^.5) -- (21/32,7/31*3^.5) -- (13/16,1/16*3^.5);
\filldraw[fill=mondrianGrey,draw=mondrianGrey,opacity=1] (15/22,5/22*3^.5) -- (12/17,4/17*3^.5) -- (1/2,2/5*3^.5) -- (1/2,7/18*3^.5) -- (15/22,5/22*3^.5);
\filldraw[fill=mondrianGrey,draw=mondrianGrey,opacity=1] (1/2,2/5*3^.5) -- (4/9,4/9*3^.5) -- (1/2,1/2*3^.5) -- (1/2,2/5*3^.5);

\begin{axis}[scale=0.175, axis line style={draw=none}, tick style={draw=none}, ticks=none, xmin=0, xmax=1.2, ymin=0, ymax=1]
\addplot[thin, black, smooth] plot coordinates
            {
                (0.5,0.288675)
                (0.583377,0.19728)
                (0.652938,0.146413)
                (0.701555,0.120386)
                (0.735454,0.106727)
                (0.759733,0.0997263)
                (0.777494,0.0967615)
                (0.790605,0.0965349)
                (0.800222,0.0983639)
                (0.807088,0.101875)
                (0.811695,0.10686)
                (0.814377,0.113205)
                (0.815366,0.120849)
                (0.814831,0.129756)
                (0.8129,0.139902)
                (0.809678,0.151259)
                (0.80526,0.163788)
                (0.79974,0.177433)
                (0.793213,0.192114)
                (0.785784,0.207729)
                (0.777564,0.224155)
                (0.768673,0.241252)
                (0.759232,0.258871)
                (0.749361,0.276858)
                (0.739178,0.295061)
                (0.728792,0.313342)
                (0.718301,0.331571)
                (0.707791,0.349639)
                (0.697336,0.367453)
                (0.686995,0.384938)
                (0.676817,0.402038)
                (0.666837,0.41871)
                (0.657085,0.434926)
                (0.647576,0.45067)
                (0.638324,0.465935)
                (0.629334,0.480723)
                (0.620605,0.495041)
                (0.466026,0.777243)
                (1/2,1/2*3^.5)
            };
\end{axis}
\end{tikzpicture}};

\node at (1.8,-0.58) {
\begin{tikzpicture}[scale=4.3]
\draw[line width=.5pt,gray] (0,0) -- (1,0) -- (1/2,1/2*3^.5) -- (0,0);
\node[scale=0.8pt] at (-0.03,-0.03) {$R$};
\node[scale=0.8pt] at (1.03,-0.03) {$B$};
\node[scale=0.8pt] at (1/2,1/2*3^.5+0.03) {$Y$};

\filldraw[fill=mondrianGrey,draw=mondrianGrey,opacity=1] (1/2,19/118*3^.5) -- (1/2,0) -- (20/21,0) -- (1/2,19/118*3^.5);
\filldraw[fill=mondrianGrey,draw=mondrianGrey,opacity=1] (1/2,1/6*3^.5) -- (41/44,1/44*3^.5) -- (225/317,75/317*3^.5) -- (1/2,1/6*3^.5);
\filldraw[fill=mondrianGrey,draw=mondrianGrey,opacity=1] (225/317,75/317*3^.5) -- (129/178,43/178*3^.5) -- (1/2,41/86*3^.5) -- (1/2,143/326*3^.5) -- (225/317,75/317*3^.5);
\filldraw[fill=mondrianGrey,draw=mondrianGrey,opacity=1] (1/2,41/86*3^.5) -- (1/2,1/2*3^.5) -- (43/88,43/88*3^.5) -- (1/2,41/86*3^.5);

\begin{axis}[scale=0.175, axis line style={draw=none}, tick style={draw=none}, ticks=none, xmin=0, xmax=1.2, ymin=0, ymax=1]
\addplot[thin, black, smooth] plot coordinates
            {
                    (0.5,0.288675)
                    (0.728073,0.158438)
                    (0.795483,0.120658)
                    (0.827643,0.103093)
                    (0.846573,0.0931256)
                    (0.85901,0.0869133)
                    (0.867731,0.0828807)
                    (0.874097,0.0802636)
                    (0.878854,0.0786519)
                    (0.882441,0.0778139)
                    (0.885128,0.0776197)
                    (0.887085,0.078006)
                    (0.888409,0.0789618)
                    (0.889145,0.0805286)
                    (0.889283,0.0828164)
                    (0.888734,0.0860525)
                    (0.887265,0.0907177)
                    (0.884216,0.0980576)
                    (0.876106,0.114472)
                    (0.873481,0.120683)
                    (0.536308,0.747285)
                    (0.533923,0.752662)
                    (0.531348,0.757827)
                    (0.528947,0.762474)
                    (0.527688,0.765346)
                    (0.526554,0.767856)
                    (0.525444,0.770187)
                    (0.524772,0.771784)
                    (0.524261,0.773049)
                    (0.523744,0.774236)
                    (0.529001,0.763531)
                    (0.52597,0.77105)
                    (0.524504,0.774395)
                    (0.52294,0.776902)
                    (0.523101,0.776293)
                    (0.521316,0.778177)
                    (0.522053,0.776091)
                    (0.521496,0.775754)
                    (0.522779,0.772575)
                    (0.522866,0.771101)
                    (0.523234,0.769167)
                    (0.523778,0.76695)
                    (0.524721,0.764126)
                    (0.525458,0.761615)
                    (0.526549,0.758584)
                    (0.518685,0.775935)
                    (0.503995,0.813074)
                (1/2,1/2*3^.5)
            };
\end{axis}

\end{tikzpicture}};

\node at (3.1,-0.58) {
\begin{tikzpicture}[scale=4.3]
\draw[line width=.5pt,gray] (0,0) -- (1,0) -- (1/2,1/2*3^.5) -- (0,0);
\node[scale=0.8pt] at (-0.03,-0.03) {$R$};
\node[scale=0.8pt] at (1.03,-0.03) {$B$};
\node[scale=0.8pt] at (1/2,1/2*3^.5+0.03) {$Y$};
\filldraw[fill=mondrianGrey,draw=mondrianGrey,opacity=1] (0,0) -- (1/4,1/4*3^.5) -- (1/2,1/6*3^.5) -- (0,0);

\begin{axis}[scale=0.175, axis line style={draw=none}, tick style={draw=none}, ticks=none, xmin=0, xmax=1.2, ymin=0, ymax=1]
\addplot[thin, black, smooth] plot coordinates
            {
            (0.5,0.288675)
            (0.370497,0.293606)
            (0.234593,0.286479)
            (0.16804,0.261241)
            (0.139934,0.235546)
            (0.124468,0.214086)
            (0.113651,0.196533)
            (0.105114,0.181997)
            (0.0980177,0.169759)
            (0.0919702,0.159295)
            (0.0867353,0.150229)
            (0.0821502,0.142288)
            (0.0780951,0.135265)
            (0.0744787,0.129001)
            (0.0712303,0.123374)
            (0.0682937,0.118288)
            (0.065624,0.113664)
            (0.0631846,0.109439)
            (0.0609457,0.105561)
            (0.0588822,0.101987)
            (0.0569735,0.0986809)
            (0.0552018,0.0956122)
            (0.0535522,0.0927551)
            (0.0520119,0.0900873)
            (0.05057,0.0875897)
            (0.0492167,0.0852459)
            (0.0479439,0.0830413)
            (0.0467442,0.0809633)
            (0.0456111,0.0790008)
            (0.0445391,0.077144)
            (0.0435231,0.0753842)
            (0.0425586,0.0737136)
            (0.0416415,0.0721253)
            (0.0407685,0.070613)
            (0,0)
            };
\end{axis}
\end{tikzpicture}};

\end{tikzpicture}
\vspace*{-4mm}
\caption{The colored curves in the top panels show $\mu$ equilibria for $\mu=(1,\varepsilon,\varepsilon^2)/(1+\varepsilon+\varepsilon^2)$ and $\varepsilon\in[0,1]$. The three permutahedra shown correspond to $\varepsilon=\deel{2}{3}$, $\varepsilon=\deel{2}{5}$, and $\varepsilon=\deel{1}{5}$. The black curves in the bottom panels show logit-QRE for $\lambda\in[0,\infty)$. The left panels correspond to the DS games, the middle panels to the ``no logit'' game, and the right panels to the KM game. The grey areas indicate $M$-choice sets.}\label{epsEqgraphs}
\end{center}
\vspace*{-8mm}
\end{figure}

However, a crucial difference between $\mu$ equilibrium and logit-QRE is that the latter is unique for any value of $\lambda$. In contrast, there are ranges of $\varepsilon$ for which multiple $\mu$ equilibria exist. Consider, for instance, the permutahedron for $\varepsilon=\deel{2}{5}$ in the top-left panel. There are three Nash equilibria corresponding to the permutahedron's vertices that lie in the blue, cyan, and green regions. These Nash equilibria occur for the \textit{same} value of the rationality parameter, i.e. $\mu$ equilibrium shares the equilibrium multiplicity that $M$ equilibrium possesses. Multiple equilibria also occur for the ``no logit'' game shown in the middle panels, and there \textit{is} a $\mu$ equilibrium curve in the blue region that logit-QRE cannot enter.

Uniqueness of logit-QRE in these games hinders its ability to fit the choice data. For the DS2 game, for instance, the heterogeneous data are scattered around the centroid, corresponding to different levels of $\lambda$. The best logit-QRE can do is to predict the centroid itself, which corresponds to a low value of the rationality parameter $\lambda$, see Table \ref{estQRE}. To a lesser extent, a similar phenomenon occurs in the ``no logit'' game. As a result, the (pooled) loglikelihood from this fitting exercise is low.

In contrast, there are multiple $\mu$ equilibria for intermediate values of $\varepsilon$ and an individual's choices can be matched to the closest $\mu$ equilibrium (where the same equilibrium is used for all choices the individual made in a particular game). The results in Table \ref{estMu} show that this results in an improved fit of the choice data and consistent estimates of the rationality parameter, $\varepsilon$. The reason for this improved performance is that $\mu$-equilibrium inherits equilibrium multiplicity from $M$ equilibrium, which naturally results in heterogeneous choices.

Better fitting models sometimes perform ``too well.'' A litmus test for over-fitting is to use the estimated parameters from one game to predict behavior in the other games.\footnote{For examples of such out-of-sample testing see \cite{wright2017} and \cite{goeree2017}.} The last four columns of Tables \ref{estQRE} and \ref{estMu} show the out-of-sample loglikelihoods obtained by using the estimated values for $\lambda$ and $\varepsilon$ from fitting the data of each game to predict choices in the other three games.  Only in four out of twelve cases does logit-QRE predict better than purely random choice, while it performs worse in five cases and the same as random choice in three cases. Again the worst performance is for the DS2 and the ``no logit'' games, as the homogeneity logit-QRE imposes is not borne out by the choice data in these games. In contrast, $\mu$ equilibrium does better than random choice in eleven out of twelve cases. The $\varepsilon$ estimated from fitting the ``no logit'' choices is such that there is a unique $\mu$ equilibrium in the DS games. This results in a reasonable out-of-sample performance for the relatively homogeneous choices in the DS1 game, but cannot capture the heterogeneity of the DS2 choice data.

The fit and out-of-sample results for $\mu$ equilibrium show its promise as a structural model for the analysis of experimental data. Like $M$ equilibrium it is easy to compute and it can produce heterogeneous choices even in games with a unique Nash equilibrium. Being parametric, it offers a more direct comparison between the empirical relevance of monotonicity and consequential unbiasedness versus the assumptions underlying other behavioral-game-theory models.

\begin{table}[t]
\begin{center}
\begin{tabular}{c|c|c|c|c|c|c|c}
\hline\hline
\multirow{2}{*}{Game} & \multirow{2}{*}{\# Obs} & \multirow{2}{*}{$\lambda$} & logL & \multicolumn{4}{c}{Out-of-sample logL}\\
 & & & (fit) & DS1 & DS2 & NL & KM \\ \hline
DS1 & 480 & .0376 & -470.0 & & -616.0 & -543.5 & -203.2 \\
DS2 & 480 & .0000 & -527.3 & -527.3 & & -527.3 & -263.7 \\
NL & 480 & .0078 & -496.2 & -501.4 & -540.5 & & -250.8 \\
KM & 240 & .0696 & -183.3 & -480.3 & -678.2 & -582.9 & \\ \hline
Pooled & 1680 & .0125 & -1784.3 & -490.4 & -550.8 & -499.9 & -243.1\\ \hline\hline
\end{tabular}
\caption{Fit and out-of-sample results for logit-QRE.}\label{estQRE}
\end{center}
\vspace*{-6mm}
\end{table}

\begin{table}[t]
\begin{center}
\begin{tabular}{c|c|c|c|c|c|c|c}
\hline\hline
\multirow{2}{*}{Game} & \multirow{2}{*}{\# Obs} & \multirow{2}{*}{$\varepsilon$} & logL & \multicolumn{4}{c}{Out-of-sample logL}\\
 & & & (fit) & DS1 & DS2 & NL & KM \\ \hline
DS1 & 480 & .382 & -429.8 & & -450.0 & -479.9 & -178.7 \\
DS2 & 480 & .378 & -449.8 & -429.9 & & -480.6 & -178.4 \\
NL & 480 & .522 & -468.3 & -468.5 & -640.0 & & -192.4 \\
KM & 240 & .290 & -175.3 & -445.5 & -465.3 & -506.0 & \\ \hline
Pooled & 1680 & .384 & -1529.9 & -429.8 & -450.1 & -479.5 & -178.8\\ \hline\hline
\end{tabular}
\caption{Fit and out-of-sample results for $\mu$ equilibrium.}\label{estMu}
\end{center}
\vspace*{-8mm}
\end{table}

\section{Conclusions}
\label{sec:conc}

\subsection{Methodological Contribution}

Two assumptions underlie Nash equilibrium: perfect maximization, $\sigma\in\br(\pi(\omega))$, and perfect foresight, $\omega=\sigma$. These stark assumptions are not unique to game theory. In macroeconomics, they have given impetus to the \q{rational-expectations revolution} that started sixty years ago. The dynamic stochastic general equilibrium models that have since come to dominate macroeconomics and finance hinge on the assumptions of optimal choices and correct expectations.

In his pioneering paper that introduced the concept, John Muth, the \q{father} of the rational-expectations revolution, remarked:
\vspace*{3mm}
\begin{center}
\parbox{14.5cm}{\setlength{\baselineskip}{6mm}
\q{\textit{To make dynamic economic models complete various expectational formulae have been used. There is, however, little evidence to suggest that the presumed relations bear a resemblance to the way the economy works},} \citeauthor{Muth1961} (\citeyear[p. 315]{Muth1961}).}
\end{center}
\vspace*{3mm}

\noindent The belief data reported in this paper (and in many papers prior) confirm Muth's suspicion \textit{even} in the context of playing simple two-player normal-form games. More generally, as evidenced in Section \ref{empMotivation}, data from over half a century of game-theory experiments strongly reject the Nash assumptions of perfect maximization and perfect foresight.

The quest for an empirically relevant game theory has produced noisy decisionmaking models that apply \q{soft} maximization rules and keep the correct-beliefs assumption (e.g. QRE), models that explicitly model non-equilibrium beliefs but keep the best-reply assumption (e.g. level-$k$), and models that relax both assumptions (e.g. Noisy Introspection). See Section \ref{sec:prior} for further details. \citeauthor{Sargent1999}'s (\citeyear{Sargent1999}) fear that abandoning rational expectations might lead to a \q{wilderness of behavioral models} has come true to some extent.

This paper offers a distinct approach that is parameter-free and detail-free. First, $M$-equilibrium predictions are based solely on the variables that define a game -- number of players, strategy sets, and payoffs -- just like the Nash equilibrium. In other words, $M$ equilibrium has no free parameters that require estimation from the data. More importantly, $M$ equilibrium does not make any detailed assumptions about how players form beliefs or how they translate beliefs into choices. $M$ equilibrium circumvents Sargent's hypothesized wilderness by imposing only minimal ordinal conditions on beliefs and choices.

The modeling assumptions underlying $M$ equilibrium stem directly from John Muth's admission for why he proposed rational expectations, a theoretical construct he felt had little to do with reality. It was to \q{complete} or \q{close} the model. Confronted with a high-dimensional system in which the optimal levels of current variables depend on expectations about their future values, which, in turn, are shaped by their current realizations, Muth cut the Gordian knot by assuming expectations were correct. Against his better empirical judgement, Muth imposed rational expectations to obtain a \q{complete} or \q{closed} set of fixed-point equations from which the variables of interest could be solved.

$M$ equilibrium neither embraces the specificity that many behavioral-game-theory models impose to achieve realism nor does it embrace the unrealistic mathematical convenience of standard game theory. Rather, $M$ equilibrium replaces the stark assumptions underlying Nash equilibrium with two behavioral postulates that are \textit{plausible} even if they seem \q{weak} because they avoid ad hoc modeling choices and even if they appear mathematically \q{inconvenient.} Specifically, $M$ equilibrium replaces perfect maximization with \textit{monotonicity}
\begin{displaymath}
  \pi_{ij}(\omega_i)\,<\,\pi_{ik}(\omega_i)\,\Rightarrow\sigma_{ij}\,<\,\sigma_{ik}
\end{displaymath}
and perfect foresight with \textit{consequential unbiasedness}
\begin{displaymath}
  \pi_{ij}(\sigma_{-i})\,<\,\pi_{ik}(\sigma_{-i})\,\Leftrightarrow\,\pi_{ij}(\omega_i)\,<\,\pi_{ik}(\omega_i)
\end{displaymath}
where $j,k$ label different strategies for player $i\in N$. In words, monotonicity accommodates non-best-response behavior such that more costly mistakes are less likely and consequential unbiasedness allows beliefs to be biased as long as the payoff predictions they generate are not.

\subsection{Theoretical Contribution}

The price we pay for plausibility is that the above set of inequalities is rather different from the closed set of fixed-point conditions underlying Nash equilibrium. However, as both the number of players and their strategy spaces are finite, there are \textit{finitely} many of them and each is \textit{polynomial} in the choice or belief variables. $M$ equilibrium sets fit the canonical definition of \textit{semi-algebraic sets} that can be computed (see Section \ref{sec:M} and Appendix A), decomposed into disjoint sets characterized by the number of payoff indifferences (Section \ref{sec:rankChar}), and their sizes can be bounded to establish falsifiability of $M$ equilibrium (Section \ref{sec:robust}).\footnote{The few existing applications of semi-algebraic geometry to game theory, e.g. \cite{KohlbergMertens1986}, \cite{SchanuelSimonZame1991}, and \cite{BlumeZame1994}), focus on non-generic games to study refinements of Nash equilibria. \cite{BlumeZame1994}, for instance, use semi-algebraic geometry to establish equivalence between \citeauthor{Selten1975}'s (\citeyear{Selten1975}) \textit{perfect equilibrium} and \citeauthor{KrepsWilson1982}'s (\citeyear{KrepsWilson1982}) \textit{sequential equilibrium} refinements. These prior papers use semi-algebraicity mainly to establish limit arguments. In contrast, in this paper we are interested in determining entire (full-dimensional) semi-algebraic sets. For applications of semi-algebraic geometry to general equilibrium, see, e.g., \cite{BlumeZame1992} and \cite{KublerSchmedders2010}.}

Surprisingly, while $M$ equilibrium is not a fixed-point theory as it does not impose rational expectations, it serves as a meta theory for seemingly unrelated models in the game-theory literature that do, see Section \ref{sec:Meta}. In particular, the $M$-equilibrium choice sets capture all Nash equilibria of games with restricted strategy sets (as in \citeauthor{Selten1975}, \citeyear{Selten1975}), all $\varepsilon$-proper equilibria (\citeauthor{Myerson1978}, \citeyear{Myerson1978}), and all QRE (\citeauthor{McKelveyPalfrey1995}, \citeyear{McKelveyPalfrey1995}), establishing unexpected equivalences between these very different approaches. Together with falsifiability of $M$-equilibrium, this allows us to refute the critique by \cite{HaileHortacsuKosenok2008} about the lack of empirical content of QRE (or other models that satisfy monotonicity). Finally, and again surprisingly, while the computation of any of the aforementioned fixed-point models can be cumbersome, computing the $M$-choice set that contains all of them can be quite simple.

$M$ equilibrium naturally gives rise to a refinement of Nash equilibria, which we term \textit{profect} equilibria. Like \citeauthor{Myerson1978}'s (\citeyear{Myerson1978}) properness, profectness is based on the idea that more costly mistakes occur less likely, a condition that may be violated by \citeauthor{Selten1975}'s (\citeyear{Selten1975}) perfectness. However, unlike properness, profectness does not require that more costly mistakes are infinitely less likely.  Hence, the set of profect equilibria contains all proper equilibria and is a subset of the set of perfect equilibria. In Section \ref{sec:Nash} we provide an example where profectness is arguably a more reasonable refinement than either perfectness or properness.

\subsection{Empirical Contribution}

$M$ equilibrium provides a novel perspective on actual play in games. For example, in games with a unique Nash equilibrium but multiple $M$ equilibria it suggests that there may be strategic mis-coordination when players' beliefs belong to different supporting belief sets. In other cases, a robust $M$ equilibrium may be a good predictor of behavior even in games where classic game theoretical reasoning expects instability. Section \ref{sec:exp-test} provides a number of examples where observed behavior in experiments cannot be reconciled with Nash predictions or those of existing behavioral game-theory models such as level-$k$ and variants of QRE. In contrast, $M$-equilibrium captures subjects' choices and beliefs.

More generally, being set-valued, $M$ equilibrium is immune to a number of issues typically encountered in experiments, including heterogeneity in subjects' behavior (both within and across subjects), order effects \citep[e.g.][]{charness2012}, framing effects \citep[e.g.][]{tversky1981}, and session effects \citep[e.g.][]{frechette2012}. Moreover, it offers a unified framework to study choices and beliefs, both of which can be reliably elicited.\footnote{\cite{danzVesterlundWilson2020} show that eliciting beliefs under the BSR without explicit information on quantitative incentives, similar to what we do in our experiment, results in high levels of accuracy.  See also \cite{holt2016} for an alternative method with similar levels of success.}

Because of its non-parametric nature, $M$ equilibrium is not amenable to calibration exercises. We introduce a closely related parametric fixed-point model, $\mu$ equilibrium, that can be applied to the data. We show that $\mu$ equilibrium shares key features with $M$ equilibrium: $\mu$ equilibrium profiles are supported by a set of beliefs, they can be calculated easily and analytically (and geometrically), and there may be multiple $\mu$ equilibria for a given value of the rationality parameter even in games with a unique Nash equilibrium. In Section \ref{subsec:empeval} we apply $\mu$ equilibrium to the data from symmetric $3\times 3$ games and show that it outperforms the commonly used logit-QRE in terms of fit and out-of-sample predictive power.

\subsection{Outlook}

Despite imposing weak restrictions on the data, we expect $M$ equilibrium to be falsified in certain games.
There are well-documented cases where psychological factors \citep[e.g.][]{rabin2013} and other behavioral factors such as other-regarding preferences \citep[e.g.][]{FehrSchmidt1999,bolton2000,CharnessRabin2002}, risk aversion \citep[e.g.][]{goeree2003}, etc. play an important role. These factors can easily be incorporated in an extension of the theory by replacing expected payoffs with expected utilities. For example, Figure \ref{CR} shows the $M$-choice and $M$-belief sets for the DS1 (top) and DS2 (bottom) games in Table \ref{3by3exp} (cf. Figure \ref{expdataDS}) when selfish preferences are replaced with other-regarding preferences based on the model proposed by \cite{CharnessRabin2002}.\footnote{The Charness-Rabin model specifies that a player's utility is given by
\begin{displaymath}
u(\pi_{own},\pi_{other})\,=\,\left\{\begin{array}{ccc} (1-\alpha)\pi_{own}+\alpha\,\pi_{other} &\mbox{if}& \pi_{own}>\pi_{other}\\
(1-\beta)\pi_{own}+\beta\,\pi_{other} &\mbox{if}& \pi_{own}<\pi_{other}\end{array}\right.
\end{displaymath}
where $\pi_{own}$ denotes own payoff and $\pi_{other}$ the other's payoff. Figure \ref{CR} is based on $\alpha=\deel{1}{3}$ and $\beta=-\deel{1}{7}$.} This adaptation results in near-perfect $M$ equilibrium predictions. The color of the $M$ set that the belief cluster belongs to matches the color of the $M$ set the corresponding choice cluster belongs to for both the DS1 and DS2 games.

\begin{figure}[t]
\vspace*{-10mm}
\begin{center}

\begin{tikzpicture}[scale=5]
\draw[line width=.5pt,gray] (0,0) -- (1,0) -- (1/2,1/2*3^.5) -- (0,0);
\node[scale=0.8pt] at (-0.03,-0.03) {$R$};
\node[scale=0.8pt] at (1.03,-0.03) {$B$};
\node[scale=0.8pt] at (1/2,1/2*3^.5+0.03) {$Y$};

\filldraw[fill=mondrianBlue,draw=mondrianBlue,opacity=1] (1/2,1/6*3^.5) -- (1/2,0) -- (50/63,0) -- (283/400,39/400*3^.5) -- (1/2,1/6*3^.5);
\filldraw[fill=mondrianCyan,draw=mondrianCyan,opacity=1] (283/400,39/400*3^.5) -- (1,0) -- (681/1010,227/1010*3^.5) -- (225/367,75/367*3^.5) -- (283/400,39/400*3^.5);
\filldraw[fill=mondrianGreen,draw=mondrianGreen,opacity=1] (681/1010,227/1010*3^.5) -- (3/4,1/4*3^.5) -- (85/124,39/124*3^.5) -- (1/2,29/81*3^.5) -- (1/2,227/658*3^.5) -- (681/1010,227/1010*3^.5);
\filldraw[fill=mondrianYellow,draw=mondrianYellow,opacity=1] (1/2,29/81*3^.5) -- (1/2,1/2*3^.5) -- (75/178,75/178*3^.5) -- (0.4714,0.6318)-- (1/2,29/81*3^.5);

\filldraw [fill=cyan,draw=cyan,opacity=1] (.6250, .1580) circle (0.31^.5*.075);
\filldraw [fill=green,draw=green,opacity=1] (.7059, .2547) circle (0.07^.5*.075);
\filldraw [fill=red,draw=red,opacity=1] (.7254, .1098) circle (0.15^.5*.075);
\filldraw [fill=gray,draw=gray,opacity=1] (.6860, .0604) circle (0.09^.5*.075);
\filldraw [fill=magenta,draw=magenta,opacity=1] (.6316, .5470) circle (0.04^.5*.075);
\filldraw [fill=orange,draw=orange,opacity=1] (.5914, .0931) circle (0.19^.5*.075);
\filldraw [fill=pink,draw=pink,opacity=1] (.5000, .0241) circle (0.15^.5*.075);

\node at (2,1/4*3^.5) {
\begin{tikzpicture}[scale=5]
\draw[line width=.5pt,gray] (0,0) -- (1,0) -- (1/2,1/2*3^.5) -- (0,0);
\node[scale=0.8pt] at (-0.03,-0.03) {$R$};
\node[scale=0.8pt] at (1.03,-0.03) {$B$};
\node[scale=0.8pt] at (1/2,1/2*3^.5+0.03) {$Y$};

\filldraw[fill=mondrianBlue,draw=mondrianBlue,opacity=1] (0,0) -- (50/63,0) --(0.4714,0.6318) -- (77/200,77/200*3^.5) -- (0,0);
\filldraw[fill=mondrianYellow,draw=mondrianYellow,opacity=1] (75/178,75/178*3^.5) --(0.4714,0.6318) -- (85/124,39/124*3^.5) -- (1/2,1/2*3^.5) -- (75/178,75/178*3^.5);
\filldraw[fill=mondrianGreen,draw=mondrianGreen,opacity=1] (0.4714,0.6318) -- (85/124,39/124*3^.5) -- (1,0) -- (0.4714,0.6318);
\filldraw[fill=mondrianCyan,draw=mondrianCyan,opacity=1] (0.4714,0.6318) -- (1,0) -- (50/63,0) -- (0.4714,0.6318);

\filldraw [fill=cyan,draw=cyan,opacity=1] (.5875, .1485) circle (0.31^.5*.075);
\filldraw [fill=green,draw=green,opacity=1] (.7338, .2741) circle (0.07^.5*.075);
\filldraw [fill=red,draw=red,opacity=1] (.7740, .0453) circle (0.15^.5*.075);
\filldraw [fill=gray,draw=gray,opacity=1] (.1931, .0900) circle (0.09^.5*.075);
\filldraw [fill=magenta,draw=magenta,opacity=1] (.5268, .5853) circle (0.04^.5*.075);
\filldraw [fill=orange,draw=orange,opacity=1] (.4449, .2669) circle (0.19^.5*.075);
\filldraw [fill=pink,draw=pink,opacity=1] (.4031, .0618) circle (0.15^.5*.075);

\end{tikzpicture}};

\node at (1/2,-0.7) {
\begin{tikzpicture}[scale=5]
\draw[line width=.5pt,gray] (0,0) -- (1,0) -- (1/2,1/2*3^.5) -- (0,0);
\node[scale=0.8pt] at (-0.03,-0.03) {$R$};
\node[scale=0.8pt] at (1.03,-0.03) {$B$};
\node[scale=0.8pt] at (1/2,1/2*3^.5+0.03) {$Y$};

\filldraw[fill=mondrianBlue,draw=mondrianBlue,opacity=1] (1/2,1/6*3^.5) -- (1/2,0) -- (135/146,0) -- (241/340,33/340*3^.5) -- (1/2,1/6*3^.5);
\filldraw[fill=mondrianCyan,draw=mondrianCyan,opacity=1] (241/340,33/340*3^.5) -- (1,0) -- (229/248,19/248*3^.5) -- (531/880,177/880*3^.5) -- (1215/2288,405/2288*3^.5) -- (241/340,33/340*3^.5);
\filldraw[fill=mondrianGreen,draw=mondrianGreen,opacity=1] (531/880,177/880*3^.5) -- (3/4,1/4*3^.5) -- (59/104,45/103*3^.5) -- (1/2,27/61*3^.5) -- (1/2,7/29*3^.5) -- (531/880,177/880*3^.5);
\filldraw[fill=mondrianYellow,draw=mondrianYellow,opacity=1] (1/2,27/61*3^.5) -- (1/2,1/2*3^.5) -- (9/20,9/20*3^.5) -- (1/2,27/61*3^.5);

\filldraw [fill=orange,draw=orange,opacity=1] (.4886, .8463) circle (0.09^.5*.075);
\filldraw [fill=red,draw=red,opacity=1] (.6067, .2194) circle (0.16^.5*.075);
\filldraw [fill=cyan,draw=cyan,opacity=1] (.5992, .3636) circle (0.27^.5*.075);
\filldraw [fill=green,draw=green,opacity=1] (.6636, .2992) circle (0.11^.5*.075);
\filldraw [fill=pink,draw=pink,opacity=1] (.5071, .6804) circle (0.15^.5*.075);
\filldraw [fill=gray,draw=gray,opacity=1] (.4896, .1624) circle (0.10^.5*.075);
\filldraw [fill=magenta,draw=magenta,opacity=1] (.5439, .4558) circle (0.12^.5*.075);

\end{tikzpicture}};

\node at (2,-0.7) {
\begin{tikzpicture}[scale=5]
\draw[line width=.5pt,gray] (0,0) -- (1,0) -- (1/2,1/2*3^.5) -- (0,0);
\node[scale=0.8pt] at (-0.03,-0.03) {$R$};
\node[scale=0.8pt] at (1.03,-0.03) {$B$};
\node[scale=0.8pt] at (1/2,1/2*3^.5+0.03) {$Y$};

\filldraw[fill=mondrianBlue,draw=mondrianBlue,opacity=1] (0,0) -- (135/146,0) -- (405/1412,405/1412*3^.5) -- (0,0);
\filldraw[fill=mondrianCyan,draw=mondrianCyan,opacity=1] (135/146,0) --(1,0) -- (229/248,19/248*3^.5) -- (59/188,59/188*3^.5) -- (405/1412,405/1412*3^.5) -- (135/146,0);
\filldraw[fill=mondrianGreen,draw=mondrianGreen,opacity=1] (59/188,59/188*3^.5) -- (229/248,19/248*3^.5) -- (59/104,45/103*3^.5) -- (9/20,9/20*3^.5) -- (59/188,59/188*3^.5);
\filldraw[fill=mondrianYellow,draw=mondrianYellow,opacity=1] (9/20,9/20*3^.5) -- (59/104,45/103*3^.5) -- (1/2,1/2*3^.5) -- (9/20,9/20*3^.5);

\filldraw [fill=orange,draw=orange,opacity=1] (.5011, .7991) circle (0.09^.5*.075);
\filldraw [fill=red,draw=red,opacity=1] (.5596, .1303) circle (0.16^.5*.075);
\filldraw [fill=cyan,draw=cyan,opacity=1] (.5048, .3446) circle (0.27^.5*.075);
\filldraw [fill=green,draw=green,opacity=1] (.7200, .2198) circle (0.11^.5*.075);
\filldraw [fill=pink,draw=pink,opacity=1] (.5745, .5455) circle (0.15^.5*.075);
\filldraw [fill=gray,draw=gray,opacity=1] (.3009, .1802) circle (0.10^.5*.075);
\filldraw [fill=magenta,draw=magenta,opacity=1] (.3825, .5163) circle (0.12^.5*.075);

\end{tikzpicture}};

\end{tikzpicture}
\vspace*{-1mm}
\caption{The four $M$ equilibria for the DS1 (top) and DS2 (bottom) games assuming the other-regarding preference model proposed by \cite{CharnessRabin2002}. The clusters represent choice and belief data as in Figure \ref{expdataDS}.}\label{CR}
\end{center}
\vspace*{-7mm}
\end{figure}

But even when accounting for behavioral and psychological elements, $M$ equilibrium is unlikely to always be correct. Given the minimal assumptions that $M$ equilibrium imposes, the reason it fails then offers important insights about behavior: is monotonicity violated or are beliefs biased in terms of payoff consequences?  $M$ equilibrium provides an easy-to-compute and parameter-free tool to answer this question.

\newpage
\addtolength{\baselineskip}{-0.9mm}
\vspace*{-15mm}
\bibliography{references}

\providecommand{\noopsort}[1]{}
\begin{thebibliography}{}

\bibitem[\protect\citeauthoryear{Agranov and Ortoleva}{Agranov and
  Ortoleva}{2017}]{agranov2017}
Agranov, M. and P.~Ortoleva (2017).
\newblock Stochastic choice and preferences for randomization.
\newblock {\em Journal of Political Economy\/}~{\em 125\/}(1), 40--68.

\bibitem[\protect\citeauthoryear{Alempaki, Colman, Koelle, Loomes, Pulford,
  et~al.}{Alempaki et~al.}{2019}]{alempaki2019}
Alempaki, D., A.~M. Colman, F.~Koelle, G.~Loomes, B.~D. Pulford, et~al. (2019).
\newblock Investigating the failure to best respond in experimental games.
\newblock Technical report.

\bibitem[\protect\citeauthoryear{Anderson, Goeree, and Holt}{Anderson
  et~al.}{2002}]{AndersonGoereeHolt2002}
Anderson, S.~P., J.~K. Goeree, and C.~A. Holt (2002).
\newblock The logit equilibrium: A perspective on intuitive behavioral
  anomalies.
\newblock {\em Southern Economic Journal\/}~{\em 69\/}(1), 21--47.

\bibitem[\protect\citeauthoryear{Arad and Rubinstein}{Arad and
  Rubinstein}{2012}]{arad2012}
Arad, A. and A.~Rubinstein (2012).
\newblock The 11-20 money request game: A level-k reasoning study.
\newblock {\em American Economic Review\/}~{\em 102\/}(7), 3561--73.

\bibitem[\protect\citeauthoryear{Ballinger and Wilcox}{Ballinger and
  Wilcox}{1997}]{ballinger1997}
Ballinger, T.~P. and N.~T. Wilcox (1997).
\newblock Decisions, error and heterogeneity.
\newblock {\em The Economic Journal\/}~{\em 107\/}(443), 1090--1105.

\bibitem[\protect\citeauthoryear{Bhatt and Camerer}{Bhatt and
  Camerer}{2005}]{bhatt2005}
Bhatt, M. and C.~F. Camerer (2005).
\newblock Self-referential thinking and equilibrium as states of mind in games:
  fmri evidence.
\newblock {\em Games and economic Behavior\/}~{\em 52\/}(2), 424--459.

\bibitem[\protect\citeauthoryear{Blume, Brandenburger, and Dekel}{Blume
  et~al.}{1991}]{BBD1991}
Blume, L., A.~Brandenburger, and E.~Dekel (1991).
\newblock Lexicographic probabilities and equilibrium refinements.
\newblock {\em Econometrica\/}~{\em 59\/}(1), 81--98.

\bibitem[\protect\citeauthoryear{Blume and Zame}{Blume and
  Zame}{1992}]{BlumeZame1992}
Blume, L. and W.~Zame (1992).
\newblock The algebraic geometry of competitive equilibrium.
\newblock In W.~Neuefeind and R.~Reizman (Eds.), {\em Economic Theory and
  International Trade; Essays in Memoriam J. Trout Rader}. Springer, Berlin.

\bibitem[\protect\citeauthoryear{Blume and Zame}{Blume and
  Zame}{1994}]{BlumeZame1994}
Blume, L. and W.~Zame (1994).
\newblock The algebraic geometry of perfect and sequential equilibrium.
\newblock {\em Econometrica\/}~{\em 62\/}(4), 783--794.

\bibitem[\protect\citeauthoryear{Bolton and Ockenfels}{Bolton and
  Ockenfels}{2000}]{bolton2000}
Bolton, G.~E. and A.~Ockenfels (2000).
\newblock Erc: A theory of equity, reciprocity, and competition.
\newblock {\em American economic review\/}~{\em 90\/}(1), 166--193.

\bibitem[\protect\citeauthoryear{Brayer}{Brayer}{1964}]{brayer1964}
Brayer, A.~R. (1964).
\newblock An experimental analysis of some variables of minimax theory.
\newblock {\em Behavioral science\/}~{\em 9\/}(1), 33--44.

\bibitem[\protect\citeauthoryear{Brown and Rosenthal}{Brown and
  Rosenthal}{1990}]{brown1990}
Brown, J.~N. and R.~W. Rosenthal (1990).
\newblock Testing the minimax hypothesis: a re-examination of o'neill's game
  experiment.
\newblock {\em Econometrica (1986-1998)\/}~{\em 58\/}(5), 1065.

\bibitem[\protect\citeauthoryear{Brunner, Camerer, and Goeree}{Brunner
  et~al.}{2011}]{brunner2011}
Brunner, C., C.~F. Camerer, and J.~K. Goeree (2011).
\newblock Stationary concepts for experimental 2 x 2 games: Comment.
\newblock {\em American Economic Review\/}~{\em 101\/}(2), 1029--40.

\bibitem[\protect\citeauthoryear{Camerer, Ho, and Chong}{Camerer
  et~al.}{2004}]{camerer2004}
Camerer, C.~F., T.-H. Ho, and J.-K. Chong (2004).
\newblock A cognitive hierarchy model of games.
\newblock {\em The Quarterly Journal of Economics\/}~{\em 119\/}(3), 861--898.

\bibitem[\protect\citeauthoryear{Capra, Goeree, Gomez, and Holt}{Capra
  et~al.}{1999}]{CapraGoereeGomezHolt1999}
Capra, C.~M., J.~K. Goeree, R.~Gomez, and C.~A. Holt (1999).
\newblock Anomalous behavior in a traveler’s dilemma?
\newblock {\em American Economic Review\/}~{\em 89\/}(3), 678--690.

\bibitem[\protect\citeauthoryear{Charness, Gneezy, and Kuhn}{Charness
  et~al.}{2012}]{charness2012}
Charness, G., U.~Gneezy, and M.~A. Kuhn (2012).
\newblock Experimental methods: Between-subject and within-subject design.
\newblock {\em Journal of Economic Behavior \& Organization\/}~{\em 81\/}(1),
  1--8.

\bibitem[\protect\citeauthoryear{Charness and Rabin}{Charness and
  Rabin}{2002}]{CharnessRabin2002}
Charness, G. and M.~Rabin (2002).
\newblock Understanding social preferences with simple tests.
\newblock {\em Quarterly Journal of Economics\/}~{\em 117\/}(3), 817--869.

\bibitem[\protect\citeauthoryear{Costa-Gomes, Crawford, and
  Broseta}{Costa-Gomes et~al.}{2001}]{costa2001cognition}
Costa-Gomes, M., V.~P. Crawford, and B.~Broseta (2001).
\newblock Cognition and behavior in normal-form games: An experimental study.
\newblock {\em Econometrica\/}~{\em 69\/}(5), 1193--1235.

\bibitem[\protect\citeauthoryear{Costa-Gomes and Weizs{\"a}cker}{Costa-Gomes
  and Weizs{\"a}cker}{2008}]{costa2008}
Costa-Gomes, M.~A. and G.~Weizs{\"a}cker (2008).
\newblock Stated beliefs and play in normal-form games.
\newblock {\em The Review of Economic Studies\/}~{\em 75\/}(3), 729--762.

\bibitem[\protect\citeauthoryear{Crawford}{Crawford}{2018}]{crawford2018}
Crawford, V.~P. (2018).
\newblock Experiments on cognition, communication, coordination, and
  cooperation in relationships.
\newblock {\em Annual Review of Economics\/}.

\bibitem[\protect\citeauthoryear{Crawford, Costa-Gomes, and Iriberri}{Crawford
  et~al.}{2013}]{crawford2013}
Crawford, V.~P., M.~A. Costa-Gomes, and N.~Iriberri (2013).
\newblock Structural models of nonequilibrium strategic thinking: Theory,
  evidence, and applications.
\newblock {\em Journal of Economic Literature\/}~{\em 51\/}(1), 5--62.

\bibitem[\protect\citeauthoryear{Croson}{Croson}{2000}]{croson2000}
Croson, R.~T. (2000).
\newblock Thinking like a game theorist: factors affecting the frequency of
  equilibrium play.
\newblock {\em Journal of economic behavior \& organization\/}~{\em 41\/}(3),
  299--314.

\bibitem[\protect\citeauthoryear{Danz, Fehr, and K{\"u}bler}{Danz
  et~al.}{2012}]{danz2012}
Danz, D.~N., D.~Fehr, and D.~K{\"u}bler (2012).
\newblock Information and beliefs in a repeated normal-form game.
\newblock {\em Experimental Economics\/}~{\em 15\/}(4), 622--640.

\bibitem[\protect\citeauthoryear{Danz, Vesterlund, and Wilson}{Danz
  et~al.}{2020}]{danzVesterlundWilson2020}
Danz, D.~N., L.~Vesterlund, and A.~Wilson (2020).
\newblock Belief elicitation: Limiting truth telling with information on
  incentives.
\newblock {\em NBER working paper\/}.

\bibitem[\protect\citeauthoryear{DellaVigna}{DellaVigna}{2018}]{dellavigna2018}
DellaVigna, S. (2018).
\newblock Structural behavioral economics.
\newblock In {\em Handbook of Behavioral Economics: Applications and
  Foundations 1}, Volume~1, pp.\  613--723. Elsevier.

\bibitem[\protect\citeauthoryear{Demichelis and Ritzberger}{Demichelis and
  Ritzberger}{2003}]{demichelis2003}
Demichelis, S. and K.~Ritzberger (2003).
\newblock From evolutionary to strategic stability.
\newblock {\em Journal of Economic Theory\/}~{\em 113}, 51--75.

\bibitem[\protect\citeauthoryear{Dhami}{Dhami}{2016}]{dhami2016}
Dhami, S. (2016).
\newblock {\em The foundations of behavioral economic analysis}.
\newblock Oxford University Press.

\bibitem[\protect\citeauthoryear{Ehrblatt, Hyndman, Ozbay, and
  Schotter}{Ehrblatt et~al.}{2005}]{ehrblatt2005}
Ehrblatt, W.~Z., K.~Hyndman, E.~Ozbay, and A.~Schotter (2005).
\newblock Convergence: An experimental study.
\newblock Technical report, Mimeo NYU.

\bibitem[\protect\citeauthoryear{Erev and Roth}{Erev and Roth}{1998}]{erev1998}
Erev, I. and A.~E. Roth (1998).
\newblock Predicting how people play games: Reinforcement learning in
  experimental games with unique, mixed strategy equilibria.
\newblock {\em American economic review\/}, 848--881.

\bibitem[\protect\citeauthoryear{Esteban-Casanelles and
  Gon{\c{c}}alves}{Esteban-Casanelles and Gon{\c{c}}alves}{2020}]{esteban2020}
Esteban-Casanelles, T. and D.~Gon{\c{c}}alves (2020).
\newblock The effect of incentives on choices and beliefs in games: An
  experiment.
\newblock Technical report, Mimeo.

\bibitem[\protect\citeauthoryear{Eyster}{Eyster}{2019}]{eyster2019}
Eyster, E. (2019).
\newblock Errors in strategic reasoning.
\newblock {\em Handbook of Behavioral Economics: Applications and Foundations
  1\/}~{\em 2}, 187--259.

\bibitem[\protect\citeauthoryear{Fehr and Schmidt}{Fehr and
  Schmidt}{1999}]{FehrSchmidt1999}
Fehr, E. and K.~Schmidt (1999).
\newblock A theory of fairness, competition, and cooperation.
\newblock {\em Quarterly Journal of Economics\/}~{\em 114\/}(3), 817--868.

\bibitem[\protect\citeauthoryear{Foroughifar}{Foroughifar}{2020}]{foroughifar2020}
Foroughifar, M. (2020).
\newblock Why are rational expectations violated in social interactions?
\newblock {\em Available at SSRN 3733115\/}.

\bibitem[\protect\citeauthoryear{Fr{\'e}chette}{Fr{\'e}chette}{2012}]{frechette2012}
Fr{\'e}chette, G.~R. (2012).
\newblock Session-effects in the laboratory.
\newblock {\em Experimental Economics\/}~{\em 15\/}(3), 485--498.

\bibitem[\protect\citeauthoryear{Friedman and Ward}{Friedman and
  Ward}{2019}]{friedman2019}
Friedman, E. and J.~Ward (2019).
\newblock Stochastic choice and noisy beliefs in games: an experiment.
\newblock Technical report, Mimeo.

\bibitem[\protect\citeauthoryear{Friedman and Mezzetti}{Friedman and
  Mezzetti}{2005}]{friedman2005}
Friedman, J.~W. and C.~Mezzetti (2005).
\newblock Random belief equilibrium in normal form games.
\newblock {\em Games and Economic Behavior\/}~{\em 51\/}(2), 296--323.

\bibitem[\protect\citeauthoryear{Fudenberg and Levine}{Fudenberg and
  Levine}{1993}]{FudenbergLevine1993}
Fudenberg, D. and D.~K. Levine (1993).
\newblock Self-confirming equilibrium.
\newblock {\em Games and Economic Behavior\/}~{\em 61\/}(3), 523--545.

\bibitem[\protect\citeauthoryear{Goeree and Holt}{Goeree and
  Holt}{2001}]{goeree2001}
Goeree, J.~K. and C.~A. Holt (2001).
\newblock Ten little treasures of game theory and ten intuitive contradictions.
\newblock {\em American Economic Review\/}~{\em 91\/}(5), 1402--1422.

\bibitem[\protect\citeauthoryear{Goeree and Holt}{Goeree and
  Holt}{2004}]{goeree2004}
Goeree, J.~K. and C.~A. Holt (2004).
\newblock A model of noisy introspection.
\newblock {\em Games and Economic Behavior\/}~{\em 46\/}(2), 365--382.

\bibitem[\protect\citeauthoryear{Goeree, Holt, Louis, Palfrey, and
  Rogers}{Goeree et~al.}{2019}]{Goeree2019}
Goeree, J.~K., C.~A. Holt, P.~Louis, T.~R. Palfrey, and B.~W. Rogers (2019).
\newblock Rank-dependent choice equilibrium: A generalization of qre.
\newblock In A.~Schram and A.~Ule (Eds.), {\em Handbook of Research Methods and
  Applications in Experimental Economics}. Edwar Elgar Publishing.

\bibitem[\protect\citeauthoryear{Goeree, Holt, and Palfrey}{Goeree
  et~al.}{2003}]{goeree2003}
Goeree, J.~K., C.~A. Holt, and T.~R. Palfrey (2003).
\newblock Risk averse behavior in generalized matching pennies games.
\newblock {\em Games and Economic Behavior\/}~{\em 45\/}(1), 97--113.

\bibitem[\protect\citeauthoryear{Goeree, Holt, and Palfrey}{Goeree
  et~al.}{2005}]{goeree2005}
Goeree, J.~K., C.~A. Holt, and T.~R. Palfrey (2005).
\newblock Regular quantal response equilibrium.
\newblock {\em Experimental Economics\/}~{\em 8}, 347--367.

\bibitem[\protect\citeauthoryear{Goeree, Holt, and Palfrey}{Goeree
  et~al.}{2016}]{GoereeHoltPalfrey2016}
Goeree, J.~K., C.~A. Holt, and T.~R. Palfrey (2016).
\newblock {\em Quantal Response Equilibrium}.
\newblock Princeton, USA: Princeton University Press.

\bibitem[\protect\citeauthoryear{Goeree, Louis, and Zhang}{Goeree
  et~al.}{2017}]{goeree2017}
Goeree, J.~K., P.~Louis, and J.~Zhang (2017).
\newblock Noisy introspection in the 11--20 game.
\newblock {\em The Economic Journal\/}~{\em 128\/}(611), 1509--1530.

\bibitem[\protect\citeauthoryear{Haile, Hortacsu, and Kosenok}{Haile
  et~al.}{2008}]{HaileHortacsuKosenok2008}
Haile, P., A.~Hortacsu, and G.~Kosenok (2008).
\newblock On the empirical content of quantal response equilibrium.
\newblock {\em American Economic Review\/}~{\em 98\/}(1), 180--200.

\bibitem[\protect\citeauthoryear{Halmos}{Halmos}{1998}]{Halmos1998}
Halmos, P.~R. (1998).
\newblock {\em Naive Set Theory}.
\newblock New York, USA: Springer-Verlag New York.

\bibitem[\protect\citeauthoryear{Harsanyi}{Harsanyi}{1973}]{Harsanyi1973}
Harsanyi, J.~C. (1973).
\newblock Games with randomly disturbed payoffs: a new rationale for mixed
  strategy equilibrium points.
\newblock {\em International Journal of Game Theory\/}~{\em 2}, 1--23.

\bibitem[\protect\citeauthoryear{Hey}{Hey}{2001}]{hey2001}
Hey, J.~D. (2001).
\newblock Does repetition improve consistency?
\newblock {\em Experimental economics\/}~{\em 4\/}(1), 5--54.

\bibitem[\protect\citeauthoryear{Hey and Orme}{Hey and Orme}{1994}]{hey1994}
Hey, J.~D. and C.~Orme (1994).
\newblock Investigating generalizations of expected utility theory using
  experimental data.
\newblock {\em Econometrica: Journal of the Econometric Society\/}, 1291--1326.

\bibitem[\protect\citeauthoryear{Hoffmann}{Hoffmann}{2014}]{hoffmann2014}
Hoffmann, T. (2014).
\newblock The effect of belief elicitation game play.

\bibitem[\protect\citeauthoryear{Holt and Smith}{Holt and
  Smith}{2016}]{holt2016}
Holt, C.~A. and A.~M. Smith (2016).
\newblock Belief elicitation with a synchronized lottery choice menu that is
  invariant to risk attitudes.
\newblock {\em American Economic Journal: Microeconomics\/}~{\em 8\/}(1),
  110--39.

\bibitem[\protect\citeauthoryear{Hossain and Okui}{Hossain and
  Okui}{2013}]{hossain2013}
Hossain, T. and R.~Okui (2013).
\newblock The binarized scoring rule.
\newblock {\em Review of Economic Studies\/}~{\em 80\/}(3), 984--1001.

\bibitem[\protect\citeauthoryear{Huck and Weizs{\"a}cker}{Huck and
  Weizs{\"a}cker}{2002}]{huck2002}
Huck, S. and G.~Weizs{\"a}cker (2002).
\newblock Do players correctly estimate what others do?: Evidence of
  conservatism in beliefs.
\newblock {\em Journal of Economic Behavior \& Organization\/}~{\em 47\/}(1),
  71--85.

\bibitem[\protect\citeauthoryear{Hyndman, Ozbay, Schotter, and
  Ehrblatt}{Hyndman et~al.}{2012}]{hyndman2012}
Hyndman, K., E.~Y. Ozbay, A.~Schotter, and W.~Z. Ehrblatt (2012).
\newblock Convergence: an experimental study of teaching and learning in
  repeated games.
\newblock {\em Journal of the European Economic Association\/}~{\em 10\/}(3),
  573--604.

\bibitem[\protect\citeauthoryear{Hyndman, Terracol, and Vaksmann}{Hyndman
  et~al.}{2009}]{hyndman2009}
Hyndman, K., A.~Terracol, and J.~Vaksmann (2009).
\newblock Learning and sophistication in coordination games.
\newblock {\em Experimental Economics\/}~{\em 12\/}(4), 450--472.

\bibitem[\protect\citeauthoryear{Hyndman, Terracol, and Vaksmann}{Hyndman
  et~al.}{2010}]{hyndman2010}
Hyndman, K., A.~Terracol, and J.~Vaksmann (2010).
\newblock Strategic interactions and belief formation: An experiment.
\newblock {\em Applied Economics Letters\/}~{\em 17\/}(17), 1681--1685.

\bibitem[\protect\citeauthoryear{Hyndman, Terracol, and Vaksmann}{Hyndman
  et~al.}{2013}]{hyndman2013}
Hyndman, K.~B., A.~Terracol, and J.~Vaksmann (2013).
\newblock Beliefs and (in) stability in normal-form games.
\newblock {\em Available at SSRN 2270497\/}.

\bibitem[\protect\citeauthoryear{Kakutani}{Kakutani}{1941}]{Kakutani1941}
Kakutani, S. (1941).
\newblock A generalization of brouwer's fixed-point theorem.
\newblock {\em Duke Mathematical Journal\/}~{\em 8\/}(3), 457--459.

\bibitem[\protect\citeauthoryear{Kirman}{Kirman}{2014}]{Kirman2014}
Kirman, A. (2014).
\newblock Is it rational to have rational expectations?
\newblock {\em Mind and Society\/}~{\em 13\/}(1), 29--48.

\bibitem[\protect\citeauthoryear{Kohlberg and Mertens}{Kohlberg and
  Mertens}{1986}]{KohlbergMertens1986}
Kohlberg, E. and J.-F. Mertens (1986).
\newblock On the strategic stability of equilibria.
\newblock {\em Econometrica\/}~{\em 54\/}(5), 1003--1037.

\bibitem[\protect\citeauthoryear{Kreps and Wilson}{Kreps and
  Wilson}{1982}]{KrepsWilson1982}
Kreps, D. and R.~Wilson (1982).
\newblock Sequential equilibria.
\newblock {\em Econometrica\/}~{\em 50}, 863--894.

\bibitem[\protect\citeauthoryear{Kubler and Scheidegger}{Kubler and
  Scheidegger}{2018}]{KublerScheidegger2018}
Kubler, F. and S.~Scheidegger (2018).
\newblock Self-justified equilibria: Existence and computation.
\newblock {\em Working paper, University of Zurich\/}.

\bibitem[\protect\citeauthoryear{Kubler and Schmedders}{Kubler and
  Schmedders}{2010}]{KublerSchmedders2010}
Kubler, F. and K.~Schmedders (2010).
\newblock Competitive equilibria in semi-algebraic economies.
\newblock {\em Journal of Economic Theory\/}~{\em 145\/}(1), 301--330.

\bibitem[\protect\citeauthoryear{Lieberman}{Lieberman}{1960}]{lieberman1960}
Lieberman, B. (1960).
\newblock Human behavior in a strictly determined 3$\times$ 3 matrix game.
\newblock {\em Behavioral Science\/}~{\em 5\/}(4), 317--322.

\bibitem[\protect\citeauthoryear{Lieberman}{Lieberman}{1962}]{lieberman1962}
Lieberman, B. (1962).
\newblock Experimental studies of conflict in some two-person games.
\newblock In J.~Criswell, H.~Solomon, and P.~Suppes (Eds.), {\em Mathematical
  Methods in Small Group Processes}. Stanford, Cal: Stanford University Press.

\bibitem[\protect\citeauthoryear{MacQueen}{MacQueen}{1967}]{macqueen1967}
MacQueen, J. (1967).
\newblock Some methods for classification and analysis of multivariate
  observations.
\newblock In {\em Proceedings of the fifth Berkeley symposium on mathematical
  statistics and probability}, Volume~1, pp.\  281--297. Oakland, CA, USA.

\bibitem[\protect\citeauthoryear{Manski}{Manski}{2004}]{manski2004}
Manski, C.~F. (2004).
\newblock Measuring expectations.
\newblock {\em Econometrica\/}~{\em 72\/}(5), 1329--1376.

\bibitem[\protect\citeauthoryear{McKelvey and Palfrey}{McKelvey and
  Palfrey}{1992}]{mckelvey1992}
McKelvey, R.~D. and T.~R. Palfrey (1992).
\newblock An experimental study of the centipede game.
\newblock {\em Econometrica\/}~{\em 60\/}(4), 803--836.

\bibitem[\protect\citeauthoryear{McKelvey and Palfrey}{McKelvey and
  Palfrey}{1995}]{McKelveyPalfrey1995}
McKelvey, R.~D. and T.~R. Palfrey (1995).
\newblock {Quantal Response Equilibria for Normal Form Games}.
\newblock {\em Games and Economic Behavior\/}~{\em 10\/}(1), 6--38.

\bibitem[\protect\citeauthoryear{McKelvey, Palfrey, and Weber}{McKelvey
  et~al.}{2000}]{mckelvey2000}
McKelvey, R.~D., T.~R. Palfrey, and R.~A. Weber (2000).
\newblock The effects of payoff magnitude and heterogeneity on behavior in
  2$\times$ 2 games with unique mixed strategy equilibria.
\newblock {\em Journal of Economic Behavior \& Organization\/}~{\em 42\/}(4),
  523--548.

\bibitem[\protect\citeauthoryear{McLennan}{McLennan}{2016}]{mclennan2016}
McLennan, A. (2016).
\newblock The index+ 1 principle.
\newblock Technical report, Mimeo.

\bibitem[\protect\citeauthoryear{Messick}{Messick}{1967}]{messick1967}
Messick, D.~M. (1967).
\newblock Interdependent decision strategies in zero-sum games: A
  computer-controlled study.
\newblock {\em Behavioral Science\/}~{\em 12\/}(1), 33--48.

\bibitem[\protect\citeauthoryear{Muth}{Muth}{1961}]{Muth1961}
Muth, J. (1961).
\newblock Rational expectations and the theory of price movements.
\newblock {\em Econometrica\/}~{\em 29\/}(3), 315--335.

\bibitem[\protect\citeauthoryear{Myerson}{Myerson}{1978}]{Myerson1978}
Myerson, R.~B. (1978).
\newblock Refinements of the nash equilibrium concept.
\newblock {\em International Journal of Game Theory\/}~{\em 7}, 73--80.

\bibitem[\protect\citeauthoryear{Nagel}{Nagel}{1995}]{nagel1995}
Nagel, R. (1995).
\newblock Unraveling in guessing games: An experimental study.
\newblock {\em The American Economic Review\/}~{\em 85\/}(5), 1313--1326.

\bibitem[\protect\citeauthoryear{Nash}{Nash}{1950}]{Nash1950}
Nash, J.~F. (1950).
\newblock Equilibrium points in $n$-person games.
\newblock {\em Proceedings of the National Academy of Sciences\/}~{\em
  36\/}(1), 48--49.

\bibitem[\protect\citeauthoryear{Nyarko and Schotter}{Nyarko and
  Schotter}{2002}]{nyarko2002}
Nyarko, Y. and A.~Schotter (2002).
\newblock An experimental study of belief learning using elicited beliefs.
\newblock {\em Econometrica\/}~{\em 70\/}(3), 971--1005.

\bibitem[\protect\citeauthoryear{Ochs}{Ochs}{1995}]{ochs1995}
Ochs, J. (1995).
\newblock Games with unique, mixed strategy equilibria: An experimental study.
\newblock {\em Games and Economic Behavior\/}~{\em 10\/}(1), 202--217.

\bibitem[\protect\citeauthoryear{O'Neill}{O'Neill}{1987}]{oneill1987}
O'Neill, B. (1987).
\newblock Nonmetric test of the minimax theory of two-person zerosum games.
\newblock {\em Proceedings of the national academy of sciences\/}~{\em
  84\/}(7), 2106--2109.

\bibitem[\protect\citeauthoryear{Palfrey and Wang}{Palfrey and
  Wang}{2009}]{palfrey2009}
Palfrey, T.~R. and S.~W. Wang (2009).
\newblock On eliciting beliefs in strategic games.
\newblock {\em Journal of Economic Behavior \& Organization\/}~{\em 71\/}(2),
  98--109.

\bibitem[\protect\citeauthoryear{Polonio and Coricelli}{Polonio and
  Coricelli}{2019}]{polonio2019}
Polonio, L. and G.~Coricelli (2019).
\newblock Testing the level of consistency between choices and beliefs in games
  using eye-tracking.
\newblock {\em Games and Economic Behavior\/}~{\em 113}, 566--586.

\bibitem[\protect\citeauthoryear{Rabin}{Rabin}{2013}]{rabin2013}
Rabin, M. (2013).
\newblock Incorporating limited rationality into economics.
\newblock {\em Journal of Economic Literature\/}~{\em 51\/}(2), 528--43.

\bibitem[\protect\citeauthoryear{Rapoport and Boebel}{Rapoport and
  Boebel}{1992}]{rapoport1992}
Rapoport, A. and R.~B. Boebel (1992).
\newblock Mixed strategies in strictly competitive games: A further test of the
  minimax hypothesis.
\newblock {\em Games and Economic Behavior\/}~{\em 4\/}(2), 261--283.

\bibitem[\protect\citeauthoryear{Rey-Biel}{Rey-Biel}{2009}]{rey2009}
Rey-Biel, P. (2009).
\newblock Equilibrium play and best response to (stated) beliefs in normal form
  games.
\newblock {\em Games and Economic Behavior\/}~{\em 65\/}(2), 572--585.

\bibitem[\protect\citeauthoryear{Rogers, Palfrey, and Camerer}{Rogers
  et~al.}{2009}]{rogers2009}
Rogers, B.~W., T.~R. Palfrey, and C.~F. Camerer (2009).
\newblock Heterogeneous quantal response equilibrium and cognitive hierarchies.
\newblock {\em Journal of Economic Theory\/}~{\em 144\/}(4), 1440--1467.

\bibitem[\protect\citeauthoryear{Sargent}{Sargent}{1991}]{Sargent1991}
Sargent, T.~J. (1991).
\newblock {\em Bounded Rationality in Macroeconomics}.
\newblock Oxford, UK: Oxford University Press.

\bibitem[\protect\citeauthoryear{Sargent}{Sargent}{1999}]{Sargent1999}
Sargent, T.~J. (1999).
\newblock {\em The Conquest of American Inflation}.
\newblock Princeton, USA: Princeton University Press.

\bibitem[\protect\citeauthoryear{Schanuel, Simon, and Zame}{Schanuel
  et~al.}{1991}]{SchanuelSimonZame1991}
Schanuel, S., L.~Simon, and W.~Zame (1991).
\newblock The algebraic geometry of games and the tracing procedure.
\newblock In R.~Selten (Ed.), {\em Game Equilibrium Models II}. Springer,
  Berlin and Heidelberg.

\bibitem[\protect\citeauthoryear{Schlag, Tremewan, and Van~der Weele}{Schlag
  et~al.}{2015}]{schlag2015}
Schlag, K.~H., J.~Tremewan, and J.~J. Van~der Weele (2015).
\newblock A penny for your thoughts: a survey of methods for eliciting beliefs.
\newblock {\em Experimental Economics\/}~{\em 18\/}(3), 457--490.

\bibitem[\protect\citeauthoryear{Schotter and Trevino}{Schotter and
  Trevino}{2014}]{schotter2014}
Schotter, A. and I.~Trevino (2014).
\newblock Belief elicitation in the laboratory.
\newblock {\em Annu. Rev. Econ.\/}~{\em 6\/}(1), 103--128.

\bibitem[\protect\citeauthoryear{Selten}{Selten}{1975}]{Selten1975}
Selten, R. (1975).
\newblock Reexamination of the perfectness concept for equilibrium points in
  extensive games.
\newblock {\em International Journal of Game Theory\/}~{\em 4}, 25--55.

\bibitem[\protect\citeauthoryear{Selten and Chmura}{Selten and
  Chmura}{2008}]{selten2008}
Selten, R. and T.~Chmura (2008).
\newblock Stationary concepts for experimental 2x2-games.
\newblock {\em American Economic Review\/}~{\em 98\/}(3), 938--66.

\bibitem[\protect\citeauthoryear{Selten, Chmura, and Goerg}{Selten
  et~al.}{2011}]{selten2011}
Selten, R., T.~Chmura, and S.~J. Goerg (2011).
\newblock Stationary concepts for experimental 2 x 2 games: Reply.
\newblock {\em American Economic Review\/}~{\em 101\/}(2), 1041--44.

\bibitem[\protect\citeauthoryear{Simon}{Simon}{1984}]{Simon1984}
Simon, H.~A. (1984).
\newblock On the behavioral and rational foundations of economic dynamics.
\newblock {\em Journal of Economic Behavior and Organization\/}~{\em 5\/}(1),
  35--55.

\bibitem[\protect\citeauthoryear{Stahl and Wilson}{Stahl and
  Wilson}{1994}]{stahl1994}
Stahl, D.~O. and P.~W. Wilson (1994).
\newblock Experimental evidence on players' models of other players.
\newblock {\em Journal of economic behavior \& organization\/}~{\em 25\/}(3),
  309--327.

\bibitem[\protect\citeauthoryear{Stahl and Wilson}{Stahl and
  Wilson}{1995}]{stahl1995}
Stahl, D.~O. and P.~W. Wilson (1995).
\newblock On players' models of other players: Theory and experimental
  evidence.
\newblock {\em Games and Economic Behavior\/}~{\em 10\/}(1), 218--254.

\bibitem[\protect\citeauthoryear{Sutter, Czermak, and Feri}{Sutter
  et~al.}{2013}]{sutter2013}
Sutter, M., S.~Czermak, and F.~Feri (2013).
\newblock Strategic sophistication of individuals and teams. experimental
  evidence.
\newblock {\em European economic review\/}~{\em 64}, 395--410.

\bibitem[\protect\citeauthoryear{Terracol and Vaksmann}{Terracol and
  Vaksmann}{2009}]{terracol2009}
Terracol, A. and J.~Vaksmann (2009).
\newblock Dumbing down rational players: Learning and teaching in an
  experimental game.
\newblock {\em Journal of Economic Behavior \& Organization\/}~{\em 70\/}(1-2),
  54--71.

\bibitem[\protect\citeauthoryear{Trautmann and van~de Kuilen}{Trautmann and
  van~de Kuilen}{2014}]{trautmann2015}
Trautmann, S.~T. and G.~van~de Kuilen (2014, 09).
\newblock {Belief Elicitation: A Horse Race among Truth Serums}.
\newblock {\em The Economic Journal\/}~{\em 125\/}(589), 2116--2135.

\bibitem[\protect\citeauthoryear{Tversky}{Tversky}{1969}]{tversky1969}
Tversky, A. (1969).
\newblock Intransitivity of preferences.
\newblock {\em Psychological review\/}~{\em 76\/}(1), 31.

\bibitem[\protect\citeauthoryear{Tversky and Kahneman}{Tversky and
  Kahneman}{1981}]{tversky1981}
Tversky, A. and D.~Kahneman (1981).
\newblock The framing of decisions and the psychology of choice.
\newblock {\em science\/}~{\em 211\/}(4481), 453--458.

\bibitem[\protect\citeauthoryear{van Damme}{van Damme}{1996}]{vanDamme1996}
van Damme, E. (1996).
\newblock {\em Stability and perfection of Nash equilibria}.
\newblock Heidelberg, Germany: Springer-Verlag.

\bibitem[\protect\citeauthoryear{von Neumann and Morgenstern}{von Neumann and
  Morgenstern}{1944}]{vonNeumannMoregenstern1944}
von Neumann, J. and O.~Morgenstern (1944).
\newblock {\em Theory of Games and Economic Behavior}.
\newblock Princeton, USA: Princeton University Press.

\bibitem[\protect\citeauthoryear{Wang}{Wang}{2011}]{wang2011}
Wang, S.~W. (2011).
\newblock Incentive effects: The case of belief elicitation from individuals in
  groups.
\newblock {\em Economics Letters\/}~{\em 111\/}(1), 30--33.

\bibitem[\protect\citeauthoryear{Weizs{\"a}cker}{Weizs{\"a}cker}{2003}]{weizsacker2003}
Weizs{\"a}cker, G. (2003).
\newblock Ignoring the rationality of others: evidence from experimental
  normal-form games.
\newblock {\em Games and Economic Behavior\/}~{\em 44\/}(1), 145--171.

\bibitem[\protect\citeauthoryear{Wilson and Vespa}{Wilson and
  Vespa}{2018}]{wilson2018}
Wilson, A. and E.~Vespa (2018).
\newblock Paired-uniform scoring: Implementing a binarized scoring rule with
  non-mathematical language.
\newblock Technical report, Working paper.

\bibitem[\protect\citeauthoryear{Wolff and Bauer}{Wolff and
  Bauer}{2021}]{wolff2021}
Wolff, I. and D.~Bauer (2021).
\newblock Why is belief-action consistency so low? e role of belief
  uncertainty.

\bibitem[\protect\citeauthoryear{Woodford}{Woodford}{1990}]{Woodford1990}
Woodford, M. (1990).
\newblock Learning to believe in sunspots.
\newblock {\em Econometrica\/}~{\em 58\/}(2), 277--307.

\bibitem[\protect\citeauthoryear{Wright and Leyton-Brown}{Wright and
  Leyton-Brown}{2017}]{wright2017}
Wright, J.~R. and K.~Leyton-Brown (2017).
\newblock Predicting human behavior in unrepeated, simultaneous-move games.
\newblock {\em Games and Economic Behavior\/}~{\em 106}, 16--37.

\bibitem[\protect\citeauthoryear{Wu and Jiang}{Wu and
  Jiang}{1962}]{WuJiang1962}
Wu, W.~T. and J.-H. Jiang (1962).
\newblock Essential equilibrium points of n-person noncooperative games.
\newblock {\em Scientia Sinica\/}~{\em 11}, 1307--1322.

\end{thebibliography}
\bibliographystyle{chicago}

\newpage

\startappendix
\addtolength{\baselineskip}{-0.55mm}

\newappendix{Variational Inequality Characterization}
\label{app:compAlt}

Here we present a characterization based on variational inequalities akin to those characterizing Nash equilibrium that is well-suited for computation. Recall that $\sigma\in\Sigma$ is a Nash equilibrium if $\pi_i(\sigma_i,\sigma_{-i})\geq\pi_i(s_i,\sigma_{-i})$ for all $s_i\in S_i$, $i\in N$. We will rewrite these inequalities as variational inequalities. For $i\in N$, let $\rho_i$ denote a permutation of the set $\{1,2,\ldots,K_i\}$ and for $v\in\field{R}^{K_i}$ let $\rho_i(v)=(v_{\rho_i(1)},\ldots,v_{\rho_i(K_i)})$. For $v,w\in\field{R}^{K_i}$ let $\langle v|w\rangle=\sum_{k=1}^{K_i}v_{k}w_{k}$ denote the usual inner product. The inequalities defining Nash equilibrium can be written as
\begin{displaymath}
  \langle\pi_i(\sigma_{-i})\,|\,\sigma_i-\rho_i(\hat{\mu}_i)\rangle\,\geq\,0
\end{displaymath}
for all permutations $\rho_i$ and $i\in N$. (Recall that $\hat{\mu}_i=s_{i1}=(1,0,\ldots,0)$ so the set of all permutations of $\hat{\mu}_i$ is $S_i$.) These variational inequalities have an intuitive
geometric interpretation. For instance, if $\sigma_i$ is a pure strategy corresponding to a vertex of $\Sigma_i$ then the differences $\sigma_i-\rho_i(\hat{\mu}_i)$ define the normal cone at $\sigma_i$, see the left panel of Figure \ref{normalfan}. The above inequality requires $\pi_i(\sigma)$ to lie in this normal cone. And if $\sigma_i$ is totally mixed, the normal cone is one-dimensional and generated by $e=(1,1,\ldots,1)$. The requirement that $\pi_i(\sigma)$ lies in this normal cone means that all expected payoffs are equal.

\begin{figure}[h]
\begin{center}
\begin{tikzpicture}[scale=0.72]
\draw[line width=1pt] (0,0) -- (4,0) -- (2,2*3^.5) -- (0,0);
\filldraw[fill=gray!50,opacity=0.1] (0,0) -- (4,0) -- (2,2*3^.5) -- (0,0);

\draw[line width=0.5pt,dashed,->] (0,0) -- (-1/2*3^.5,1/2);
\draw[line width=0.5pt,dashed,->] (0,0) -- (0,-1);
\draw[line width=0.5pt,dashed,->] (0,0) -- (-2/3,-2/9*3^.5);
\filldraw[fill=red!50,opacity=.3] (-3^.5,1) -- (0,0) -- (0,-2) decorate [decoration=snake] { -- (-3^.5,1)};

\draw[line width=0.5pt,dashed,->] (2,2*3^.5) -- (2+1/2*3^.5,2*3^.5+1/2);
\draw[line width=0.5pt,dashed,->] (2,2*3^.5) -- (2-1/2*3^.5,2*3^.5+1/2);
\draw[line width=0.5pt,dashed,->] (2,2*3^.5) -- (2,2*3^.5+4/9*3^.5);
\filldraw[fill=yellow!50,opacity=.3] (2+3^.5,2*3^.5+1) -- (2,2*3^.5) -- (2-3^.5,2*3^.5+1) decorate [decoration=snake] { -- (2+3^.5,2*3^.5+1)};

\draw[line width=0.5pt,dashed,->] (4,0) -- (4+1/2*3^.5,1/2);
\draw[line width=0.5pt,dashed,->] (4,0) -- (4,-1);
\draw[line width=0.5pt,dashed,->] (4,0) -- (14/3,-2/9*3^.5);
\filldraw[fill=blue!50,opacity=.3] (4+3^.5,1) -- (4,0) -- (4,-2) decorate [decoration=snake] { -- (4+3^.5,1)};

\draw[line width=0.5pt,dashed,->] (2,0) -- (2,-1);
\draw[line width=0.5pt,dashed,->] (3,3^.5) -- (3+1/2*3^.5,3^.5+1/2);
\draw[line width=0.5pt,dashed,->] (1,3^.5) -- (1-1/2*3^.5,3^.5+1/2);

\fill[black] (0,0) circle (3pt); \node at (.7,.4) {$\sigma_i$};

\node at (11,1.3) {
\begin{tikzpicture}[scale=0.72]
\draw[line width=1pt,gray] (0,0) -- (8,0) -- (4,4*3^.5) -- (0,0);
\draw[line width=1pt] (10/3,2/3*3^.5) -- (8/3,4/3*3^.5) -- (10/3,2*3^.5) -- (14/3,2*3^.5) -- (16/3,4/3*3^.5) -- (14/3,2/3*3^.5) -- (10/3,2/3*3^.5);
\filldraw[fill=gray!50,opacity=0.1] (10/3,2/3*3^.5) -- (8/3,4/3*3^.5) -- (10/3,2*3^.5) -- (14/3,2*3^.5) -- (16/3,4/3*3^.5) -- (14/3,2/3*3^.5) -- (10/3,2/3*3^.5);

\draw[line width=0.5pt,dashed,->] (10/3,2/3*3^.5) -- (10/3,2/3*3^.5-1);
\draw[line width=0.5pt,dashed,->] (10/3,2/3*3^.5) -- (10/3-1/2*3^.5,2/3*3^.5-1/2);
\filldraw[fill=red!50,opacity=.2] (10/3,2/3*3^.5-1) -- (10/3,2/3*3^.5) -- (10/3-1/2*3^.5,2/3*3^.5-1/2) decorate [decoration=snake] { -- (10/3,2/3*3^.5-1)};

\draw[line width=0.5pt,dashed,->] (8/3,4/3*3^.5) -- (8/3-1/2*3^.5,4/3*3^.5-1/2);
\draw[line width=0.5pt,dashed,->] (8/3,4/3*3^.5) -- (8/3-1/2*3^.5,4/3*3^.5+1/2);
\filldraw[fill=orange!50,opacity=.2] (8/3-1/2*3^.5,4/3*3^.5-1/2) --(8/3,4/3*3^.5) -- (8/3-1/2*3^.5,4/3*3^.5+1/2) decorate [decoration=snake] { -- (8/3-1/2*3^.5,4/3*3^.5-1/2)};

\draw[line width=0.5pt,dashed,->] (10/3,2*3^.5) -- (10/3-1/2*3^.5,2*3^.5+1/2);
\draw[line width=0.5pt,dashed,->] (10/3,2*3^.5) -- (10/3,2*3^.5+1);
\filldraw[fill=yellow!50,opacity=.2] (10/3-1/2*3^.5,2*3^.5+1/2) -- (10/3,2*3^.5) -- (10/3,2*3^.5+1) decorate [decoration=snake] { -- (10/3-1/2*3^.5,2*3^.5+1/2)};

\draw[line width=0.5pt,dashed,->] (14/3,2*3^.5) -- (14/3,2*3^.5+1);
\draw[line width=0.5pt,dashed,->] (14/3,2*3^.5) -- (14/3+1/2*3^.5,2*3^.5+1/2);
\filldraw[fill=green!50,opacity=.2] (14/3,2*3^.5+1) -- (14/3,2*3^.5) -- (14/3+1/2*3^.5,2*3^.5+1/2) decorate [decoration=snake] { -- (14/3,2*3^.5+1)};

\draw[line width=0.5pt,dashed,->] (16/3,4/3*3^.5) -- (16/3+1/2*3^.5,4/3*3^.5+1/2);
\draw[line width=0.5pt,dashed,->] (16/3,4/3*3^.5) -- (16/3+1/2*3^.5,4/3*3^.5-1/2);
\filldraw[fill=cyan!50,opacity=.2] (16/3+1/2*3^.5,4/3*3^.5+1/2) -- (16/3,4/3*3^.5) -- (16/3+1/2*3^.5,4/3*3^.5-1/2) decorate [decoration=snake] { -- (16/3+1/2*3^.5,4/3*3^.5+1/2)};

\draw[line width=0.5pt,dashed,->] (14/3,2/3*3^.5) -- (14/3+1/2*3^.5,2/3*3^.5-1/2);
\draw[line width=0.5pt,dashed,->] (14/3,2/3*3^.5) -- (14/3,2/3*3^.5-1);
\filldraw[fill=blue!50,opacity=.2] (14/3+1/2*3^.5,2/3*3^.5-1/2) -- (14/3,2/3*3^.5) -- (14/3,2/3*3^.5-1) decorate [decoration=snake] { -- (14/3+1/2*3^.5,2/3*3^.5-1/2)};

\fill[black] (10/3,2/3*3^.5) circle (3pt); \node at (10/3+0.3,2/3*3^.5+0.3) {$\sigma_i$};

\end{tikzpicture}};
\end{tikzpicture}
\caption{The left panel shows a probability simplex and its normal fan, i.e. the collection of normal cones. The right panel shows a permutahedron inside the simplex and its normal fan.}\label{normalfan}
\end{center}
\vspace*{-6mm}
\end{figure}

\begin{absolutelynopagebreak}
Now consider the variational inequalities
\begin{equation}\label{Mvars}
\begin{aligned}
  \langle\pi_i(\sigma_{-i})\,|\,\sigma_i-\rho_i(\sigma_i)\rangle &\geq\,\,0\\
  \langle\pi_i(\omega_i)\,|\,\sigma_i-\rho_i(\sigma_i)\rangle &\geq\,\, 0
\end{aligned}
\end{equation}
The interpretation of the top inequality is that $\sigma$ is a Nash equilibrium when, for $i\in N$, strategies are restricted to the permutahedron generated by $\sigma_i$, see the right panel of Figure \ref{normalfan}. This permutahedron has fewer than $K_i!$ vertices when $\rho_i(\sigma_i)=\sigma_i$, i.e. when $\sigma_i$ belongs to one or more simplex diagonals. In this case, the contrapositive of \eqref{M} implies
\begin{equation}\label{diagonals}
  \rho_i(\sigma_{-i})=\sigma_{-i}\,\Rightarrow\,
  \rho_i(\pi_i(\sigma_{-i}))=\pi_i(\sigma_{-i})\,\wedge\,\rho_i(\pi_i(\omega_{i}))=\pi_i(\omega_{i})
\end{equation}
\vspace*{-9mm}
\begin{proposition}\label{prop:finite}
$\overline{\mathcal{M}}(G)$ is the closure of $\{(\sigma,\omega)\in\Sigma_{int}\times\,\Omega\,|\,\eqref{Mvars}\,\,\text{\em and}\,\,\eqref{diagonals}\,\,\text{\em hold}\,\,\forall\rho_i,\,i\in N\}$.
\end{proposition}
\end{absolutelynopagebreak}

\newappendix{Proofs}
\label{app:proofs}

\noindent\textbf{Proof of Prop 1.}
The rank conditions are equivalent to stating that strict orderings of $\pi(\sigma)$ imply strict orderings of $\sigma$ and equalities in $\sigma$ imply equalities in $\pi(\sigma)$, which is equivalent to $\pi_{ij}(\sigma)<\pi_{ik}(\sigma)\Longrightarrow\sigma_{ij}<\sigma_{ik}$. The proof for the beliefs, $\omega$, is similar. $\hspace*{5mm}\blacksquare$

\smallskip

\noindent\textbf{Proof of Prop 2.}
Existence follows from \cite{Kakutani1941} as the image of the $rank^{\mu}$ correspondence is a closed and convex set (see e.g. Figure~\ref{rank3}). $\hspace*{5mm}\blacksquare$

\vspace*{-3mm}
\begin{lemma}\label{lem:idempotent}
$rank^{\mu}\circ rank^{\mu}=rank^{\mu}$ and $rank^{\mu}\circ rank^{\mu'}=rank^{\mu}$ for $\mu\in\Sigma$, regular $\mu'\in\Sigma$.
\end{lemma}
\vspace*{-3mm}

\noindent\textbf{Proof.} Ranking twice is the same as ranking once, which implies that $rank^{\mu}\circ rank^{\mu}=rank^{\mu}$. Moreover, ranking using one set of labels (say $A$, $B$, $C$, $\ldots$) and then ranking again using other labels (say $1$, $2$, $3$, $\ldots$) is no different from ranking once using the latter labels (if the first ranking lost no information). Hence, $rank^{\mu}\circ rank^{\mu'}=rank^{\mu}$ for $\mu\in\Sigma$, regular $\mu'\in\Sigma$.$\hspace*{5mm}\blacksquare$

\smallskip

\noindent\textbf{Proof of Cor. 1.} For regular $\mu$, $rank$ applied to (\ref{OREdef}) yields $rank(\sigma)\subseteq rank(\pi(\omega))=rank(\pi(\sigma))$, i.e. any regular $\mu$-equilibrium belongs to an $M$-equilibrium. Existence of a $\mu$-equilibrium for regular $\mu$ thus implies there exists at least one $M$ equilibrium for any $G$. $\hspace*{5mm}\blacksquare$

\smallskip

\noindent\textbf{Proof of Prop 3.} Let $\sigma\in\Sigma_{int}$.
\begin{itemize}\addtolength{\itemsep}{-2mm}
\vspace*{-2mm}
\item[1$\Leftrightarrow$2] Suppose $\sigma\in E_{\mu}(G)$ for some regular $\mu$ then $\sigma\in rank^{\mu}(\pi(\sigma))$. Applying $rank$ and using Lemma~\ref{lem:idempotent} yields $rank(\sigma)\subseteq rank(\pi(\sigma))$, i.e. $\sigma$ belongs to an $M$-choice set. Conversely, if $\sigma$ is an interior element of an $M$-choice set then $rank(\sigma)\subseteq rank(\pi(\sigma))$. Applying $rank^{\sigma}$ and using Lemma~\ref{lem:idempotent} and the fact that $rank^{\sigma}(\sigma)=\sigma$ yields $\sigma\in rank^{\sigma}(\pi(\sigma))$, i.e. $\sigma$ is a $\mu$-equilibrium profile for $\mu=\sigma$.
\item[1$\Leftrightarrow$3] If $\sigma(\varepsilon)$ is a $\varepsilon$-proper equilibrium then $\pi_{ij}(\sigma(\varepsilon))<\pi_{ik}(\sigma(\varepsilon))\Longrightarrow\sigma_{ij}(\varepsilon)\,\leq\,\varepsilon\sigma_{ik}(\varepsilon)$. Since $\varepsilon\in(0,1)$ this implies $\pi_{ij}(\sigma(\varepsilon))<\pi_{ik}(\sigma(\varepsilon))\Longrightarrow\sigma_{ij}(\varepsilon)\,<\,\sigma_{ik}(\varepsilon)$, i.e. $\sigma(\varepsilon)$ is an element of an $M$-choice set. Conversely, if $\sigma$ belongs to an $M$-choice set then $\pi_{ij}(\sigma)<\pi_{ik}(\sigma)\Longrightarrow\sigma_{ij}\,<\,\sigma_{ik}$. There exists $\varepsilon<1$ such that this can be written as $\pi_{ij}(\sigma)<\pi_{ik}(\sigma)\Longrightarrow\sigma_{ij}\,\leq\,\varepsilon\,\sigma_{ik}$, i.e. $\sigma$ is an $\varepsilon$-proper equilibrium.
\item[1$\Leftrightarrow$4] If $R\in\mathscr{R}$ denotes the quantal response function then $\pi_{ij}<\pi_{ik}\Longrightarrow R_{ij}(\pi_i)<R_{ik}(\pi_i)$, see Definition 2.1 in \cite{GoereeHoltPalfrey2016}, i.e. any QRE profile belongs to an $M$-choice set. To show the converse, we restrict attention to games with the property that for each player $i$, the set of strategy profiles $\sigma_{-i}$ that make $i$ indifferent between two strategies has measure zero in $\Sigma_{-i}$. This is a generic class of games (open and dense in the space of all finite normal-form games). The proof follows by constructing for each element, $\sigma$, of an $M$-choice set a quantal response function $R\in\mathscr{R}$ such that $\sigma$ is a QRE for $R$, see \cite{Goeree2019}. We include a shortened version here for completeness. Without loss of generality, relabel $i$'s actions so that $j<k$ implies $\pi_{ij}(\sigma)\leq\pi _{ik}(\sigma)$ and $\sigma _{ij}\leq\sigma _{ik}$, where payoff equality holds iff choice probabilities are equal (except possibly for a measure zero set). Construct a strictly positive, continuous, strictly increasing function, $h_{i}:\field{R}\rightarrow\field{R}_+$ as follows. For $1\leq j\leq K_{i}$ let $h_{i}(\pi_{ij}(\sigma))=\sigma _{ij}$. This pins down the value of $h_i$ at $K_{i}$ ordered points (or fewer if there are payoff indifferences at $\sigma$) and these values are increasing. Next let
    \begin{displaymath}
    h_{i}(x)\,=\,\left\{
    \begin{array}{ccc}
    \frac{\sigma _{i1}}{1+\pi _{i1}(\sigma)-x} & \mbox{for} & x\,\leq\,\pi _{i1}(\sigma)\\[2mm]
    \frac{\sigma _{ij+1}(x-\pi _{ij}(\sigma))+\sigma_{ij}(\pi _{ij+1}(\sigma)-x)}{\pi _{ij+1}(\sigma)-\pi _{ij}(\sigma)} & \mbox{for} & \pi _{ij}(\sigma)\,\leq\,x\,\leq\,\pi _{ij+1}(\sigma),\,\,\,\,j\,=\,1,\ldots,K_i-1\\[1mm]
    \sigma _{iK_{i}}+x-\pi _{iK_{i}}(\sigma) & \mbox{for} & x\geq\pi _{iK_{i}}(\sigma)
    \end{array}\right.
    \end{displaymath}
    which extends $h_i$ to the real line such that it is strictly increasing and strictly positive everywhere. Now, for each $i\in N$ and $1\leq j\leq K_i$, define  $R_i:\field{R}^{K_i}\rightarrow\Sigma_i$ as follows
    \begin{displaymath}
    R_{ij}(\pi _{i})\,=\,\frac{h_{i}(\pi _{ij})}{\sum_{k\,=\,1}^{K_{i}}h_{i}(\pi _{ik})}
    \end{displaymath}
    $R_{i}$ satisfies Definition 2.1 in \cite{GoereeHoltPalfrey2016} and, hence, is a quantal response function. By construction we have $\sigma_{ij}=R_{ij}(\pi _{i}(\sigma))$, i.e. $\sigma$ is a QRE for $R=(R_{1},\ldots,R_{n})$. $\hspace*{5mm}\blacksquare$
\end{itemize}

\smallskip

\noindent\textbf{Proof of Prop 4.}
This follows since, for all $i\in N$ and $1\leq k\leq|S_i|$, $\varepsilon\sigma_{ik}\leq\min(\varepsilon,\sigma_{ik})\leq\varepsilon$ for $\varepsilon\in(0,1)$ and $\sigma\in\Sigma_{int}$. The examples in the main text show the inclusions can be strict. $\hspace*{5mm}\blacksquare$

\smallskip

\noindent\textbf{Proof of Prop \ref{gen1}.}
Let $\sigma^u$ denote the profile where all players randomize uniformly over their available strategies and let $\mathcal{O}\subset\Gamma$ denote the set of normal-form games with non-thick indifference curves\footnote{I.e. the set of profiles $\sigma_{-i}$ that make player $i$ indifferent between two choices has measure zero in $\Sigma_{-i}$.} such that $\pi_{ij}(\sigma^u)\neq\pi_{ik}(\sigma^u)$ for all $i\in N$ and $1\leq j<k\leq K_i$. By continuity of expected payoffs there exists an open ball around $\sigma^u$ such that expected payoffs are strictly ranked. Since all possible rankings of choice profiles occur arbitrarily close to $\sigma^u$, there exists a full-dimensional $M$-equilibrium for one of the possible rankings. Profiles in the interior of this full-dimensional set are strictly ranked as are the associated expected payoffs, i.e. they are colorable. Since expected payoffs are continuous in the payoff numbers, they will be ranked the same for games that are sufficiently close. Hence, the profiles are robust. This establishes (i). Property (ii) follows since each $\Sigma_i$ can be partitioned into $K_i!$ equally-sized subsets indexed by the ranks of the entries of the choice profiles it contains. Since $rank_i(\sigma)$ must be constant on an $M$-equilibrium choice set for all $i\in N$, the $M$-equilibrium choice set must be contained in the Cartesian product of a single such subset for each player. Hence, its size cannot be larger than $\prod_{i=1}^n 1/K_i!$. To show Property (iii), pick any player $i\in N$. For almost all $\sigma_{-i}\in\Sigma_{-i}$, player $i$ has a strict ranking of expected payoffs associated with the $K_i$ possible choices. The measure of the set of $\sigma_i\in\Sigma_i$ consistent with this strict ranking is bounded by $1/K_i!$. Therefore, also the sum of the products of all such sets of $\sigma_i\in\Sigma_i$ and the corresponding set of all $\sigma_{-i}\in\Sigma_{-i}$ is bounded by $1/K_i!$. Property (iv) holds, for instance, for any game in which all players have a dominant strategy. $\hspace*{5mm}\blacksquare$

\smallskip

\noindent\textbf{Proof of Prop \ref{prop:finite}.}
For choices the conditions are equivalent to $(\sigma_{ik}-\sigma_{ij})(\pi_{ik}(\sigma)-\pi_{ij}(\sigma))\geq 0$
and $\sigma_{ik}=\sigma_{ij}\Rightarrow\pi_{ik}(\sigma)=\pi_{ij}(\sigma)$ for $i\in N$ and $1\leq j,k\leq K_i$. These imply $\pi_{ik}(\sigma)>\pi_{ij}(\sigma)\Rightarrow\sigma_{ik}>\sigma_{ij}$ and $\sigma_{ik}=\sigma_{ij}\Rightarrow\pi_{ik}(\sigma)=\pi_{ij}(\sigma)$. This is Definition \ref{Mdef} restricted to choices. Conversely, Definition \ref{Mdef} implies $(\sigma_{ik}-\sigma_{ij})(\pi_{ik}(\sigma)-\pi_{ij}(\sigma))\geq 0$ and $\pi_{ij}(\sigma)=\pi_{ik}(\sigma)$ if $\sigma_{ij}=\sigma_{ik}$. This is Proposition \ref{prop:finite} restricted to choices. The proof for the beliefs, $\omega$, is similar. $\hspace*{5mm}\blacksquare$

\newappendix{Belief Elicitation}

In the method we use to elicit beliefs, subjects submit the ``percentage chance'' with which they believe their opponent chooses each action by moving a single slider between endpoints labeled $A$ and $B$. Any point on the slider corresponds to a unique chance of $A$ being played (and $B$ with complementary chance) with the $A$ endpoint ($B$ endpoint) corresponding to the belief that the opponent chooses $A$ ($B$) for sure. The point chosen by the subject is then compared with two computer-generated random points on the slider. If the chosen point is closer to the opponent's actual choice (one of the endpoints) than at least one of the two randomly-drawn points then the subject receives a fixed prize. The next proposition generalizes this method so that it can be used for general normal-form games.

\setlength{\abovedisplayskip}{3pt}
\setlength{\belowdisplayskip}{3pt}
\begin{proposition} \label{incentivecompatible}
Consider the elicitation mechanism where player $i\in N$ uses $S_i=\sum_{j\neq i}K_j$ sliders, labeled $S_{jk}$ for $j\neq i$ and $k=1,\ldots,K_j$, to report her beliefs, $q_{jk}$, about player $j$'s choice. For each slider, two (uniform) random numbers are drawn and player $i$'s belief for that slider is ``correct'' if the reported belief is closer to the actual outcome (0 or 1) than at least one of the random draws. If players are risk-neutral then the elicitation mechanism is incentive compatible when a prize, $P\geq 0$, is paid for all correct beliefs for any randomly selected subset of sliders. If players are not risk neutral then the elicitation mechanism is incentive compatible if a prize is paid when the stated belief is correct for a single randomly selected slider.
\end{proposition}
\addtolength{\baselineskip}{-0.0mm}
\noindent\textbf{Proof of Prop \ref{incentivecompatible}.}
Let $u_i(x)$ denote player $i$'s utility of being paid a prize amount $x$ (with $u_i(0)=0$) and let $p$ and $q$ denote the concatenations of player $i$'s true and reported beliefs respectively. Player $i$ wins a prize $P$ for slider $S_{jk}$ with chance
\begin{displaymath}
  P_{jk}\,=\,p_{jk}(1-(1-q_{jk})^2)+(1-p_{jk})(1-q_{jk}^2)
\end{displaymath}
and gets 0 with complementary probability. If all correct beliefs for a random subset $S\subseteq S_i$ of sliders pay a prize $P$ then player $i$'s expected utility of reporting $q$ when her true beliefs are $p$ is
\begin{equation}\label{expUtil}
U_i(p,q)\,=\,\sum_{W\,\subseteq\,S}u_i(P|W|)\Bigl(\!\!\begin{array}{c}|S|\\|W|\end{array}\!\!\Bigr)\,\prod_{S_{jk}\,\in\,W}\!P_{jk}
\,\prod_{S_{jk}\,\not\in\,W}\!1-P_{jk}
\end{equation}
where $W$ is the subset of selected sliders for which player $i$ wins the prize $P$. If player $i$ is risk neutral, i.e. $u_i(x)=x$, then this reduces to the expected number of wins
\begin{displaymath}
U_i(p,q)\,=\,\frac{P\,|S|}{|S_i|}\,\sum_{S_{jk}\,\in\,S_i}P_{jk}\nonumber
\end{displaymath}
and optimizing with respect to $q$ yields
\begin{displaymath}
\nabla_qU_i(p,q)\,=\,2\,\frac{P\,|S|}{|S_i|}\,(p-q)
\end{displaymath}
so truthful reporting is optimal. If player $i$ is not risk neutral then there can only be two possible payoff outcomes, 0 and $P$, for the elicitation mechanism to be incentive compatible. In other words, $|S|=1$ and (\ref{expUtil}) reduces to
\begin{displaymath}
U_i(p,q)\,=\,\frac{u_i(P)}{|S_i|}\,\sum_{S_{jk}\,\in\,S_i}P_{jk}
\end{displaymath}
and truthful reporting is again optimal. $\hspace*{5mm}\blacksquare$

\newappendix{Empirical Support}

\noindent\textbf{Support for Result 1}
\begin{itemize}\addtolength{\itemsep}{-2mm}
\item[(i)] As can be seen in the top left graph of figure \ref{expdataAMP}, there is very little overlap between the confidence intervals for the average choice frequncies in the five games. Formal statistical testing confirms this. The left panel of table \ref{pvals_only_red} reports the p-values from a Hotteling's t-test in pairwise comparisons of the average choice frequency in different games. In most cases, pairwise comparisons show differences that are significant at the 5\% level. The null hypothesis that average choice frequencies in all games are jointly equal can be rejected at the 1\% level.\footnote{For the joint test we use the Bonferroni correction. The null of simultaneous equality is rejected at a level of significance $\alpha$ by comparing the lowest p-value in all $m$ pairwise comparisons to $\frac{\alpha}{m}$ and rejecting whenever it is lower. This ensures that the family-wise error rate (FWER), which is the probability that at least one of the averages is not equal to the others, remains at the desired level $\alpha$. This is the most conservative test based on the FWER, in the sense that it rejects the null less often.} This remains true even when excluding game 4 from the test, as it is the single most different game.
\item[(ii)] Again, one can see in the top right graph of figure \ref{expdataAMP} that there is little overlap of the confidence intervals for average beliefs for the five games. The results of Hotteling's t-tests for the comparison of means confirms what is observed in the graph. The right panel of table \ref{pvals_only_red} reports the p-values from the test. Average beliefs are statistically different across games for any reasonable significance level.
\item[(iii)] A careful inpsection of the two graphs in figure \ref{expdataAMP} reveals that for all five games, average choice frequencies lie far from the corresponding average reported beliefs. Again, using Hotelling's t-test we confirm this finding. The p-values from comparing actions to beliefs in each of the games are very close to zero in all five cases.
\item[(iv)] Figure \ref{only_red_cdfs} shows the empirical cumulative density functions (cdf's) for the average choice by Row and Column in each game. If actions were homogeneous, then these should look like the cdf  for the observed success rate in a binomial distribution with 8 trials and a true success rate in each trial equal to the average choice across all individuals. We use a Monte Carlo simulation with 1000 repetitions to produce this cdf and compare the two. Using a Kolmogorov-Smirnov test we can reject that the two are the same for all cases at the 1\% level for all games except game 2. For game 2 we can reject at the 5\% level for Row (\emph{p-val = .031}) but not for Column (\emph{p-val = .117}). This result is further supported by a Cochran's Q test comparing individual choices in each game. The test rejects those being the same for both Row and Column in all games for any level of significance.
\item[(v)] We use Friedman's test to compare reported beliefs across individuals in each game, separately for Row and Column. In all cases we find there is heterogeneity that is significant at any reasonable significance level.
\end{itemize}

\begin{table}[t]
\begin{center}
\begin{tabular}{ccccc}
\multicolumn{1}{r}{\textbf{actions}}& \multicolumn{1}{|c}{$2$} & \multicolumn{1}{c}{$3$}  & \multicolumn{1}{c}{$4$}  & \multicolumn{1}{c}{$5$}  \\ \cline{1-5}
\multicolumn{1}{r}{\rule{0pt}{4mm}$1$} & \multicolumn{1}{|c}{.2334} & \multicolumn{1}{c}{.0011}  & \multicolumn{1}{c}{.0045} & \multicolumn{1}{c}{.0863}\\
\multicolumn{1}{r}{\rule{0pt}{4mm}$2$} & \multicolumn{1}{|c}{} & \multicolumn{1}{c}{.1374}  & \multicolumn{1}{c}{.0004} & \multicolumn{1}{c}{.0052}\\
\multicolumn{1}{r}{\rule{0pt}{4mm}$3$} & \multicolumn{1}{|c}{} & \multicolumn{1}{c}{}  & \multicolumn{1}{c}{.0000} & \multicolumn{1}{c}{.0000}\\
\multicolumn{1}{r}{\rule{0pt}{4mm}$4$} & \multicolumn{1}{|c}{} & \multicolumn{1}{c}{}  & \multicolumn{1}{c}{} & \multicolumn{1}{c}{.0491}\\
\end{tabular}
\hspace*{1cm}
\begin{tabular}{ccccc}
\multicolumn{1}{r}{\textbf{beliefs}}& \multicolumn{1}{|c}{$2$} & \multicolumn{1}{c}{$3$}  & \multicolumn{1}{c}{$4$}  & \multicolumn{1}{c}{$5$}  \\ \cline{1-5}
\multicolumn{1}{r}{\rule{0pt}{4mm}$1$} & \multicolumn{1}{|c}{.0000} & \multicolumn{1}{c}{.0000}  & \multicolumn{1}{c}{.0045} & \multicolumn{1}{c}{.6374}\\
\multicolumn{1}{r}{\rule{0pt}{4mm}$2$} & \multicolumn{1}{|c}{} & \multicolumn{1}{c}{.0489}  & \multicolumn{1}{c}{.0974} & \multicolumn{1}{c}{.0000}\\
\multicolumn{1}{r}{\rule{0pt}{4mm}$3$} & \multicolumn{1}{|c}{} & \multicolumn{1}{c}{}  & \multicolumn{1}{c}{.0001} & \multicolumn{1}{c}{.0000}\\
\multicolumn{1}{r}{\rule{0pt}{4mm}$4$} & \multicolumn{1}{|c}{} & \multicolumn{1}{c}{}  & \multicolumn{1}{c}{} & \multicolumn{1}{c}{.0002}\\
\end{tabular}
\vspace*{0mm}\\
\caption{Experiment 1: p-values from pairwise comparisons between games.}\label{pvals_only_red}
\end{center}
\vspace*{-6mm}
\end{table}

\begin{figure}[h]
\begin{center}
\includegraphics[scale=0.4]{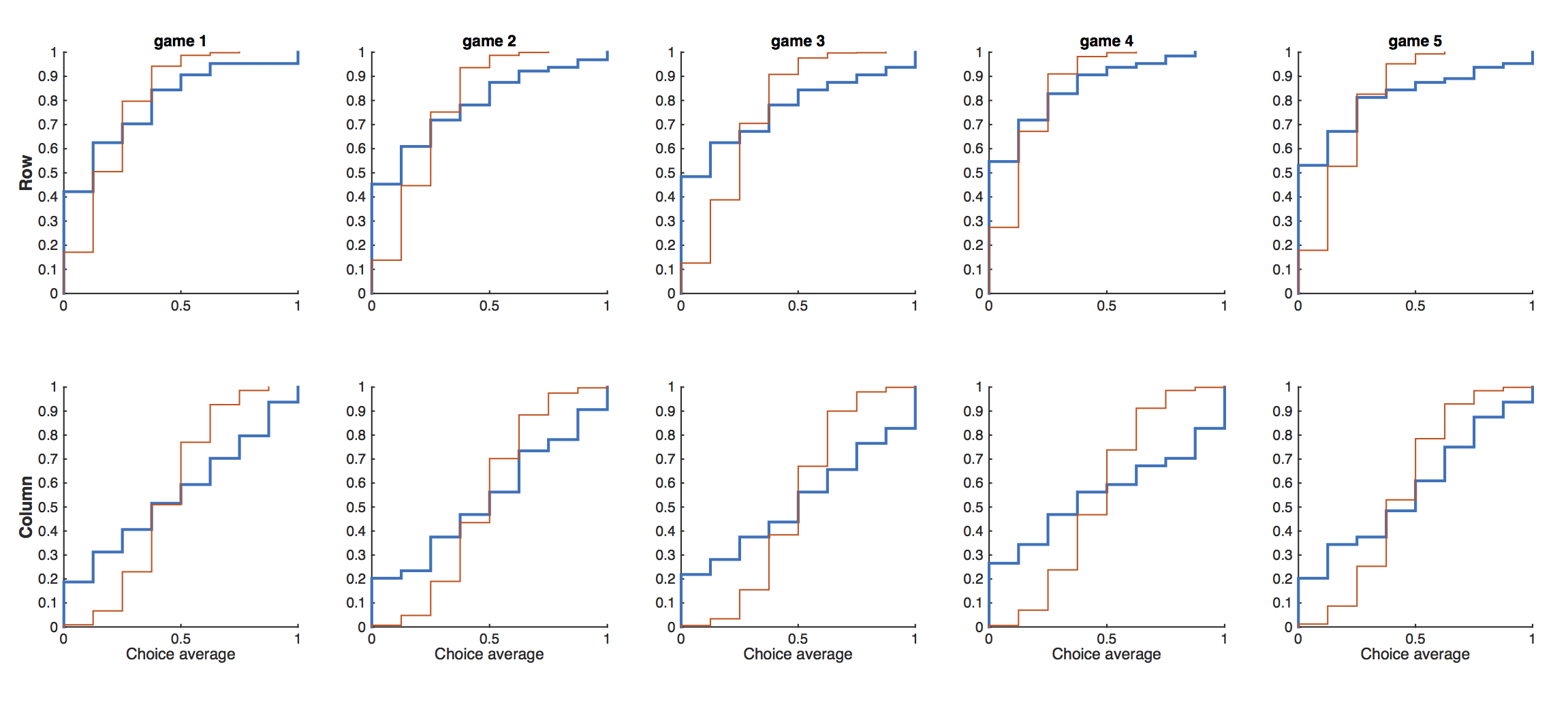}
\vspace*{-3mm}
\caption{Heterogeneity of actions in Experiment 1. Empirical cdf's for subjects' actions (blue) and cdf's of the observed success rate in MC simulations of binomial distributions with 8 trials and true success rate equal to subjects' overall average choice (orange).}\label{only_red_cdfs}
\end{center}
\vspace*{-8mm}
\end{figure}

\medskip

\newappendix{Cluster Analysis with the $k$-Means Algorithm.}

As mentioned in the main text, we classify our data of elicited beliefs into different clusters. This is done in a theory-free way using the $k$-means algorithm (MacQueen, 1967).
Given $k$, the number of clusters one wants to use, $k$ random `centroids' are chosen, which are points in the same space as that of the data. In the next step, each observation is matched to its closest centroid (in terms of Euclidean distance). Next, new centroids are calculated by taking the mean of the observations in each cluster. The algorithm proceeds iteratively until it converges to a stable set of clusters.

While convergence is guaranteed, the resulting clusters will depend on the random initialisation of the algorithm. To address the issue we take the standard approach which is to repeat the analysis for a large number of times (5000 in our case) and choose the outcome with the smallest sum of errors (distance of each observation from its cluster centroid).

The final issue to address is the number of clusters, $k$. We approach the problem using what is dubbed the `elbow method'. We run the analysis for $k\in\{2,...,15\}$ and calculate the sum of errors in each case. This is then plotted against $k$, giving a convex, decreasing curve. For each game we choose the $k$ that is at the `elbow' of the curve. Figure~\ref{kmeans} shows the curves for the case of the two DS games. In the analysis we use $k=7$ for these games.

\begin{figure}[h]
\begin{center}
\vspace*{-2mm}
\includegraphics[scale=0.56]{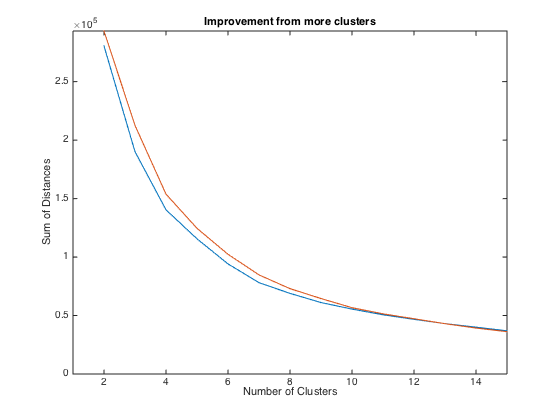}
\vspace*{-3mm}
\caption{K means in the DS games. The curves show the sum of distances for different values of $k$ for $DS1$ (blue line) and $DS2$ (red line).}\label{kmeans}
\end{center}
\vspace*{-5mm}
\end{figure}

\newpage
\newappendix{Derivation of $\mu$ Equilibria}

In this appendix, we provide explicit formulae for the $\mu$ equilibria of the $3\times 3$ games in Table \ref{3by3exp} using $\mu_i(\varepsilon)=(1,\varepsilon,\varepsilon^2,\ldots,\varepsilon^{K_i})$ (appropriately normalized) for $\varepsilon\in[0,1]$ and $i\in N$.

For the DS1 and DS2 games, the $\mu$ equilibrium is given by
\begin{displaymath}
  \sigma(\varepsilon)\,=\,\frac{1}{(1+\varepsilon+\varepsilon^2)}\left\{\begin{array}{ccc}
  (\varepsilon,1,\varepsilon^2) & \text{if} & 0.146\,\leq\,\varepsilon\,\leq\,1 \\[1mm]
  (\deel{1}{8}(1+\varepsilon+\varepsilon^2),1,\deel{1}{8}(-1+7\varepsilon+7\varepsilon^2) & \text{if} & 0.146\,\leq\,\varepsilon\,\leq\,0.456\\[1mm]
  (\varepsilon^2,1,\varepsilon) & \text{if} & 0.333\,\leq\,\varepsilon\,\leq\,0.456 \\[1mm]
  (\varepsilon^2,2+2\varepsilon-15\varepsilon^2,-1-\varepsilon+15\varepsilon^2) & \text{if} & 0.333\,\leq\,\varepsilon\,\leq\,0.4 \\[1mm]
  (\varepsilon^2,\varepsilon,1) & \text{if} & 0.354\,\leq\,\varepsilon\,\leq\,0.4\\[1mm]
  (\deel{1}{8},-\deel{1}{8}+\varepsilon+\varepsilon^2,1) & \text{if} & 0.125\,\leq\,\varepsilon\,\leq\,0.354 \\[1mm]
  (\varepsilon,\varepsilon^2,1) & \text{if} & 0\,\leq\,\varepsilon\,\leq\,0.125
  \end{array}\right.
\end{displaymath}
see the colored curve in the upper-left panel of Figure \ref{epsEqgraphs}, which also shows the greyed out $M$ sets. The bottom-left panel shows the logit-QRE curve in black. The latter is unique for every value of $\lambda$. Note, however, that multiple $\mu$ equilibria may exist for certain ranges of $\varepsilon$.

For the ``no logit'' game, the $\mu$ equilibrium is given by
\begin{displaymath}
  \sigma(\varepsilon)\,=\,\frac{1}{(1+\varepsilon+\varepsilon^2)}\left\{\begin{array}{ccc}
  (\varepsilon,1,\varepsilon^2) & \text{if} & 0.053\,\leq\,\varepsilon\,\leq\,0.947 \\[1mm]
  (\deel{1}{40}+\hf\varepsilon+\hf\varepsilon^2,1,-\deel{1}{40}+\hf\varepsilon+\hf\varepsilon^2) & \text{if} & 0.053\,\leq\,\varepsilon\,\leq\,0.947\\[1mm]
  (\varepsilon^2,1,\varepsilon) & \text{if} & 0.252\,\leq\,\varepsilon\,\leq\,1 \\[1mm]
  (\varepsilon^2,\deel{1}{3}(10+10\varepsilon-150\varepsilon^2),\deel{1}{3}(-7-7\varepsilon+150\varepsilon^2)) & \text{if} & 0.252\,\leq\,\varepsilon\,\leq\,0.283 \\[1mm]
  (\varepsilon^2,\varepsilon,1) & \text{if} & 0.178\,\leq\,\varepsilon\,\leq\,0.283\\[1mm]
  (\deel{1}{45}(1+2\varepsilon+2\varepsilon^2),\deel{1}{45}(-1+43\varepsilon+43\varepsilon^2),1) & \text{if} & 0.024\,\leq\,\varepsilon\,\leq\,0.178 \\[1mm]
  (\varepsilon,\varepsilon^2,1) & \text{if} & 0\,\leq\,\varepsilon\,\leq\,0.024
  \end{array}\right.
\end{displaymath}
see the colored curve in the upper-middle panel of Figure \ref{epsEqgraphs}. The colored curve does enter the lower $M$-choice set unlike the black logit-QRE curve shown in the bottom-middle panel.

Finally for the KM game, the $\mu$ equilibrium is given by
\begin{displaymath}
  \sigma(\varepsilon)\,=\,\frac{(1,\varepsilon^2,\varepsilon)}{(1+\varepsilon+\varepsilon^2)}
\end{displaymath}
for $\varepsilon\in[0,1]$, see the orange curve in the upper-right panel of Figure \ref{epsEqgraphs}, which is similar to the black logit-QRE curve shown in the bottom-right panel.

Except for the KM game, the above formulae may appear somewhat involved. However, they are relatively easy to determine. First, we check for ``pure'' strategies on each of the six parts of the barycentric division of the simplex, see the middle-left panel of Figure \ref{mondrianGame}. For instance, for the part with $\sigma_R>\sigma_B>\sigma_Y$, we check if the expected payoffs based on the belief $\omega=(1,\varepsilon,\varepsilon^2)/(1+\varepsilon+\varepsilon^2)$ are ranked $\pi_R>\pi_B>\pi_Y$. For the DS and ``no logit'' games, this yields the first, third, fifth, and seventh cases in the formulae above. The remaining (second, fourth, and sixth) cases simply follow by connecting the pure strategy cases across the different parts of the barycentric division using the payoff indifference curves, see the top panels of Figure \ref{epsEqgraphs}.

\newpage
\newappendix{Instructions for the Experiment}

\emph{The following is the translation of the instructions used in the experiment. The original text in Greek is available from the authors upon request. }

Thank you for participating in today’s session. Please remain quiet. The entire experiment will be conducted through your computer terminal and you will only interact with other participants via the computer. Please do not talk or make other noises during the experiment. The use of mobile phones and similar devices is not allowed.
\paragraph{General instructions:} During the experiment you can earn points. At the end of the experiment you will receive 1 euro for every 250 points you earned. The amount earned may be different for each participant and depends on his/her decisions, the decisions of other participants and luck. In addition, you will receive 5 euros as a show-up fee. Payment in cash will take place privately, immediately after the experiment.
\paragraph{The games:} The experiment consists of a series of different games. You will play each game for 15 or 8 periods (this will be announced beforehand). In each period you will be randomly matched with another participant. The participant you are matched with will be different in each period. On your screen you will see a table like the one below:
\begin{center}
\includegraphics[scale=.8]{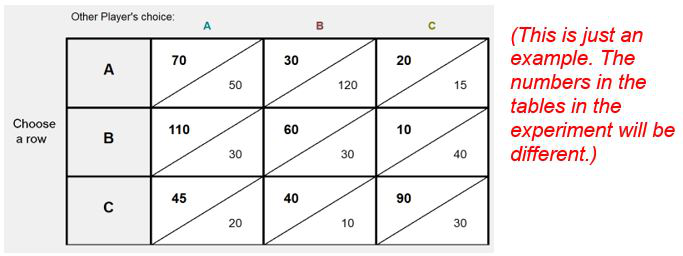}
\end{center}
Each player selects A, B or C by clicking the corresponding row. Each player wins the points indicated in the cell that corresponds to these choices. In the example above, if you choose C and the other player chooses B, you win 40 points and the other player wins 10 points.

Notice that the other player, exactly like yourself, sees the table with the points he/she may win depending on both of your choices, as well as the points you may win.
\paragraph{Belief task:} In each period you are also asked to state your beliefs about the percentage chance that the other player chooses A, B or C.

On the screen, above the game you will see three bars like in the picture below. You can indicate what you believe to be the percentage chance that the other player chooses A, B and C by clicking and dragging the corresponding bar. The sum must always be equal to 100\%.
\begin{center}
\includegraphics[scale=.73]{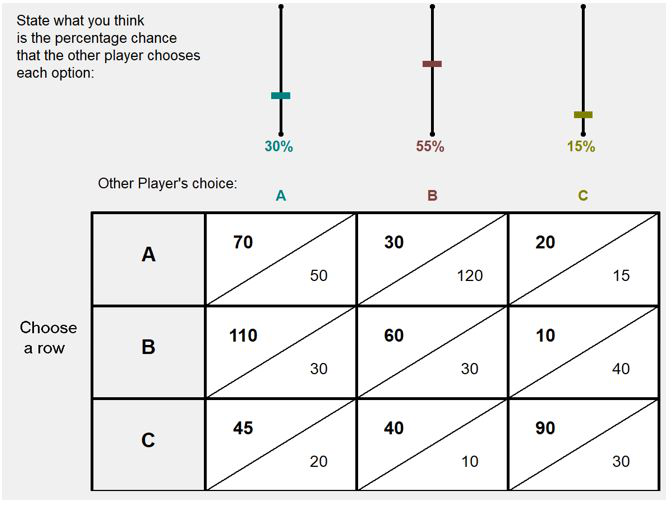}
\end{center}
Each time, the computer chooses randomly one of the three bars and on that bar it chooses randomly two points. If your stated belief in that bar is closer to the other player's actual choice than at least one of these two points, then you may win 60 points.
\begin{center}
\includegraphics[scale=.73]{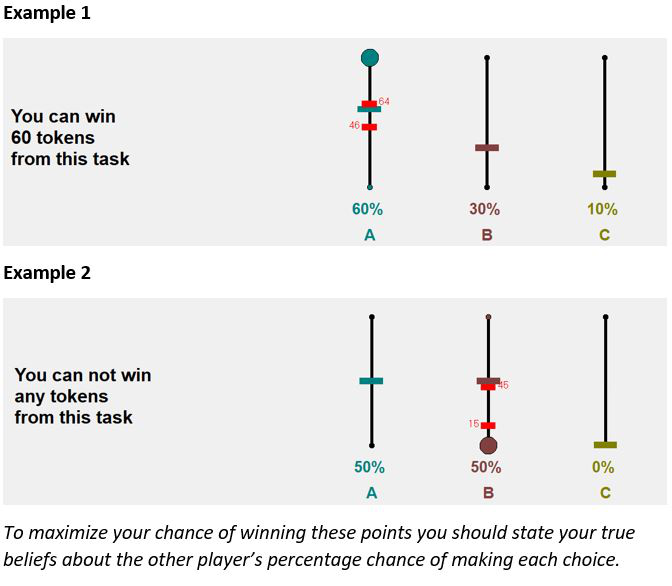}
\end{center}
\textbf{Notice:}
\begin{itemize}
\item You may change your choice as many times as you like (both for the game and the beliefs) When you are ready, press ``Submit.''
\item For the first period of the game you have to wait 30 seconds before you are able to submit your choices. Use that time to familiarize yourself with the new game table.
\item Please submit your choices within a reasonable time period. All participants will need to submit their choices before the experiment can move on to the next period.
\item After submitting your choices, you will be informed about the other player's choice and your possible point earnings from the game. You will also be informed about the bar and points selected by the computer and whether you may earn points from the belief task.
\item To determine your earnings in each period, the computer will randomly choose either the game or the belief task (both with equal probability) and you will earn the points you were awarded in the corresponding task.
\end{itemize}
\begin{center}
\includegraphics[scale=.73]{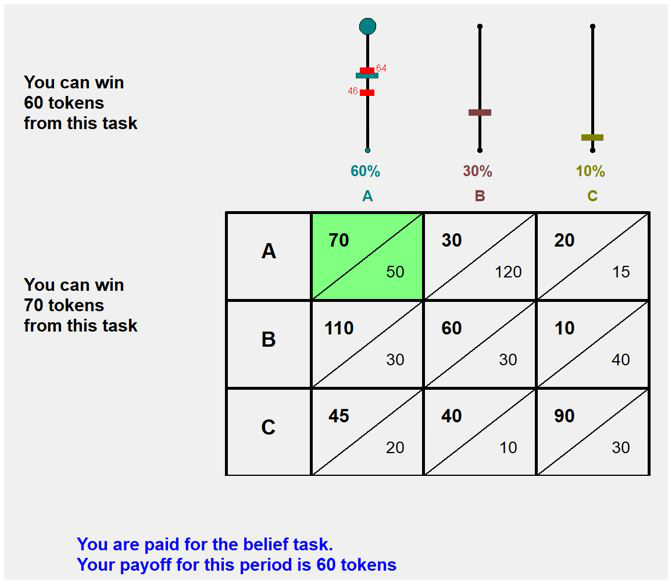}
\end{center}
\end{document}